\documentclass[10pt,english,draftclsnofoot,onecolumn]{IEEEtran}
\usepackage[T1]{fontenc}
\usepackage[latin9]{inputenc}
\usepackage{amstext}
\usepackage{amssymb}
\usepackage{amsmath}
\usepackage{amsthm}
\usepackage{graphicx}
\usepackage[font={small,it}]{caption}
\usepackage{subfig}
\usepackage{verbatim} % for use of begin{comment}
\usepackage{cite}
\usepackage{url}
\usepackage[lined,boxed,commentsnumbered]{algorithm2e}
\usepackage{array}
\newcolumntype{L}[1]{>{\raggedright\let\newline\\\arraybackslash\hspace{0pt}}m{#1}}
\newcolumntype{C}[1]{>{\centering\let\newline\\\arraybackslash\hspace{0pt}}m{#1}}
\newcolumntype{R}[1]{>{\raggedleft\let\newline\\\arraybackslash\hspace{0pt}}m{#1}}
\theoremstyle{definition}

\makeatletter
\def\namedlabel#1#2{\begingroup
    #2%
    \def\@currentlabel{#2}%
    \phantomsection\label{#1}\endgroup
}
\makeatother
\usepackage{myPack}\usepackage{myMathPack}\usepackage{Steven}

\cset{\Integers}{Z}
\cset{\Reals}{R}

\cset{\Matroids}{M}

\vect{\matroidRankVec}{r}
\scalar{\matroidRankFunc}{r}

\newcommand{\sinkDefSet}{\mathcal{W}}
\newcommand{\edgeDefSet}{\mathcal{Q}}
\newcommand{\orbit}{\mathcal{O}}
\newcommand{\symmGroup}{\boldsymbol{\mathsf{S}}}
\newcommand{\groupG}{\boldsymbol{\mathsf{G}}}

\makeatother
\newtheorem{remark}{Remark}
\newtheorem{definition}{Definition}
\newtheorem{theorem}{Theorem}
\newtheorem{corollary}{Corollary}
\newtheorem{example}{Example}
\begin{document}

\title{On Multi-source Networks: Enumeration, Rate Region Computation, and Hierarchy}

\author{%
Congduan Li,~\IEEEmembership{Student~Member, IEEE,} %
Steven Weber,~\IEEEmembership{Senior~Member, IEEE,} and %

John MacLaren Walsh, ~\IEEEmembership{Member, IEEE.}
\thanks{%
Support under National Science Foundation awards CCF--1016588 and 1421828 is gratefully acknowledged.
}%
\thanks{%
C.~Li, S.~Weber and J.~W.~Walsh are with the Department of Electrical and Computer Engineering, Drexel University, Philadelphia, PA USA (email: \textsf{congduan.li@drexel.edu}, \textsf{sweber@coe.drexel.edu}, and \textsf{jwalsh@coe.drexel.edu}).
Preliminary results were presented at Allerton 2014 \cite{CongduanAllerton2014}, NetCod 2015 \cite{CongduanNetCod2015}, and {ITW 2015 \cite{CongduanITW2015}}.
}%
}

\maketitle
\begin{abstract}
%abstract required to be in passive voice
This paper investigates the enumeration, rate region computation, and hierarchy of general multi-source multi-sink hyperedge networks under network coding, which includes multiple network models, such as independent distributed storage systems and index coding problems, as special cases.  A notion of minimal networks and a notion of network equivalence under group action are defined.  An efficient algorithm capable of directly listing single minimal canonical representatives from each network equivalence class is presented and utilized to list all minimal canonical networks with up to 5 sources and hyperedges.  Computational tools are then applied to obtain the rate regions of all of these canonical networks, providing exact expressions for 744,119 newly solved network coding rate regions corresponding to more than 2 trillion isomorphic network coding problems.  In order to better understand and analyze the huge repository of rate regions through hierarchy, several embedding and combination operations are defined so that the rate region of the network after operation can be derived from the rate regions of networks involved in the operation.  The embedding operations enable the definition and determination of a list of forbidden network minors for the sufficiency of classes of linear codes.  The combination operations enable the rate regions of some larger networks to be obtained as the combination of the rate regions of smaller networks.  The integration of both the combinations and embedding operators is then shown to enable the calculation of rate regions for many networks not reachable via combination operations alone.
\end{abstract}
\section{Introduction}

Many important practical problems, including efficient information transfer over networks \cite{NetworkInfoFlow2000,DFZMatroidNetworks}, the design of efficient distributed information storage systems \cite{DimakisTranIT2010,Tian433Journalversion}, and the design of streaming media systems \cite{WalshWeberTranIT2009,CISS2012Paper,HoISIT2013Streaming}, have been shown to involve determining the rate region of an abstracted network under network coding.
Yan \emph{et al.}'s celebrated paper \cite{YanYeungTranIT2012} has provided an exact representation of these rate regions of networks under network coding.  Their essential result is that the rate region of a network can be expressed as the intersection of the region of entropic vectors \cite{ZhangYeungTranIT1998Entropy,YeungBook} with a series of linear (in)equality constraints created by the network's topology and the sink-source requirements, followed by a projection of the result onto the entropies of the sources and edge variables.  However, this is only an implicit description of the rate region, because the region
of entropic vectors $\bar{\Gamma}_N^*$ is still unknown for $N\geq 4$.

Nevertheless, as we have previously demonstrated in \cite{CongduanAllerton2012,CongduanNetCod2013,CongduanTranIT2014}, through the use of appropriate inner and outer bounds to $\bar{\Gamma}_N^*$ that we will review in \S\ref{sec:bounds}, this implicit formulation can be used to develop algorithms by which a computer can very rapidly calculate the rate region, its proof, and the class of capacity achieving codes, for small networks, each of which would previously have taken a trained information theorist hours or longer to derive.   While the development of this rate region calculation, code selection, and converse proof generation algorithm \cite{CongduanTranIT2014} is not the focus of the present paper, it involves developing techniques to derive and project polyhedral inner and outer bound descriptions for constrained regions of entropic vectors.  When it comes to rate regions, \S \ref{sec:resultssmall} and \S \ref{sec:bounds} of the paper will focus more on exactly what was calculated, what can be calculated, and, later in the paper, what can be learned from the resulting rate regions, rather than the exact computations by which the rate regions were reached.  The rate region algorithm design and specialization, which involve a separate and parallel line of investigation, are left to discussion by another series of papers \cite{JayantISIT2014,JayantNetCod2015,JayantISIT2015}.

The ability to calculate network coding rate regions for small networks rapidly with a computer motivates an alternative, more computationally thinking oriented, agenda to the study of network coding rate regions.  At the beginning, one's goal is to demonstrate the method's power by applying the algorithm to derive the rate region of as many networks and applications of network coding as possible.    To do this, in \S\ref{sec:model} we first slightly generalize Yeung's labeled directed acyclic graph (DAG)  model for a network coding problem (\cite{YeungBook}, Ch. 21) to a directed acyclic hypergraph context, then demonstrate how the enlarged model handles as special cases the wide variety of other models in applications in which network coding is being employed, 
including, but not limited to, index coding, multilevel diversity coding systems (MDCS),  and distributed storage. 

With the slightly more general model in hand, the first issue in the computationally thinking oriented agenda is \emph{network generation and enumeration}, i.e., how to list all of the networks falling in this model class.  In order to avoid repetitive work, thereby reaching the largest number of networks possible with a constant amount of computation, it is desirable to understand precisely when two instances of this model (i.e., two network coding problems) are equivalent to one another, in the sense that the solution to one directly provides a solution to another.  This notion of network coding problem equivalence, which provides a very different approach but is in the same high level spirit as the transformation of network channel coding problems to network coding problems in \cite{KoetterEffrosMedardNetworkEquivalenceII2013} and the transformation of network coding problems to index coding problems in \cite{EffrosRouayhebLangbergTranIT2015NetEqu}, will be revisited at multiple points of the paper, beginning in this present context of enumeration, but also playing an important role in the discussion of hierarchy later.  

The first notion of equivalence we develop is that of \emph{minimality}, by removing any redundant or unnecessary parts in the network instance.  Partially owing to the generality of the network coding problem model, many valid instances of it include within a network parts which can be immediately detected as extraneous to the determination of the instance's rate region.  In this sense, an instance is directly reducible to another smaller instance by removing completely unnecessary and unhelpful sources, nodes, or edges.  In order to provide the smallest possible instance by not including these extraneous components, we formalize in \S \ref{sec:minimality} the notion of \emph{network minimality}, listing a series of conditions which a network coding problem description must obey to not contain any obviously extraneous sources, nodes, or edges.

The next notion of {\it equivalence} looks to symmetry or isomorphism between problem descriptions.  Beginning by observing that a network must be labeled to specify its graph, source availability, and demands to the computer, yet the underlying network coding problem is insensitive to the selection of these labels, we define in \S \ref{sec:netEquivGroupAct} a notion of network coding problem equivalence through the isomorphism associated with the selection of these labels.  We review that the proper way to formalize this notion is through identifying equivalence classes with orbits under group actions.  A na\"{i}ve algorithm to provide the list of network coding problems to the rate region computation software would simply list all possible labeled network coding problems, test for isomorphism, then narrow the list down to only those which are inequivalent to one another, keeping only one element, the canonical representative, of the equivalence class.  However, the key reason for formalizing this notion of equivalence is that the number of labeled network coding problem instances explodes far faster than the number of network coding problem equivalence classes.  Hence, we develop a better technique for generating lists of canonical network coding problem instances by harnessing techniques and algorithms from computational group theory that enable us to directly list the minimal canonical representatives of the network coding problem equivalence classes as described in \S \ref{net:nonisoalg}.

With the list of all minimal canonical network coding problems up to a certain size in hand, we can utilize our algorithm and software to calculate the rate region bounds, the Pareto optimal network codes, and the converse proofs, for each, building a very large database of rate regions of network coding problems up to this size.  Owing to the variety of the model, even for tiny problems, this database quickly grows very large relative to what a human would want to read through.  For instance, our previous paper applying this computational agenda to the narrower class of MDCS problems \cite{CongduanTranIT2014}, yielded the rate regions of 6,868 equivalence classes of MDCS problems and bounds for 492 more MDCS problem equivalence classes, while the database developed in this paper contains the rate regions of 744,119
% SW: CL, please fill in this value
% (2,3) not closed = 53018 non iso, 625281 with edge iso, 1774574964884 with node iso
% (3,2) not closed = 37057 non iso, 435973 with edge iso, 60551477676 with node iso
%  So number closed is 
% 744119 - 53018 - 37057 = 635481 non iso
% 8619064 - 625281 - 435973 = 7557810 with edge iso
% 2381012415004 - 1774574964884 - 60551477676 = 545 885 972 444 with node iso
%  And number open is
%  53018 + 37057 =  90075 non iso
%  625281 + 435973 =  1 061 254 with edge iso 
% 1774574964884 + 60551477676 =  1 835 126 442 560 with node iso
equivalence classes of network coding problems.  These equivalence classes of networks correspond to solutions for 9,050,490 network coding problems with graphs specified via edge dependences and 2,381,624,632,119 network coding problems specified in the typical node representation of a  graph.  While it is possible to use the database to report statistics regarding the sufficiency of certain classes of codes as will be done in \S \ref{sec:resultssmall}, in order to more meaningfully enable humans to learn from the database, as well as from the computational research, one must utilize some notion of network structure to organize it for analysis.  

Our method of endowing structure on the set of network coding problems is through \emph{hierarchy}, in which we explain the properties and/or rate regions of larger networks as being inherited from smaller networks (or vice-versa).  Of course, part of a network coding problem is the network graph, and further, network coding and entropy is related to matroids, and these nearby fields of graph theory and matroid theory have both undergone a thorough study of hierarchy which directly inspires our approach to it.  In graph theory, this notion of hierarchy is achieved by recognizing smaller graphs within large graphs which can be created by deleting or contracting the larger graph's edges, called \emph{minors}, and is directly associated with a crowning achievement.  Namely, the celebrated well-quasi-ordering result of graph theory \cite{RobertsonSeymour1983,RobertsonSeymour2004}, showed that any minor closed family of graphs (i.e., ones for which any minor of a graph in that family is also in the family) has at most a finite list of forbidden minors, which no graphs in that family can contain.  While the families of minor closed graphs are typically infinite, they are then, in a sense, capable of being studied through a finite object which is their forbidden minors: if a graph does not have one of these forbidden minors, it is then in the family.
In matroid theory, which in a certain sense extends graph theory, one has a similar notion of hierarchy endowed through matroid minors, generated through matroid contraction and deletion.  While one can generate minor closed families of matroids with an infinite series of forbidden minors, the celebrated, and possibly recently proved, Rota's conjecture \cite{OxleyMatroidBook,solvingrota}, stated that those matroids capable of being represented over a particular finite field have at most a finite list of forbidden minors.  In this paper, inspired by these hierarchy theories in graphs and matroids, we aim to derive a notion of network minors created from a series of contraction and deletion-like operators that shrink network coding problems, called \emph{embedding operators}, as well as operations for building larger network coding problems from smaller ones, called \emph{combination operators}.  These operators work together to build our notion of minors and a sense of hierarchy among network coding problems.
%In our previous work with MDCS problems, this was achieved with the notion of embedding operators, which explain certain properties of larger networks, such as the insufficiency of classes of codes, in terms of their containment with smaller problems

Developing a notion of network coding problem hierarchy is important for several reasons.  First of all, as explained above, even after one has calculated the rate regions of all networks up to a certain size, it is of interest to make sense of this very large quantity of information by studying its structure, and hierarchy is one way of creating a notion of structure.  Second of all, the computational techniques for proving network rate regions can only handle networks with tens of ``variables'', the sum of the number of sources and number of hyperedges in the graph, and hence are limited to direct computation of fairly small problem instances.  If one wants to be able to utilize the information gathered about these small networks to understand the rate regions of networks at scale, one needs methods for putting the smaller networks together into larger networks in a way such that the rate region of the larger network can be directly calculated from those of the smaller networks.

Our embedding operators, defined and discussed in \S\ref{sec:embedding}, extend the series of embedding operations we had for MDCS problems in \cite{CongduanTranIT2014}, and augment them, to provide methods for obtaining small networks from big networks in such a way that the rate region of the smaller network, and its properties, are directly inherited from the larger network.  Our combinations operators, discussed in \S\ref{sec:combination}, work in the opposite direction: they provide methods for putting together smaller networks to make larger networks in such a way that the rate region of the larger network can be directly calculated from the rate region of the smaller networks.  Both of our lists of operators are small and somewhat simple, however, when they work together, they provide a very powerful way of endowing hierarchical structure in network coding problems.  In particular, the joint use of the combination and embedding operators provide a very powerful way of obtaining rate regions of large networks from small ones, as well as describing the properties of families of network coding problems, as we demonstrate in \S\ref{sec:resultsoperators}.  They open the door to many new avenues of network coding research, and we shall describe briefly some of the related future problems for investigation in \S \ref{sec:conclusion}.

\section{Background: Network Coding Problem Model, Capacity Region, and Bounds}\label{sec:model}
%\vspace{-0.2cm}
\begin{figure}
\centerline{\includegraphics[scale=0.5]{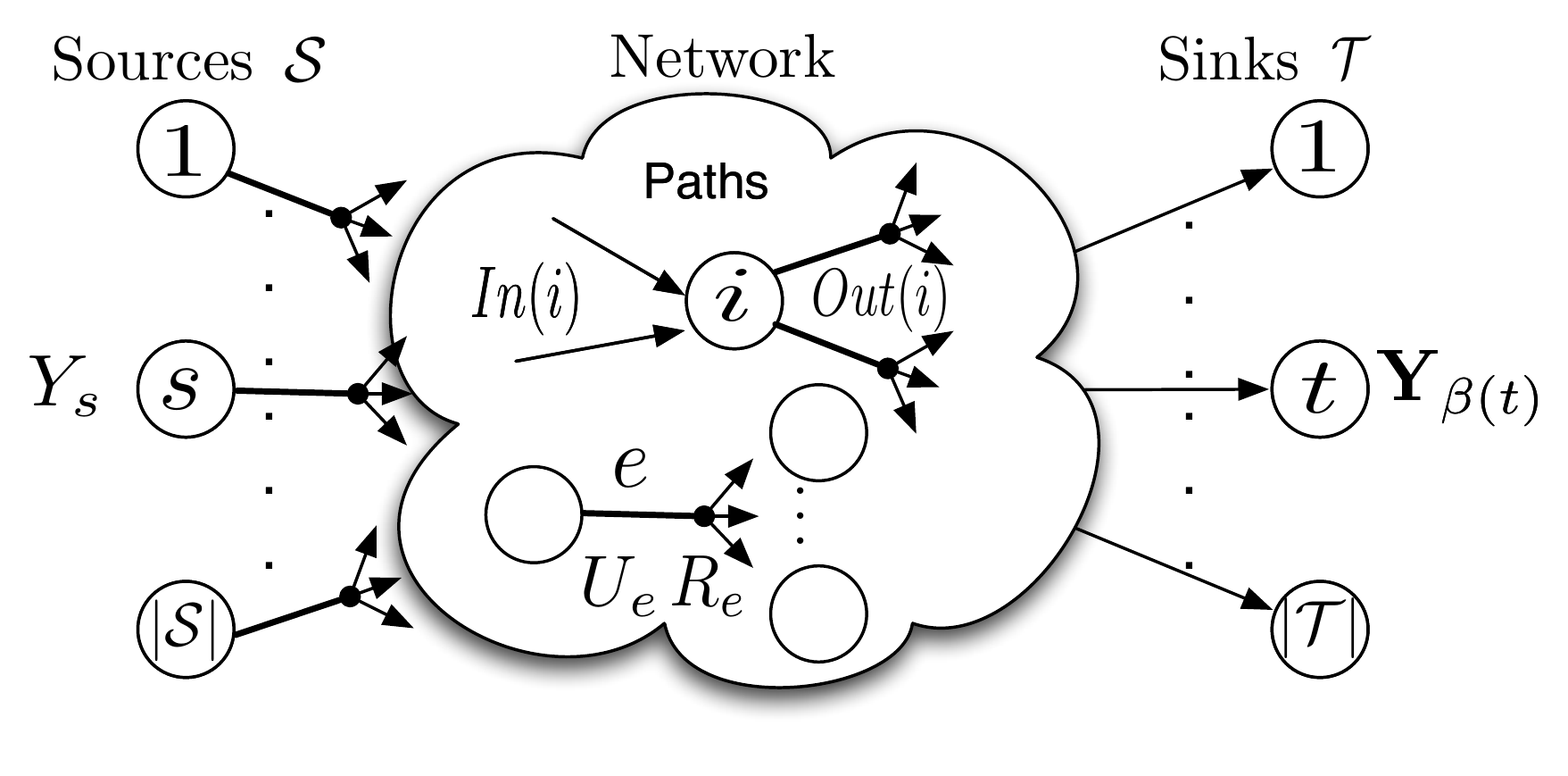}}
\caption{\label{fig:generalnetwork}A general network model $\Asf$.}
%\vspace{-0.4cm}
\end{figure}

\begin{table}
\caption{Notation table}\label{tab:notation}
\centering
\begin{tabular}{p{.35\textwidth} p{.6\textwidth}}
\hline
$\Asf$ & a network instance (\S\ref{sec:model})\\
$\Amc,\Bmc,\Cmc,\Dmc$ & general sets \\
$[x_a | a \in \Amc]$ & vector with elements $x_a$ indexed by/ for each $a \in \Amc$ \\
$\beta$ & demands of sink nodes (\S\ref{sec:model})\\
$\drm,d,f,\pi$ & mappings or functions \\
%$e:$ & an (hyper)edge or encoder\\
$e,\Emc$ & an (hyper)edge or encoder, set of all (hyper)edges or encoders (\S\ref{sec:model})\\
%$\Emc_S$ & set of outgoing (hyper)edges from sources\\
%$\Emc_U$ & set of all intermediate (hyper)edges\\
$\Fmc$ & head nodes of a hyperedge (\S\ref{sec:model})\\
$\Fbb_q$ & finite field of order $q$ (\S\ref{sec:rrexp})\\
%$g:$ & an intermediate node\\
$g,\Gmc$ & an intermediate node, set of intermediate nodes (\S\ref{sec:model})\\
$\groupG ,\symmGroup,\langle g_1,\cdots,g_k\rangle$ & acting group, symmetric group, group generated by $g_1,\ldots,g_k$ (\S\ref{net:probInstance})\\
%$H(\cdot)$ & entropy function (\S\ref{sec:model})\\
$H(\cdot), \hbf,h_{\Amc}$ & entropy function, an entropy vector, coordinate in $\hbf$ associated with $\Amc$ (\S\ref{sec:model})\\
%$h_{\Amc}:$ & coordinate in $\hbf$ associated with $\Amc$\\
$\rm{Hd}(e),\rm{Tl}(e)$ & head nodes, tail node of edge $e$ (\S\ref{sec:model})\\
$\rm{In}(g),\rm{Out}(g)$ & incoming, outgoing edges of node $g$ (\S\ref{sec:model})\\
$i,j,k,l$ & general index terms\\
$\Imc$ & independent set (\S\ref{sec:bounds})\\
$K,L,N$ & number of sources, intermediate (hyper)edges/encoders, and total variables in a network ($N=K+L$) (\S\ref{sec:model})\\
%$L:$ & number of intermediate (hyper)edges or encoders\\
$\Lmc_i,\Lmc_{\Asf}$ & sets associated with network constraints (\S\ref{sec:rrexp})\\
$\Msf,\Mmc$ & a matroid, ground set of a matroid (\S\ref{sec:bounds})\\
%$\Mmc$ & ground set of a matroid\\
$M$ & increased dimension of space in calculating rate region (\S\ref{sec:rrexp})\\
%$\begin{array}{l}
$\mathrm{minimal}(\Asf'),\mathrm{minimal}_{\Asf' \rightarrow \Asf}(\Rmc_*(\Asf'))$ & function to reduce a network $\Asf'$ to its minimal representation and associated rate region operators (\S\ref{sec:minimality})\\
$\Nmc$ & collection of variables in a network (\S\ref{sec:rrexp})\\
$\Omc$ & collection of networks in an equivalence class under our definition (\S\ref{sec:netEquivGroupAct})\\
%$p:$ & error probability\\
$p,\Pbb$ & error probability, probability function (\S\ref{sec:rrexp})\\
$\Pmc_i(\Xmc)$ & all size $i$ subsets of a given set $\Xmc$ (\S\ref{net:nonisoalg})\\
$\Qmc,\Wmc:$ & edge encodings for a network topology and sink demands (\S\ref{net:probInstance})\\
$r_{\Msf}$ & rank function of a matroid $\Msf$ (\S\ref{sec:bounds})\\
$\Gamma_N^*,\bar{\Gamma}_N^*,\Gamma_N,\Gamma_N^q,\Gamma_{N,N'}^q,\Gamma_{N,\infty}^q, \Gamma_N^{{\rm linear}} $& region of entropic vectors on $N$ variables, its closure, Shannon outer bound, inner bound from $\Fbb_q$-representable matroids on $N$ elements, inner bound from $\Fbb_q$-representable matroids on $N'$ elements,inner bound from $\Fbb_q$-representable matroids on infinity number of  elements, inner bound from linear subspace arrangement (\S\ref{sec:bounds})\\
$R_e, \Rbf$ & edge capacity on edge $e$, rate vector (\S\ref{sec:rrexp})\\
%$\Rbf$ & rate vector (\S\ref{sec:rrexp})\\
$\boldsymbol{r},\boldsymbol{\omega}, {\rm Proj}_{\boldsymbol{r},\boldsymbol{\omega}}(\cdot)$ & dimensions associated with all edge capacities and all source entropies, projection operator with the projecting dimensions are associated with $\boldsymbol{r},\boldsymbol{\omega}$ (\S\ref{sec:rrexp})\\
$\Rmc_c(\Asf),\Rmc_*(\Asf),\Rmc_{o}(\Asf),\Rmc_{s,q},\Rmc_{q}^{N'},\Rmc_{q}, \Rmc_{{\rm linear}}$ & rate region or bounds associated with $\Gamma_N^*,\Gamma_N,\Gamma_N^q,\Gamma_{N,N'}^q,\Gamma_{N,\infty}^q,\Gamma_N^{{\rm linear}} $ (\S\ref{sec:rrexp},\S\ref{sec:bounds})\\
%$s:$ & a source node\\
$s,\Smc$ & a source node, set of all sources (\S\ref{sec:model})\\
%$\Sbf:$ & symmetry group\\
%$t:$ & a sink node\\
$t,\Tmc$ & a sink node, set of all sinks (\S\ref{sec:model})\\
$T$ & canonical representatives, i.e., transversal, output from Leiterspiel algorithm (\S\ref{net:nonisoalg})\\
$U,\Umc$ & edge variable, support set of $U$ (\S\ref{sec:rrexp})\\
$\Vmc$ & set of all nodes in a network (\S\ref{sec:model})\\
$V,\Vbf$ & a vector space, multiple vector spaces\\
$X,Y$ & random variables (\S\ref{sec:model})\\
$\Xbf,\Ybf$ & vectors of variables (\S\ref{sec:rrexp})\\
$\Xmc,\Ymc$ & general set, support set on $Y$ (\S\ref{sec:rrexp}, \S\ref{sec:enumeration})\\
$\Zmc$ & collection of network instances (\S\ref{net:nonisoalg})\\
\hline
\end{tabular}
\end{table}

The class of problems under study in this paper are the rate regions of multi-source multi-sink network coding problems with hyperedges, which we hereafter refer to as the hyperedge MSNC problems.  For ease of reading the paper, a notation table is presented in Table \ref{tab:notation}.  A network coding problem in this class, denoted by the symbol $\Asf$, includes a directed acyclic hypergraph $(\Vmc,\Emc)$ \cite{gallo1993directed} as in Fig.\,\ref{fig:generalnetwork}, consisting of a set of nodes $\Vmc$ and a set $\Emc$ of directed hyperedges in the form of ordered pairs $e=(v,\Amc)$ with $v\in \Vmc$ and $\Amc \subseteq \Vmc \setminus v$.  The nodes $\Vmc$ in the graph are partitioned into the set of source nodes $\Smc$, intermediate nodes $\Gmc$, and sink nodes $\Tmc$, i.e., $\Vmc=\Smc\cup\Gmc\cup\Tmc$.  Each of the source nodes $s\in\Smc$ will have a single outgoing edge $(s,\mathcal{A}) \in \Emc$.  The source nodes in $\Smc$ have no incoming edges, the sink nodes $\Tmc$ have no outgoing edges, and the intermediate nodes $\Gmc$ have both incoming and outgoing edges.  The number of sources will be denoted by $|\Smc|=K$, and each source node $s\in\Smc$ will be associated with an independent random variable $Y_s$, $s\in\Smc$, with entropy $H(Y_s)$, and an associated independent and identically distributed (IID) temporal sequence of random values.  For every source $s\in\Smc$, define ${\rm Out}(s)$ to be its single outgoing edge, which is connected to a subset of intermediate nodes and sink nodes.   A hyperedge $e\in\Emc$ connects a source, or an intermediate node to a subset of non-source nodes, i.e., $e=(i,\Fmc)$, where $i\in\Smc\cup\Gmc$ and $\Fmc\subseteq (\Gmc\cup\Tmc\setminus i)$.  For brevity, we will refer to hyperedges as edges if there is no confusion.  For an intermediate node $g\in\Gmc$, we denote its incoming edges as ${\rm In}(g)$ and outgoing edges
as ${\rm Out}(g)$.   For each edge $e= (i,\Fmc)$, the associated random variable 
$U_{e}=f_e({\rm In}(i))$ is a function of all the inputs of node $i$, obeying the edge capacity constraint $R_{e} \geq H(U_e)$.  The tail (head) node of edge $e$ is denoted as $\text{Tl}(e)$ ($\text{Hd}(e)$).  For notational simplicity, the unique outgoing edge of each source node will be the source random variable, $U_{e}=Y_s$ if $\text{Tl}(e)=s$, denoting $\Emc_S=\{e\in\Emc|\text{Tl}(e)=s,s\in\Smc\}$ to be the variables associated with outgoing edges of sources, and $\Emc_U=\Emc\setminus \Emc_S$ to be the non-source edge random variables.     For each sink $t\in\mathcal{T}$, 
%if we define $\sigma(t)=\{k\in\Smc|\exists \text{ a path from }k \text{ to } t \}$,  THIS HANDLED BY MINIMALITY
the collection of sources this sink will demand will be labeled by the non-empty set $\beta(t)\subseteq \Smc$.  Thus, a network can be represented as a tuple $\Asf=(\Smc,\Gmc,\Tmc,\Emc,\beta)$, where $\beta=(\beta(t),t\in\Tmc)$.  Note that, though this commonly used node-representation of a network is convenient for understanding the network topology, we will use a more concrete representation in \S\ref{sec:enumeration} for enumeration of network instances.  For convenience, networks with $K$ sources and $L=|\Emc_U|$ edges are referred as $(K,L)$ instances.

As our focus in the manuscript will be on rate regions of a very similar form to those in \cite{YeungBook,YanYeungTranIT2012},
this network coding problem model is as close to the original one in \cite{YeungBook,YanYeungTranIT2012} as possible while covering the multiple instances of applications in network coding in which the same message can be overheard by multiple parties.  These applications include index coding, wireless network coding, independent distributed source coding, and distributed storage.  The simplest and most direct model change to incorporate this capability is to switch from directed acyclic graphs to directed acyclic hypergraphs.  As we shall see in \S\ref{sec:rrexp}, this small change to the model is easily reconciled with the network coding rate region expression and its proof from \cite{YeungBook,YanYeungTranIT2012}.

\subsection{Special Network Classes}\label{sec:specialclasses}
The network coding problem model just described has been selected to be general enough to include a variety of models from the  applications of network coding as special cases.  A few of these special cases that will be of interest in examples later in the paper are reviewed here, including a description of the extra restrictions on the model to fall into this subclass of problems.

\begin{figure}
\centerline{\includegraphics[scale=0.6]{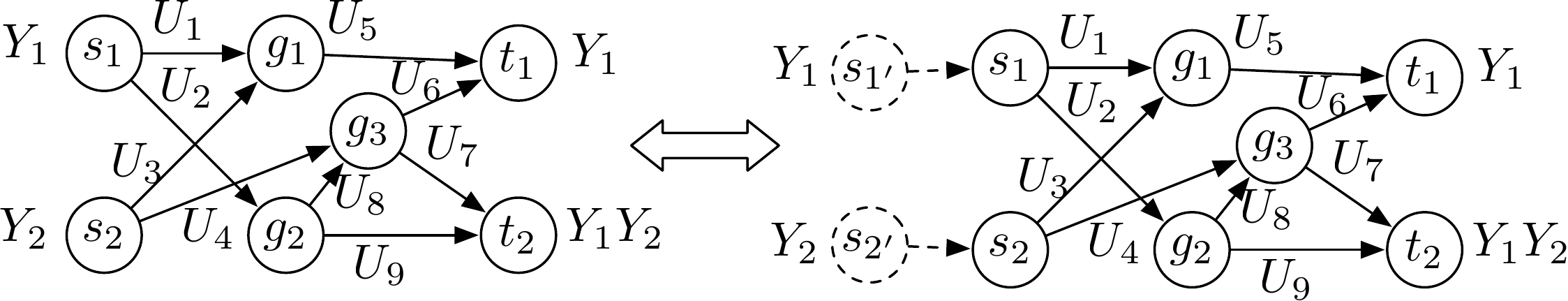}}
\caption{\label{fig:MSNC_general}A normal MSNC in \cite{YeungBook} can be viewed as a special instance of the hyperedge MSNC.}
%\vspace{-0.4cm}
\end{figure}
\begin{example}[Yan Yeung Zhang MSNC] The network model in \cite{YeungBook}, where the edges are not hyperedges and the outputs of a source node can be multiple functions of the source, can be viewed as a special class of networks in our model.  This is because the sources  can be viewed as intermediate nodes and a virtual source node connecting to each of them can be added.  For instance, a small network instance of the model in \cite{YeungBook}, as shown in Fig.\,\ref{fig:MSNC_general}, can be viewed as a hyperedge network instance introduced in this paper, by adding virtual source nodes $s_{1'},s_{2'}$ to sources $s_1,s_2$, respectively.
\end{example}

%On the other hand, the model in this paper can be converted to a normal directed acyclic graph as in \cite{YeungBook} with predefined identity coding functions on some edges, because every hyperedge can be replaced with some normal edges and one extra node, as shown in Fig.\,\ref{fig:hyperedgeconversion}, and the outgoing edges of the extra node are identical to the single incoming edge.

%\begin{figure}
%\centerline{\includegraphics[scale=0.6]{hyperedgeconversion}}
%\caption{\label{fig:hyperedgeconversion}Convert a hyperedge to normal edges with one extra relay node.}
%\vspace{-0.4cm}
%\end{figure}

\begin{example}[Independent Distributed Source Coding] \label{ex:IDSC}
The independent distributed source coding (IDSC) problems, which were motivated from satellite communication systems \cite{YeungZhang99SourceCoding}, can be viewed as a special class of networks in our model.  They are three-layer networks, where sources are connected with some intermediate nodes and those intermediate nodes will transmit coded messages to sinks.  For instance, the IDSC problem in Fig.\,\ref{fig:IDSC_general} can be converted to a hyperedge multi-source network coding problem.  As a special class of IDSC problems with decoding priorities among sources, the multi-level diversity coding systems \cite{HauThesis1995,CongduanTranIT2014} are naturally a class of networks in our current general model.  In the experimental results section \S\ref{sec:resultssmall}, we will not only show results on general hyperedge networks, but also some results on IDSC problems.
\end{example}
\begin{figure}
\centerline{\includegraphics[scale=0.6]{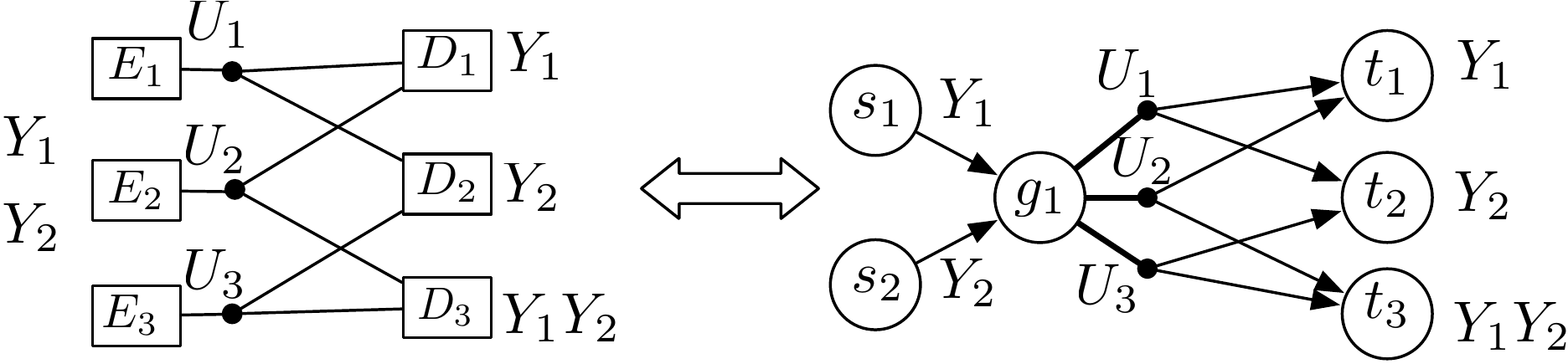}}
\caption{\label{fig:IDSC_general}An IDSC problem can be viewed as a hyperedge MSNC.}
\end{figure}

\begin{example}[Index Coding]
Since direct access to sources as side information is allowed in our network model,  index coding problems are also a special class of our model with only one intermediate edge.  That is, a $K$-source index coding problem can be viewed as a $(K,1)$ hyperedge MSNC and vice versa.  For instance, an index coding problem with $3$ sources, as shown in Fig.\,\ref{fig:indexcodingexample}, is a $(3,1)$ instance in our model. 
\end{example}

\begin{figure}
\centerline{\includegraphics[scale=0.6]{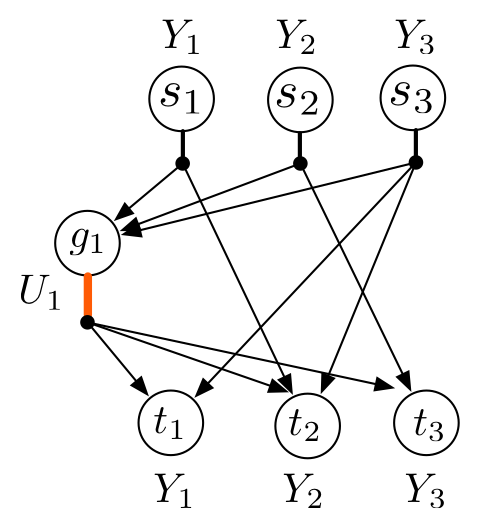}}
\caption{\label{fig:indexcodingexample}An index coding problem is a three-layer hyperedge MSNC with only one intermediate edge.}
\end{figure}

\subsection{Rate Region}\label{sec:rrexp}
Having defined a network coding problem, we now define a network code and the network coding rate region.

%Let $q\in\Nbb$ and $\Fbb_q$ denote the finite field used for messages.  

%Each source node $s\in\Smc$ is associated with a random variable $Y^_s$ distributed according to its independent random variable $Y_s$, taking values in its finite support set $\Ymc_s$.   
%Let $\Rbf=(\tau_1,\ldots,\tau_K,R_{1},\ldots,R_{L})\in\Rbb_{+}^{|\Emc|}$ be a vector of source and edge rates.  
\begin{definition}
An $(n,\Rbf)$ block code, with $\Rbf=[\tau_1,\ldots,\tau_K,R_{1},\ldots,R_{L}]\in\Rbb_{+}^{|\Emc|}$, consists of a series of mutually independent sources $Y_s^{(n)}$, uniformly distributed in $\Ymc_s=\{1,\ldots,\lceil 2^{n\tau_s}\rceil\}$, and block encoders and decoders.  

$i)$ The block encoders, one for each $e\in\Emc_U$, are functions that map a block of $n$ source observations from all sources in $\Emc_{\Smc} \cap \textrm{In}(\textrm{Tl}(e))$, and the incoming messages associated with the edges $\Emc_U \cap \textrm{In}(\textrm{Tl}(e))$, to one of $\lceil 2^{nR_{e}} \rceil$ different descriptions in $\Umc_e=\{0,1,\ldots,\eta_e-1\}$, where $\eta_e=\lceil 2^{nR_{e}} \rceil$,
\begin{equation}
f_{e}^{(n)}:\prod_{s\in\Emc_{\Smc}\cap{\rm In}({\rm Tl}(e))}\Ymc_{s}\times\prod_{i\in\Emc_U\cap{\rm In}({\rm Tl}(e))} \Umc_i \to\Umc_e,~e\in \Emc_U.
\end{equation}
$ii)$ The block decoders, one for each sink $t\in\Tmc$, are functions
\begin{equation} \label{eq:decfun}
d_{t}^{(n)}:\prod_{e\in {\rm In}(t)\cap\Emc_U}\Umc_e\times\prod_{s\in{\rm In}(t)\cap\Emc_S}\Ymc_{s}\to\prod_{s\in\beta(t)}\Ymc_{s},~t\in \Tmc.
\end{equation}
\end{definition}

Denote by $U_{e}^{(n)}\in \Umc_e$ the random message on edge $e\in \Emc_U$, which is the result of the encoding function $f_e^{(n)}$.

Further, we can define the
probability of error for each sink $t\in\Tmc$ as
\begin{equation}
p_{t}^{(n),{\rm err}}(\Rbf)=\Pbb\left[ d_{t}^{(n)}(U_{\textrm{In}(t)}^{(n)})\neq [Y_s^{(n)}|s\in\beta(t)] \right],
\end{equation}
and the maximum over these as 
\begin{equation}
p^{(n),{\rm err}}(\Rbf)=\max_{t\in \Tmc} p_{t}^{(n),err}.
\end{equation}

%Then, we can define the rate region of a network.

\begin{definition}
The rate region of a network $\Asf$, denoted as $\Rmc_{c}(\Asf)$, is the closure of the set of all achievable rate vectors $\Rbf$, where a rate vector $\Rbf\in\Rmc_c(\Asf)$ is achievable if there
exist a sequence of encoders $\{f^{(n)}=[f_{e}^{(n)}|e\in\Emc]\}$ and decoders $\{d^{(n)}=[d_{t}^{(n)}|t\in\Tmc]\}$ such that $p^{(n),{\rm err}}(\Rbf)\to 0$
as $n\to\infty$.
\end{definition}

The rate region $\mathcal{R}_{c}(\Asf)$ can be expressed in terms of the region
of entropic vectors, $\Gamma_{N}^{*}$, as in \cite{YeungBook,YanYeungTranIT2012}.   The discussion on $\Gamma_N^*$ and its bounds is deferred to \S\ref{subsec:entropicregion}.  For the
hyperedge MSNC problem, define a set  
$\mathcal{N}=\left\{Y_{s},U_{e} | s\in \Smc,e\in \Emc_U\right\}$ with single letter random variables associated with sources and edges, respectively,  and define $N=|\mathcal{N}|=K+L$.  Then, if we collect joint entropies of all non-empty subsets of $\Nmc$ into a vector $\hbf=[h_{\Amc}|\Amc\subseteq 2^{\Nmc}]$, we have $\hbf\in\Gamma_N^*$.
\begin{comment}
The rate region $\mathcal{R}_{*}(\Asf)$ is the set of rate vectors $\mathbf{R}$ such that
there exists $\mathbf{h}\in\Gamma_{N}^{*}$ satisfying the following
\begin{eqnarray}
h_{\Ybf_{\Smc}} & = & \sum_{s\in\Smc}h_{Y_{s}}\label{eq:indepSources}\\
h_{\Ubf_{{\rm Out}(g)}|(\Ybf_{\Smc\cap{\rm In}(g)}\cup \Ubf_{\Emc_U\cap{\rm In}(g)})} & = & 0, ~ g\in \Gmc\label{eq:EncFunc}\\
h_{Y_{s}} & \geq & H(Y_{s}),~ s\in \Smc\label{eq:sourceEnt}\\
h_{\beta(t)|\Ubf_{{\rm In}(t)}} & = & 0,\, ~ t\in \Tmc \label{eq:DecFunc}\\
R_{e} & \geq & h_{U_{e}}, ~ e\in\Emc_U\label{eq:RateCons}
\end{eqnarray}
where $h_{A}$ is the coordinate in the entropy vector $\hbf$ associated with $A$ and the conditional entropies $h_{A|B}$ are naturally equivalent to $h_{AB}-h_{B}$.  These constraints can be interpreted as follows: (\ref{eq:indepSources}) represents that sources are independent; (\ref{eq:EncFunc}) represents that outgoing messages of an intermediate node are functions of all its incoming sources and edges; \eqref{eq:sourceEnt} represents that the associated entropies in the entropic vector must obey the actual source rate constraints; (\ref{eq:DecFunc}) represents that recovered source messages at a sink are a function of the input edges available to it; and (\ref{eq:RateCons}) represents the edge capacity constraints.  
\end{comment}

As will be shown in \S\ref{subsec:entropicregion}, $\Gamma_N^*$ is in the space of $\Rbb^{2^N-1}$.  Note that the edge capacities, $\boldsymbol{r}=[R_e | e\in\Emc_U]$, are extra variables associated with each edge.  Therefore, we will consider the space in $\mathbb{R}^M$, where $M=2^N-1+L,\ L=|\Emc_U|$.  We define $\mathcal{L}_{i},i=1,3,4',5$ as network constraints representing source
independence, coding by intermediate nodes, edge capacity constraints, and sink nodes decoding constraints respectively:
\begin{eqnarray}
\mathcal{L}_{1} & = & \{\mathbf{h}\in\mathbb{R}^M:h_{\Ybf_{\mathcal{S}}}=\Sigma_{s\in\mathcal{S}}h_{Y_{s}}\}  \label{eq:rrcondef1} \\
%\mathcal{L}_{2} & = & \{\mathbf{h}\in\mathbb{R}^M:h_{\Ubf_{\mathrm{Out}(s)}|Y_{s}}=0,\forall s\in\Smc\}\\
\mathcal{L}_{3} & = & \{\mathbf{h}\in\mathbb{R}^M:h_{\Ubf_{{\rm Out}(g)}|(\Ybf_{\Smc\cap{\rm In}(g)}\cup \Ubf_{\Emc_U\cap{\rm In}(g)})} =0,g\in \Gmc\} \label{eq:rrcondef2} \\
\Lmc_{4'}&=&\{[\mathbf{h}^T,\boldsymbol{r}^T]^T\in\mathbb{R}_{+}^{M}:R_e\geq h_{U_{e}},e\in\Emc_U\} \label{eq:Lratefree}\\
\mathcal{L}_{5} & = & \{\mathbf{h}\in\mathbb{R}^M:h_{\Ybf_{\beta(t)}|\Ubf_{\text{In}(t)}}=0,\forall t\in\Tmc\} \label{eq:rrcondef4}.\vspace{-0.2cm}
\end{eqnarray}
and we will denote $\mathcal{L}_{13} = \mathcal{L}_1 \cap \mathcal{L}_3$, $\mathcal{L}_{4'5} = \mathcal{L}_{4'}\cap\mathcal{L}_5$ and $\mathcal{L}_{\Asf} = \mathcal{L}_1  \cap \mathcal{L}_3 \cap \mathcal{L}_{4'} \cap \mathcal{L}_5$.  Note that we do not have $\Lmc_2$ constraints (which represent the coding function at each source) as in \cite{YanYeungTranIT2012}, due to our different notation with $U_{e}=Y_s$ if $\text{Tl}(e)=s$.
Further, $\Gamma^*_N$ and $\mathcal{L}_{i},i=1,3,5$
are viewed as subsets of $\mathbb{R}^M$ with indexed by $\boldsymbol{r}^T$ unconstrained, since they actually are in the space of $\Rbb^{2^N-1}$.  $\Lmc_{4'}$ is also viewed as subset of $\mathbb{R}^M$, with the unreferenced dimensions (i.e. all non-singleton entropies) left unconstrained.  The following extension of the rate region from \cite{YanYeungTranIT2012} characterizes our slightly different rate region formulation $\Rmc_c(\Asf)$ for our slightly different problem.

\begin{theorem}
\label{thm:rateregion}
The expression of the rate region of a network $\Asf$ is
\begin{equation}
\Rmc_{c}(\Asf)=\mathrm{Proj}_{\boldsymbol{r},\boldsymbol{\omega}}(\overline{\rm{con}(\Gamma_{N}^{*}\cap\mathcal{L}_{13})}\cap\mathcal{L}_{4'5}),\label{eq:generalrateregionfree}
\end{equation}
where ${\rm con}(\mathcal{B})$ is the conic
hull of $\mathcal{B}$, and $\mathrm{Proj}_{\boldsymbol{r},\boldsymbol{\omega}}(\mathcal{B})$
is the projection of the set $\mathcal{B}$ on the coordinates $\left[\boldsymbol{r}^T,\boldsymbol{\omega}^T \right]^T$ where $\boldsymbol{r} = \left[ R_e | e\in\Emc_U\right]$ and $\boldsymbol{\omega} = \left[H(Y_s) | s\in\Smc\right]$.  
\end{theorem}
\begin{IEEEproof}
We present a sketch of the proof here and a detailed proof in Appendix \ref{app:thm1proof}.  First observe that the proof of Theorem 1 in \cite{YanYeungTranIT2012} can be extended to networks presented above, with hyperedges and intermediate nodes having direct access to sources.  Some differences include: I) the hyperedge model potentially makes one edge variable connected with more than one node and thus be involved in more than one intermediate node constraint ($\Lmc_3$).  Therefore, it may constrain more on the edge variable in the region of entropic vectors; II) the coding function for each intermediate (hyper)edge may encode some source edges with some other non-source edges together; III) the decoding at sink nodes may be a function of some source edges and non-source edges as well; IV) there is only one outgoing edge for each source and it carries the source variable itself.  The differences will not destroy the essence of the proofs in \cite{YanYeungTranIT2012}.  For the converse and achievability proof, we  view the edge capacities as constant (recall that our rate vector include both source entropies and edge capacities), and then consider the converse and achievability of the associated source entropies, which becomes essentially the proof in \cite{YanYeungTranIT2012}.
\end{IEEEproof}

While the analytical expression determines, in principle, the rate region of any network
under network coding, it is only an implicit characterization.  This is because $\Gamma_{N}^{*}$ is unknown and even non-polyhedral for $N\geq4$.  Further, while $\bar{\Gamma}^*_N$ is a convex cone for all $N$, $\Gamma^*_N$ is already non-convex by $N=3$, though it is also known that the closure only adds points at the boundary of $\bar{\Gamma}^*_N$.  Thus, the direct calculation of rate regions from  \eqref{eq:generalrateregionfree}
for a network with 4 or more variables is infeasible.  On a related note, at the time of writing, it appears to be unknown by the community whether or not the closure after the conic hull is actually necessary\footnote{The closure would be unnecessary if $\bar{\Gamma}^*_N = \textrm{con}(\Gamma^*_N)$, i.e. if every extreme ray in $\bar{\Gamma}^*_N$ had at least one point along it that was entropic (i.e. in $\Gamma^*_N$).  At present, all that is known is that $\Gamma^*_N$ has a solid core, i.e. that the closure only adds points on the boundary of $\bar{\Gamma}^*_N$.} in (\ref{eq:generalrateregionfree}), and the uncertainty that necessitates its inclusion muddles a number of otherwise simple proofs and ideas.  For this reason, some of the discussion in the remainder of the manuscript will study a closely related inner bound to $\Rmc_c(\Asf)$ described in the following corollary.  In all of the cases where the rate region has been computed to date these two regions are equivalent to one another. 

\begin{corollary}
The rate region $\Rmc_{c}(\Asf)$ of a network $\Asf$ is inner bounded by the region
\begin{equation}
\Rmc_{\ast}(\Asf) = \textrm{Proj}_{\boldsymbol{r},\boldsymbol{\omega}} \textrm{con} ( \Gamma^*_N )\cap \Lmc_{\Asf} 
\end{equation}
\end{corollary}
\begin{IEEEproof}
Clearly $\Rmc_c = \mathrm{Proj}_{\boldsymbol{r},\boldsymbol{\omega}}(\overline{\rm{con}(\Gamma_{N}^{*}\cap\mathcal{L}_{13})}\cap\mathcal{L}_{4'5}) \supseteq \mathrm{Proj}_{\boldsymbol{r},\boldsymbol{\omega}}(\rm{con}(\Gamma_{N}^{*}\cap\mathcal{L}_{13})\cap\mathcal{L}_{4'5})$.  Next, observe that intersecting with $\Lmc_{13}$ is equivalent to requiring certain information inequalities (which are non-negative for all entropic vectors) to be identically zero, and a conic combination of such entropic vectors thus can only yield such an information inequality identically zero if the same information inequality was identically zero for each entropic vector. Hence $\textrm{con}(\Gamma^*_N \cap \Lmc_{13}) = \textrm{con}(\Gamma^*_N) \cap \Lmc_{13}$, and thus, $\mathrm{Proj}_{\boldsymbol{r},\boldsymbol{\omega}}(\rm{con}(\Gamma_{N}^{*}\cap\mathcal{L}_{13})\cap\mathcal{L}_{4'5}) = \mathrm{Proj}_{\boldsymbol{r},\boldsymbol{\omega}}(\rm{con}(\Gamma_{N}^{*}) \cap\mathcal{L}_{\Asf})$.  This completes the proof.
\end{IEEEproof}

Again, both $\Rmc_c(\Asf)$ and its closely related inner bound $\Rmc_{\ast}(\Asf)$ are not directly computable because they depend on the unknown region of entropic vectors and its closure.  However,
replacing $\Gamma_{N}^{*}$ with finitely generated inner and outer bounds, as described in the following corollaries, transforms (\ref{eq:generalrateregionfree}) into a polyhedral computation problem,
which involves applying some linear constraints onto a polyhedron and then
projecting down onto some coordinates.  

\begin{corollary}
Let $\Amc \subset \Gamma^*_N$ be some finite set of entropic vectors, then a polyhedral inner bound to the rate region is given by
\begin{equation}
\Rmc_{c}(\Asf) \supseteq \Rmc_{\ast}(\Asf)  \supseteq \mathrm{Proj}_{\boldsymbol{r},\boldsymbol{\omega}}(\textrm{con}(\Amc) \cap \Lmc_{\Asf}).
\end{equation}
\end{corollary}
\begin{IEEEproof}
It is clear that $\mathrm{Proj}_{\boldsymbol{r},\boldsymbol{\omega}}(\overline{\rm{con}(\Amc \cap\mathcal{L}_{13})}\cap\mathcal{L}_{4'5})$ will be an inner bound to $\mathrm{Proj}_{\boldsymbol{r},\boldsymbol{\omega}}(\overline{\rm{con}(\Gamma_{N}^{*}\cap\mathcal{L}_{13})}\cap\mathcal{L}_{4'5})$ and hence $\Rmc_{\ast}(\Asf)$.  Furthermore, for such a finite set $\Amc \cap\mathcal{L}_{13}$ must also be a finite set, and hence $\overline{\rm{con}(\Amc \cap\mathcal{L}_{13})} = \rm{con}(\Amc \cap\mathcal{L}_{13})$ is a closed polyhedral cone.  Additionally, observe that every equality in $\mathcal{L}_{13}$ can be viewed as setting a non-negative definite information inequality quantity to zero, and since every point in $\Amc$ must thus lie in only the non-negative half spaces these equalities generate, the extreme rays of $\rm{con}(\Amc \cap\mathcal{L}_{13})$ must be those extreme rays of $\rm{con}(\Amc)$ in $\mathcal{L}_{13}$, implying $\rm{con}(\Amc \cap\mathcal{L}_{13}) = \rm{con}(\Amc) \cap \mathcal{L}_{13}$.  Putting these facts together we observe that the inner bound $\mathrm{Proj}_{\boldsymbol{r},\boldsymbol{\omega}}(\overline{\rm{con}(\Amc \cap\mathcal{L}_{13})}\cap\mathcal{L}_{4'5}) = \mathrm{Proj}_{\boldsymbol{r},\boldsymbol{\omega}} (\rm{con}(\Amc) \cap \mathcal{L}_{13} \cap \mathcal{L}_{4'5} )= \mathrm{Proj}_{\boldsymbol{r},\boldsymbol{\omega}}(\rm{con}(\Amc) \cap \mathcal{L}_{\Asf})$.
\end{IEEEproof}

Similarly, polyhedral cones $\Gamma_N^{\rm out}$ outer bounding the convex cone $\bar{\Gamma}_N^*$ yield polyhedral outer bounds to the rate region.

\begin{corollary}
Let $\Gamma_N^{\rm out}$ be a closed polyhedral cone that contains $\bar{\Gamma}^*_N$, then a polyhedral outer bound to the rate rate region is given by
\begin{equation}
\Rmc_{c}(\Asf) \subseteq \mathrm{Proj}_{\boldsymbol{r},\boldsymbol{\omega}}(\Gamma_N^{\rm out} \cap \Lmc_{\Asf})
\end{equation}
\end{corollary}
\begin{IEEEproof}
Since $\Gamma^*_N \subset \bar{\Gamma}^*_N \subset \Gamma_N^{\rm out}$, $\Gamma^*_N \cap \Lmc_{13} \subseteq \bar{\Gamma}^*_N\cap \Lmc_{13}  \subseteq \Gamma_N^{\rm out}\cap \Lmc_{13} $.  Thus, $\overline{\rm{con}(\Gamma^*_N \cap\mathcal{L}_{13})} \subseteq \overline{\rm{con}(\Gamma^{\rm out}_N \cap\mathcal{L}_{13})}  = \Gamma^{\rm out}_N \cap\mathcal{L}_{13}$.  Hence $\Rmc_{\ast}(\Asf)=\mathrm{Proj}_{\boldsymbol{r},\boldsymbol{\omega}}(\overline{\rm{con}(\Gamma_{N}^{*}\cap\mathcal{L}_{13})}\cap\mathcal{L}_{4'5}) \subseteq \mathrm{Proj}_{\boldsymbol{r},\boldsymbol{\omega}}(\Gamma^{\rm out}_N \cap \Lmc_{\Asf})$.
\end{IEEEproof}

These corollaries inspire us to substitute
$\Gamma_{N}^{*}$ with such closed polyhedral outer and inner bounds $\Gamma_N^{\rm out}$,and $\Gamma_N^{\rm in}= \textrm{con}(\mathcal{A}), \mathcal{A} \subset \Gamma_N^*$ to $\bar{\Gamma}^*_N$,
respectively, to obtain an outer and inner bound on the rate region:
\vspace{-0.3cm}
\begin{eqnarray}\label{eq:regionout}
\mathcal{R}_{{\rm out}}(\Asf)&=&\mathrm{proj}_{\boldsymbol{r},\boldsymbol{\omega}}(\Gamma_{N}^{{\rm out}}\cap\mathcal{L}_{\Asf}),\label{eq:regionoutfree}\\
\label{eq:regionin}
\mathcal{R}_{{\rm in}}(\Asf)&=&\mathrm{proj}_{\boldsymbol{r},\boldsymbol{\omega}}(\Gamma_N^{{\rm in}}\cap\mathcal{L}_{\Asf}).\label{eq:regioninfree}
\end{eqnarray}

If $\mathcal{R}_{{\rm out}}(\Asf)=\mathcal{R}_{{\rm in}}(\Asf)$, we know $\mathcal{R}_c(\Asf)=\Rmc_{\ast}(\Asf)=\mathcal{R}_{{\rm out}}(\Asf)=\mathcal{R}_{{\rm in}}(\Asf)$. Otherwise, tighter bounds are necessary.

In this work, we will use \eqref{eq:regionoutfree} and \eqref{eq:regioninfree} to calculate the rate region. Typically the Shannon outer bound $\Gamma_N$ and some inner bounds obtained from matroids, especially representable matroids, are used. We will briefly review the definition of these bounds in the next subsection, while details on the polyhedral computation methods with these bounds are available in \cite{CongduanNetCod2013,CongduanAllerton2012,JayantISIT2014,JayantNetCod2015}. 

\subsection{Construction of bounds on rate region}\label{sec:bounds}
An introduction to the region of entropic vectors and the polyhedral inner and outer bounds we will utilize from it can be found in greater detail in \cite{CongduanTranIT2014}.  Here we briefly review their definitions for accuracy, completeness, and convenience.

\subsubsection{Region of entropic vectors $\Gamma_{N}^{*}$}
\label{subsec:entropicregion}
Consider an arbitrary collection $\mathbf{X}=[X_{1},\ldots,X_{N}]$
of $N$ discrete random variables with joint probability mass function
$p_{X}$. To each of the $2^{N}-1$ non-empty subsets of the collection
of random variables, $X_{\Amc}:=[X_{i} | i\in\Amc]$
with $\Amc\subseteq\{1,\ldots,N\}$, there is associated
a joint Shannon entropy $H(X_{\Amc})$. Stacking these subset
entropies for different subsets into a $2^{N}-1$ dimensional vector
we form an entropy vector 
\begin{equation}
\mathbf{h}=[H(X_{\Amc}) | \Amc\subseteq\{1,\ldots,N\},\Amc\neq\emptyset].\end{equation}
 By virtue of having been created in this manner, the vector $\mathbf{h}$
must live in some subset of $\mathbb{R}_{+}^{2^{N}-1}$, and is said
to be \emph{entropic} due to the existence of $p_{X}$. However, not
every point in $\mathbb{R}_{+}^{2^{N}-1}$ is entropic since, for many points, there does not exist an associated valid distribution $p_{X}$.
All entropic vectors form a region denoted as $\Gamma_{N}^{*}$.
It is known that the closure of the region of entropic vectors $\bar{\Gamma}_{N}^{*}$ is a convex cone \cite{YanYeungTranIT2012}.  Elementary inequalities on Shannon entropies should form a fundamental outer bound on $\bar{\Gamma}_N^*$, named the {\it Shannon outer bound $\Gamma_N$}.

\subsubsection{Shannon outer bound $\Gamma_{N}$}

We observe that elementary properties of Shannon entropies indicate
that $H(X_{\Amc})$ is a non-decreasing submodular function,
so that $\forall\Amc\subseteq\Bmc\subseteq\{1,\ldots,N\},\forall\mathcal{C},\mathcal{D}\subseteq\{1,\ldots,N\}$\vspace{-0.2cm}
\begin{eqnarray}
H(X_{\Amc}) & \leq & H(X_{\Bmc})\label{eq:nonDec}\\
H(X_{\mathcal{C}\cup\mathcal{D}})+H(X_{\mathcal{C}\cap\mathcal{D}}) & \leq & H(X_{\mathcal{C}})+H(X_{\mathcal{D}}).\label{eq:submod}
\end{eqnarray}
 Since they are true for any collection of subset entropies, these
linear inequalities (\ref{eq:nonDec}), (\ref{eq:submod}) can be
viewed as supporting halfspaces for $\Gamma_{N}^{*}$.

Thus, the intersection of all such inequalities form a polyhedral
outer bound $\Gamma_{N}$ for $\Gamma_{N}^{*}$ and $\bar{\Gamma}_{N}^{*}$,
where 
\[
\Gamma_{N}:=\left\{ \mathbf{h}\in\mathbb{R}^{2^{N}-1}\left|\begin{array}{c}
h_{\Amc}\leq h_{\Bmc}\quad\forall\Amc\subseteq\Bmc\\
h_{\mathcal{C}\cup\mathcal{D}}+h_{\mathcal{C}\cap\mathcal{D}}\leq h_{\mathcal{C}}+h_{\mathcal{D}}\quad\forall\mathcal{C},\mathcal{D}
\end{array}\right.\right\} .
\]
 This outer bound $\Gamma_{N}$ is known as the \emph{Shannon outer
bound}, as it can be thought of as the set of all inequalities resulting
from the positivity of Shannon's information measures among the random
variables. 
While $\Gamma_{2}=\Gamma_{2}^{*}$ and $\Gamma_{3}=\bar{\Gamma}_{3}^{*}$,
$\bar{\Gamma}_{N}^{*}\subsetneq\Gamma_{N}$ for all $N\geq4$ \cite{YanYeungTranIT2012},
and indeed it is known \cite{Matus_ISIT_2007} that $\bar{\Gamma}_{N}^{*}$
is non-polyhedral for $N\geq4$.

The inner bounds on $\bar{\Gamma}_N^*$ we consider are based on representable matroids.  We briefly review the basic definitions of matroids and representable matroids.

\subsubsection{\label{sec:Matroid-theory}Matroid basics}

Matroid theory \cite{OxleyMatroidBook} is an abstract generalization
of independence in the context of linear algebra and graphs to the more general
setting of set systems.  There are numerous equivalent
definitions of matroids, however, we will present the definition of matroids utilizing {\em{rank
functions}} as this is best matched to our purposes. 
\begin{definition} \label{def:mtrk}
A set function on a ground set $\Mmc$, $r_{\Msf}:2^{\mathcal{M}}\to\{0,\ldots,|\Mmc|\}$, is a rank function of a matroid $\Msf$
if it obeys the following axioms: 
\begin{enumerate}
\item Cardinality: $r_{\Msf}(\Amc)\leq|\Amc|$; 
\item Monotonicity: if $\Amc\subseteq \Bmc\subseteq \mathcal{M}$ then $r_{\Msf}(\Amc)\leq r_{\Msf}(\Bmc)$; 
\item Submodularity: if $\Amc,\Bmc\subseteq \mathcal{M}$ then $r_{\Msf}(\Amc\cup \Bmc)+r_{\Msf}(\Amc\cap \Bmc)\leq r_{\Msf}(\Amc)+r_{\Msf}(\Bmc)$. 
\end{enumerate}
\end{definition} 

A subset with rank function $r_{\Msf}(\Amc)=|\Amc|$ is called an independent set of the matroid.  Though there are many classes of matroids, we are especially interested
in one of them, {\em{representable matroids}}, because they can
be related to linear codes to solve network coding problems as discussed
in \cite{CongduanNetCod2013,CongduanAllerton2012}.

\subsubsection{Representable matroids}

Representable matroids are an important class of matroids which connect
the independent sets to the notion of independence in
a vector space. 
\begin{definition} A matroid $\Msf$ with ground set
$\Mmc$ of size $|\Mmc|=N$ and rank $r_{\Msf}(\Mmc)=r$ is representable over a
field $\mathbb{F}$ if there exists a matrix $\Abb\in\mathbb{F}^{r\times N}$
such that for each set $\Amc\subseteq\Mmc$ the rank $r_{\Msf}(\Amc)$ equals the linear rank of the corresponding
columns in $\Abb$, viewed as vectors in $\mathbb{F}^{r}$. 
\end{definition} 

Note that, for an independent set $\Imc$, the corresponding columns in the matrix $\Abb$ are linearly independent.  There has been significant effort towards
characterizing the set of matroids that are representable over various
field sizes, with a complete answer only available for fields of sizes
two, three, and four. For example, a matroid $\Msf$ is binary representable (representable over a binary
field) iff it does not have the matroid $U_{2,4}$ as a minor. 
Here, a minor is obtained by series of operations of contraction and deletion \cite{OxleyMatroidBook}. $U_{k,N}$ is the {\em uniform} matroid on the ground set
$\Mmc=\{1,\ldots,N\}$ with independent sets $\mathcal{I}$ equal to all subsets
of $\{1,\ldots,N\}$ of size at most $k$. For example, $U_{2,4}$ has as its
independent sets \vspace{-0.2cm}
\begin{equation}
\mathcal{I}=\{\emptyset,1,2,3,4,\{1,2\},\{1,3\},\{1,4\},\{2,3\},\{2,4\},\{3,4\}\}.
\end{equation}
Another important observation is that the first non-representable
matroid is the so-called {\em{Vámos}} matroid, a well known matroid on ground
set of size $8$. That is to say, all matroids are representable,
at least in some field, for $N\leq7$.

\subsubsection{Inner bounds from representable matroids}

Suppose a matroid $\Msf$ with ground set $\Mmc$ of size $|\Mmc|=N$
and rank $r_{\Msf}(\Mmc)=r$ is representable over the finite field $\Fbb_{q}$
of size $q$ and the representing matrix is $\Abb\in\Fbb_{q}^{r\times N}$
such that  $r_{\Msf}(\Bmc)=\textrm{rank}(\Abb_{:,\Bmc}), \forall \Bmc\subseteq\Mmc$,
the matrix rank of the columns of $\Abb$ indexed by $\Bmc$. Let $\Gamma_{N}^{q}$
be the conic hull of all rank functions of matroids with $N$ elements
and representable in $\Fbb_{q}$. This provides an inner bound $\Gamma_{N}^{q}\subseteq\bar{\Gamma}_{N}^{*}$,
because any extremal rank function of $\Gamma_{N}^{q}$ is by
definition representable and hence is associated with a matrix representation
$\Abb\in\Fbb_{q}^{r\times N}$, from which and $r$ random variables $\ubf$ uniformly distributed in $\Fbb_q$, we can create the random
variables 
\begin{equation}
[X_{1},\ldots,X_{N}]=\ubf\Abb,~~\ubf\sim\mathrm{Uniform}(\Fbb_{q}^{r}),\label{eq:entVecRepMat}
\end{equation}
whose elements have joint entropies $h_{\Amc}=r_{\Msf}(\Amc)\log_{2}q,\ \forall \Amc\subseteq\Mmc$.  
Hence, all extreme rays of $\Gamma_{N}^{q}$ are entropic, and $\Gamma_{N}^{q}\subseteq\bar{\Gamma}_{N}^{*}$. Further, if a vector in the rate region of a network is (a projection of) an $\Fbb_q$-representable matroid rank, the representation $\Abb$ can be used as a linear code to achieve that rate vector, and this code is denoted a {\it basic scalar $\Fbb_q$ code}. For an interior point in the rate region, which is the conic hull of projections of $\Fbb_q$-representable matroid ranks, the code to achieve it can be constructed by time-sharing between the basic scalar codes associated with the ranks involved in the conic combination. This code is denoted a {\it scalar $\Fbb_q$ code}. Details on construction of such a code can be found in \cite{CongduanNetCod2013} and \cite{CongduanTranIT2014}.

One can further generalize the relationship between representable matroids and entropic vectors established by (\ref{eq:entVecRepMat}).  Suppose the ground set $\Mmc'=\{1',\ldots,N'\}$ and a partition $\Mmc=\{1,\ldots,N\}$.  We define a partition mapping $\pi:\Mmc'\rightarrow\Mmc$ such that $\cup_{i'\in{\Mmc'}}\pi(i')=\Mmc'$, and $\pi(i')\cap \pi(j')=\emptyset,i',j'\in\Mmc', i'\neq j'$.  That is, the set $\Mmc'$ is partitioned into $N$ disjoint sets.  Suppose the variables associated with $\Mmc'$ are $X_{1'},\ldots,X_{N'}$.  Now we define for $n\in\Mmc$ the new vector-valued random variables $\Ybf_{n}=[X_{i'}|i'\in \pi^{-1}(n)]$.  The associated entropic vector will have entropies $h_{\Amc}=r_{\Msf}(\cup_{n\in \Amc}\pi^{-1}(n))\log_{2}q,\Amc\subseteq \Mmc$, and is thus proportional to a \emph{projection} of the original rank vector $\rbf$ keeping only those elements corresponding to all elements in a set in the partition appearing together.  Thus, such a projection of $\Gamma_{N'}^{q}$ forms an inner bound to $\bar{\Gamma}_{N}^{*}$, which we will refer to as a \emph{vector representable matroid inner bound} $\Gamma_{N,N'}^{q}$. As $N' \to \infty$, $\Gamma_{N,\infty}^q$ is the conic hull of all ranks of subspaces on $\Fbb_q$. The union over all field sizes for $\Gamma_{N,\infty}^q$ is the conic hull of the set of ranks of subspaces. Similarly, if a vector in the rate region of a network is (a projection of) a vector $\Fbb_q$-representable matroid rank, the representation $\Abb$ can be used as a linear code to achieve that rate vector, and this code is denoted as a {\it basic vector $\Fbb_q$ code}. The time-sharing between such basic vector codes can achieve any point inside the rate region \cite{CongduanNetCod2013,CongduanTranIT2014}.

\subsubsection{Dimension function of linear subspace arrangements}

As stated above, $\Gamma_{N,N'}^q$ becomes a tighter and tighter inner bound on $\bar{\Gamma}_N^*$ as $N'\to\infty$.  This considers the increase in dimension but does not consider the fields other than $\Fbb_q$.  Actually, if we consider all possible $\Fbb_q$ fields and let $N'\to \infty$, we will get the inner bound associated with all linear codes, denoted by $\Gamma_N^{{\rm linear}}$, which is tighter than $\Gamma_{N,\infty}^q$ for a fixed $\Fbb_q$.  Specifically, consider a collection of $N$ linear vector subspaces $\Vbf=(V_{1},\ldots,V_{N})$ of a finite dimensional vector space, and define the set function $\drm:2^{\Vbf}\to\Nbb_{+}$, where $\drm(\Amc)=\mathrm{dim}\left(\sum_{i\in \Amc}V_{i}\right)$ for each $\Amc\subseteq \{1,\ldots,N\}$ is the dimension of the vector space generated by the union of subspaces indexed by $\Amc$. For any collection of subspaces $\Vbf$, the function $\drm$ is integer valued, and obeys monotonicity and submodularity. Additionally, for every subspace dimension function $\drm$, there is an associated entropic vector. Indeed, one can place the vectors forming a basis for each $V_i$, over all $i$, side by side into a matrix $\Abb$, which when utilized in (\ref{eq:entVecRepMat}), will yield random subvectors having the desired entropies.  Thus, the conic hull of dimensions of linear subspace arrangements forms an inner bound on $\bar{\Gamma}_{N}^{*}$, we denote it by $\Gamma_{N}^{{\rm linear}}$. 

Integrality, monotonicity, and submodularity are necessary but insufficient for for a given set function $\drm : 2^{\Vbf}\to\Nbb_{+}$ to be dimension function of subspace arrangements. That is, there exist additional inequalities that are necessary to describe the conic hull of all possible subspace dimension set functions. As discussed in \cite{Hammer2000ShannEntr}, Ingleton's inequality \cite{Ingleton1971} together with the Shannon outer bound $\Gamma_4$, completely characterizes $\Gamma_{4}^{{\rm linear}}$.

For $N=5$ subspaces \cite{DFZ2009Ineqfor5var} found $24$ new inequalities in addition to the Ingleton inequalities that hold, and prove this set is irreducible and complete in that all inequalities are necessary and no additional non-redundant inequalities exist. For $N\geq6$, \cite{DFZ2009Ineqfor5var,Kinser2011NewIneqSubspaceArra} there are new inequalities from $N-1$ to $N$, and $\Gamma_N^{{\rm linear}}$ remains unknown.

All the bounds discussed in this section could be used in \eqref{eq:regionoutfree} and \eqref{eq:regioninfree} to calculate bounds on rate regions for a network $\Asf$. If we substitute the Shannon outer bound $\Gamma_N$ into \eqref{eq:regionoutfree}, we get
\begin{equation}
\mathcal{R}_{o}(\Asf)=\mathrm{proj}_{\boldsymbol{r},\boldsymbol{\omega}}(\Gamma_{N}\cap\mathcal{L}_{\Asf}).
\end{equation}

Similarly, we substitute the representable matroid inner bound $\Gamma_N^q$, the vector representable matroid inner bound $\Gamma_{N,N'}^q$, $\Gamma_{N,\infty}^q$ and the linear inner bound $\Gamma_{N}^{{\rm linear}}$ into \eqref{eq:regioninfree}, to obtain
\begin{eqnarray}
\mathcal{R}_{s,q}(\Asf)&=&\mathrm{proj}_{\boldsymbol{r},\boldsymbol{\omega}}(\Gamma_{N}^{q}\cap\mathcal{L}_{\Asf}),\label{eq:rsq}\\
\mathcal{R}_{q}^{N'}(\Asf)&=&\mathrm{proj}_{\boldsymbol{r},\boldsymbol{\omega}}(\Gamma_{N,N'}^{q}\cap\mathcal{L}_{\Asf}),\\
\mathcal{R}_{q}(\Asf)&=&\mathrm{proj}_{\boldsymbol{r},\boldsymbol{\omega}}(\Gamma_{N,\infty}^{q}\cap\mathcal{L}_{\Asf}),\\
\mathcal{R}_{{\rm linear}}(\Asf)&=&\mathrm{proj}_{\boldsymbol{r},\boldsymbol{\omega}}(\Gamma_{N}^{{\rm linear}}\cap\mathcal{L}_{\Asf}).
\end{eqnarray}

We will present the experimental results utilizing these bounds to calculate the rate regions of various networks in \S\ref{sec:resultssmall}.  However, before we do this, we will first aim to generate a list of network coding problems to which we may apply our computations and thereby calculate their rate regions.  In order to tackle the largest collection of networks possible in this study, in the next section we will seek to obtain a minimal problem description for each network coding problem instance, removing any redundant sources, edges, or nodes.

\section{Minimality Reductions on Networks}\label{sec:minimality}
Though in principle, any network coding problem as described in \S \ref{sec:model} forms a valid network coding problem, such a problem can include networks with nodes, edges, and sources which are completely extraneous and unnecessary from the standpoint of determining the rate region.  To deal with this, in this section, we show how to form a network instance with equal or fewer number of sources, edges, or nodes, from an instance with extraneous components.  We will show the rate region of the instance with the extraneous components is trivial to calculate from the rate region of the reduced network.  Network coding problems without such extraneous and unnecessary components will be called {\it minimal}.  We first define a minimal network coding problem, then show, via Theorem \ref{thm:minimality}, how to map a non-minimal network to a minimal network, and then form the rate region of the non-minimal network directly from the minimal one.

%We define a {\em minimal} network as follows.  
\begin{definition}\label{def:minimal}
An acyclic network instance $\Asf=(\Smc,\Gmc,\Tmc,\Emc,\beta)$ is {\em minimal} if it obeys the following constraints:
\begin{enumerate}
\item[] {\em Source minimality}:
\item [\namedlabel{c1}{(\textbf{C1})}] all sources cannot be only directly connected with sinks: $\forall s\in \Smc$, ${\rm Hd}({\rm Out}(s)) \cap \Gmc \neq \emptyset$;
\item [\namedlabel{c2}{(\textbf{C2})}] sinks do not demand sources to which they are directly connected: $\forall s\in \Smc,\ t\in\Tmc$, if $t \in {\rm Hd}({\rm Out}(s))$ then $s \notin \beta(t)$;
\item [\namedlabel{c3}{(\textbf{C3})}] every source is demanded by at least one sink: $\forall s\in \Smc$, $\exists\,t\in \Tmc$ such that $s\in\beta(t)$ ;
\item [\namedlabel{c4}{(\textbf{C4})}] sources connected to the same intermediate node and demanded by the same set of sinks should be merged: $\nexists s,s'\in\Smc$ such that ${\rm Hd}({\rm Out}(s))={\rm Hd}({\rm Out}(s'))$ and $\gamma(s)=\gamma(s')$, where $\gamma(s)=\{t\in\Tmc|s\in\beta(t)\}$;
\item[] {\em Node minimality}:
\item [\namedlabel{c5}{(\textbf{C5})}] intermediate nodes with identical inputs should be merged: $\nexists\,k,l\in \Gmc$ such that ${\rm In}(k)={\rm In}(l)$;
\item [\namedlabel{c6}{(\textbf{C6})}] intermediate nodes should have nonempty inputs and outputs, and sink nodes should have nonempty inputs: $\forall g\in\Gmc, t\in\Tmc, {\rm In}(g)\neq\emptyset,{\rm Out}(g)\neq\emptyset,{\rm In}(t)\neq\emptyset$;
\item[] {\em Edge minimality}:
\item [\namedlabel{c7}{(\textbf{C7})}] all hyperedges must have at least one head: $\nexists e\in\Emc$ such that $\text{Hd}(e)=\emptyset$;
\item [\namedlabel{c8}{(\textbf{C8})}] identical edges should be merged: $\nexists e,e'\in\Emc$ with ${\rm Tl}(e)={\rm Tl}(e')$,  ${\rm Hd}(e)={\rm Hd}(e')$;
\item [\namedlabel{c9}{(\textbf{C9})}] intermediate nodes with unit in and out degree, and whose in edge is not a hyperedge, should be removed: $\nexists e,e'\in\Emc,g\in\Gmc$ such that ${\rm In}(g)=e$,  ${\rm Hd}(e)=g$, ${\rm Out}(g)=e'$;
\item[] {\em Sink minimality}:
\item [\namedlabel{c10}{(\textbf{C10})}] there must exist a path to a sink from every source wanted by that sink: $\forall t\in\Tmc, \beta(t)\subseteq \sigma(t)$, where $\sigma(t)=\{k\in\Smc|\exists \text{ a path from }k \text{ to } t \}$;
\item [\namedlabel{c11}{(\textbf{C11})}] every pair of sinks must have a distinct set of incoming edges: $\forall t,t'\in\Tmc,i\neq j$, $\mathrm{In}(t)\neq \mathrm{In}(t')$;
\item [\namedlabel{c12}{(\textbf{C12})}] if one sink receives a superset of inputs of a second sink, then the two sinks should have no common sources in demand: If $\mathrm{In}(t)\subseteq \mathrm{In}(t')$, then $\beta(t)\cap\beta(t')=\emptyset$;
%\item [(\textbf{C13}.)] $\nexists t_i\in \Tmc,t_{\Jmc}\subseteq (\Tmc\setminus t_i)$, such that $\mathrm{In}(t_k)=\cup_{j\in\Jmc}\mathrm{In}(t_j)$ and $\beta(t_k)=\cup_{j\in\Jmc}\beta(t_j)$;
%\item [(\textbf{C13}.)] If $|\Smc|>1$, $\exists t\in\Tmc$ such that $\beta(t)\subsetneq \Smc$.
\item [\namedlabel{c13}{(\textbf{C13})}]  if one sink receives a superset of inputs of a second sink, then the sink with superset input should not have direct access to the sources that demanded by the sink with subset input: If $\mathrm{In}(t)\subseteq \mathrm{In}(t')$ then $t'\notin \textrm{Hd}(\textrm{Out}(s))$ for all $s \in\beta(t)$.
\item[] {\em Connectivity}:
\item [\namedlabel{c14}{(\textbf{C14})}] the direct graph associated with the network  $\Asf$ is weakly connected.
\end{enumerate}
\end{definition}
%C1 --> just remove that source
%C2 --> CHANGED TO not demanding any source we have access to (otherwise can trivially satisfy, so remove constraint)
%C3 --> combine the two nodes into one node  (also will never happen in edge representation)
%C4 --> remove the redundant source and any edges with no path to any other source but it
%C5 --> make the two sources into the same source
%C6 --> remove that edge
%C7 --> remove one of the duplicate edges
%C8 --> contract the node/its in-edge
%C9 --> delete the node and any incoming edges connected only to it.  or delete any outgoing edges (if no outgoing edges)
%C10 -> remove any source not in \sigma(t) from \beta(t)
%C11 -> combine the two identical sinks
%C12 -> remove \beta(t_i) \cap \beta(t_j) from \beta(t_j), since it is already implied by conditions at t_i
%C13 -> otherwise, it is already implied that it can decode it from the edges connected to both places, so it doesn't benefit or change at all from requiring the direct access
%C14 -> consider as the two or more separate networks associated with the weakly connected components

%%%%%%% order changed May 16 2015: move old C3, C9 after C5 %%%%
To better highlight this definition of network minimality, we explain the conditions involved in greater detail.  The first condition \ref{c1} requires that a source cannot be only directly connected with some sinks, for otherwise no sink needs to demand it, according to \ref{c2} and \ref{c10}.  Therefore, this source is extraneous.   The condition \ref{c2} holds because otherwise the demand of this sink will always be trivially satisfied, hence removing this reconstruction constraint will not alter the rate region.  Note that other sources not demanded by a given sink can be directly available to that sink as side information (e.g., as in index coding problems), as long as condition \ref{c13} is satisfied.  The condition \ref{c3} indicates that each source must be demanded by some sink nodes, for otherwise it is extraneous and can be removed.   The condition \ref{c4} says that no two sources have exactly the same paths and set of demanders (sinks requesting the source), because in that case the network can be simplified by combining the two sources as a super-source.  The condition \ref{c5} requires that no two intermediate nodes have exactly the same input, for otherwise the two nodes can be combined.   The condition \ref{c6} requires that no nodes have empty input except the source nodes, for otherwise these nodes are useless and extraneous from the standpoint of satisfying the network coding problem.  The condition \ref{c7} requires that every edge variable must be in use in the network, for otherwise it is also extraneous and can be removed.   The condition \ref{c8} guarantees that there is no duplication of hyperedges, for otherwise they can be merged with one another.  The condition \ref{c9} says that there is no trivial relay node with only one non-hyperedge input and output, for otherwise the head of the input edge can be replaced with the head of the output edge.    The condition \ref{c10} reflects the fact that a sink can only decode the sources to which it has at least one path of access,  and any demanded source not meeting this constraint will be forced to have entropy rate of zero.  The condition \ref{c11} indicates the trivial requirement that no two decoders should have the same input, for otherwise these two decoders can be combined.  The condition \ref{c12} simply stipulates that implied capabilities of sink nodes are not to be stated, but rather inferred from the implications.  In particular, if $\mathrm{In}(t)\subseteq \mathrm{In}(t')$, and $\beta(t)\cap\beta(t')\neq\emptyset$, the decoding ability of $\beta(t)$ is implied at $t'$: pursuing minimality, we only let $t'$ demand extra sources, if any.  The condition \ref{c13} is also necessary because the availability of $s$ is already implied by having access to ${\rm In}(t)$, hence, there is no need to have direct access to $s$.

%We will show that impact on rate regions of networks when doing the minimality reductions in \S\ref{sec:minimality}.  It is not difficult to imagine that the reduction are necessary because they will have essentially no impact on rate regions due to the redundancies.   We now would like to define the network codes and rate region associated with it.

We next show that the rate region of the network with extraneous components can be easily derived from the network without extraneous components, and vice versa.  Following the same order of the constraints \ref{c1}--\ref{c14}, we give the actions on each reduction and how the rate region of the network with those extraneous components can be derived.

\begin{figure}
\captionsetup[subfigure]{labelformat=empty}
\centering 
%\subfloat [\label{fig:exampleds}Demonstration of direct sum of two networks: two networks are simply concatenated.]{\includegraphics[scale=0.35]{ConcatenateNetworks} }
%\addtocounter{subfigure}{-1}
\subfloat [%\label{fig:C1examplesi}
\ref{c1}: source $s_3$ does not connected with any intermediate node, and thus is extraneous.]{\includegraphics[scale=0.4]{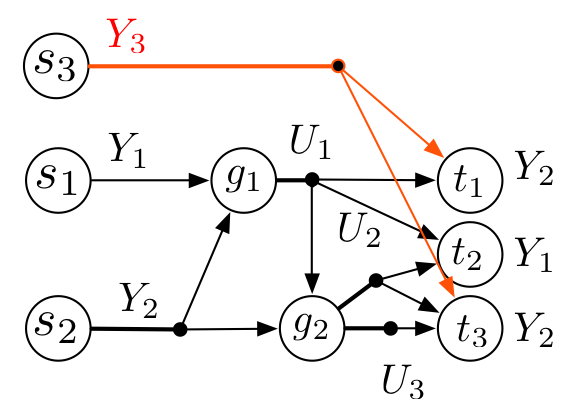} }\hspace{2mm}
\subfloat [\label{fig:C2examplesc}\ref{c2}: sink $t_3$ has direct access to $Y_2$, the demand of $Y_2$ is trivially satisfied and thus $t_3$ is redundant.]{\includegraphics[scale=0.4]{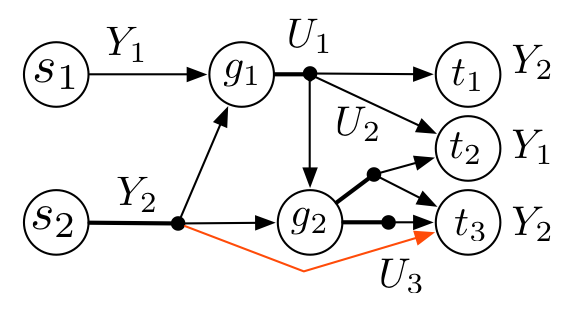} }\hspace{2mm}
\subfloat [\label{fig:C4exampleec}\ref{c3}: source $Y_3$ is not demanded by any sink, and thus is redundant.]{\includegraphics[scale=0.4]{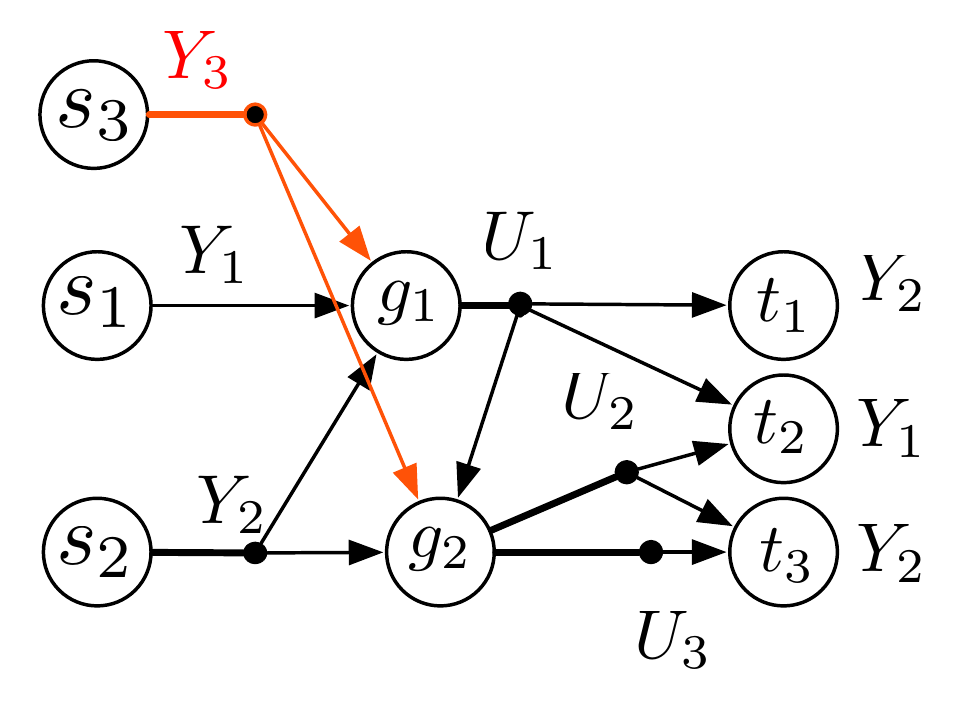} }\hspace{2mm}
\subfloat [\label{fig:C5exampleec}\ref{c4}: sources $Y_1,Y_3$ have exactly the same output and demanders, and thus can be combined.]{\includegraphics[scale=0.4]{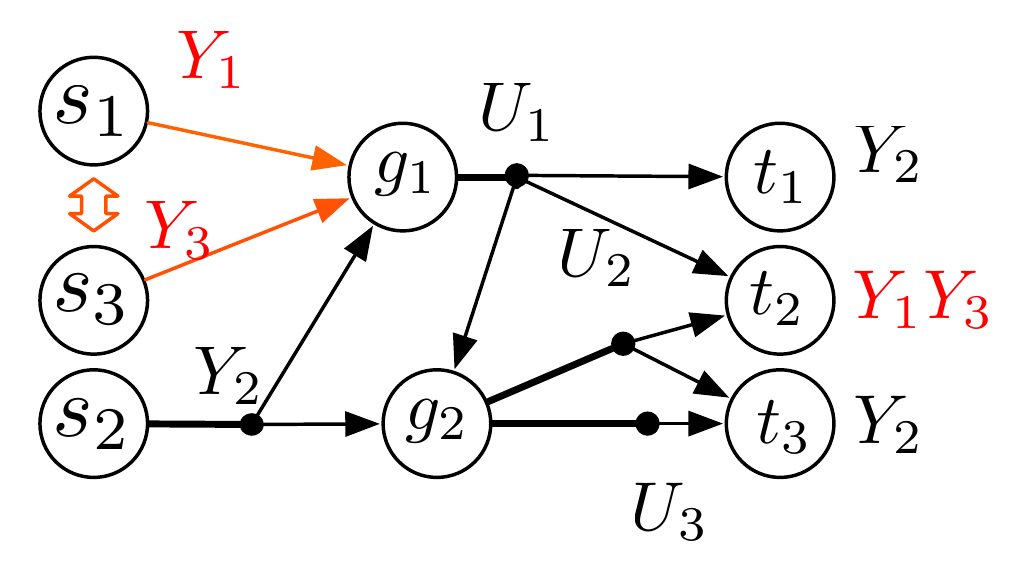} }\hspace{2mm}
\subfloat [\label{fig:C3examplenc}\ref{c5}: node $g_1,g_2$ have same input, and thus can be combined.]{\includegraphics[scale=0.4]{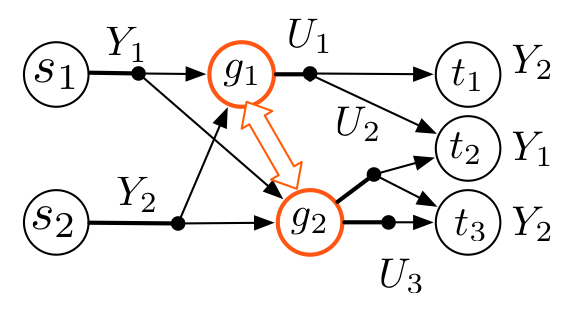} }\hspace{2mm}
\subfloat [\label{fig:C9exampleec}\ref{c6}: node $g_3,g_4$ and sink $t_1$ have empty input/ output, and thus are redundant.]{\includegraphics[scale=0.4]{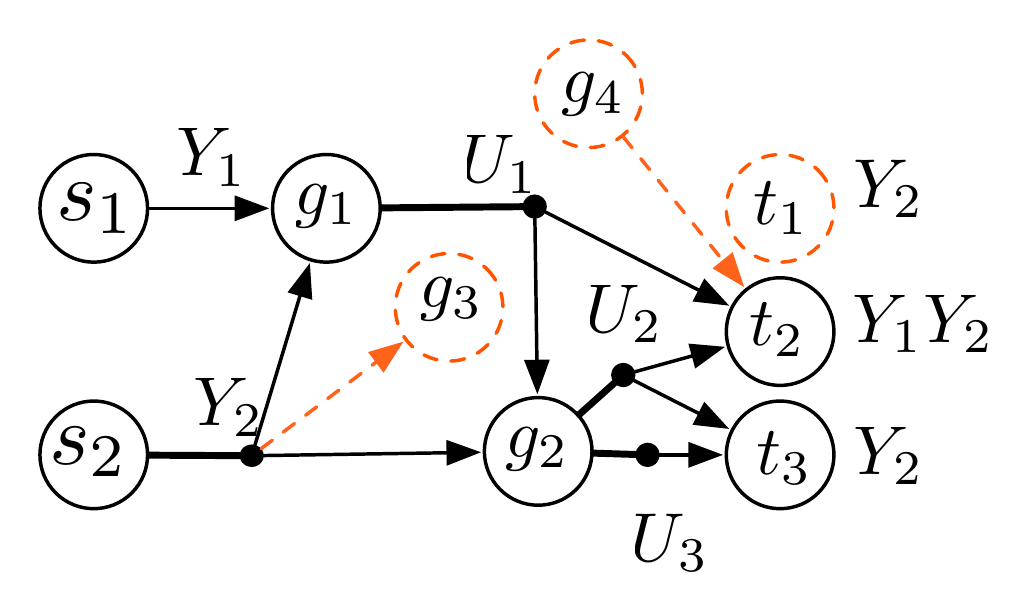} }\hspace{2mm}
\subfloat [\label{fig:C6exampleec}\ref{c7}: edge $U_2$ is not connected to any other nodes, and thus is redundant.]{\includegraphics[scale=0.4]{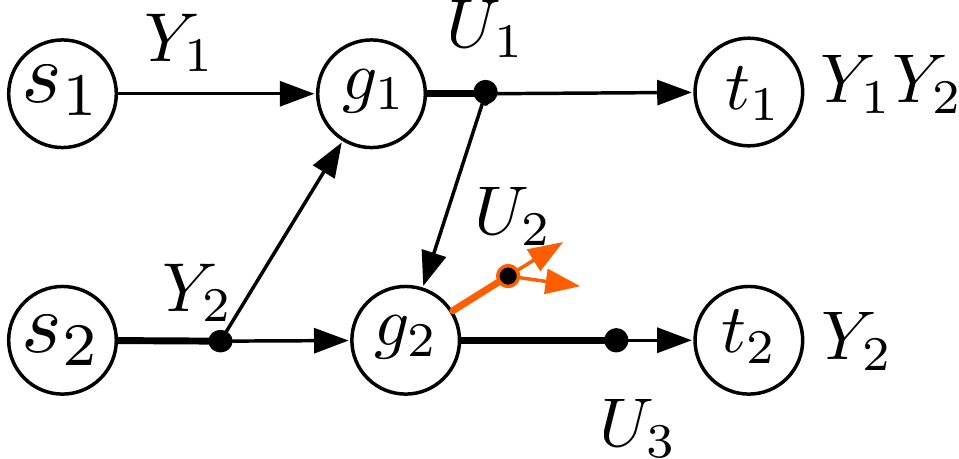} }\hspace{2mm}
\subfloat [\label{fig:C7exampleec}\ref{c8}: edges $U_2,U_3$ have exactly the same input and output nodes, and thus can be combined.]{\includegraphics[scale=0.4]{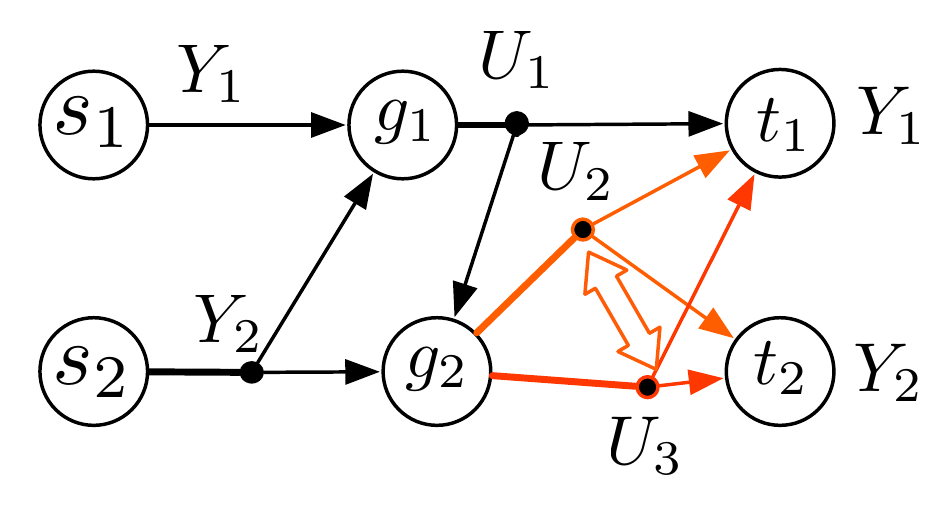} }\hspace{2mm}
\subfloat [\label{fig:C8exampleec}\ref{c9}: node $g_{1'}$ has exactly one input and one output, and they can be combined.]{\includegraphics[scale=0.4]{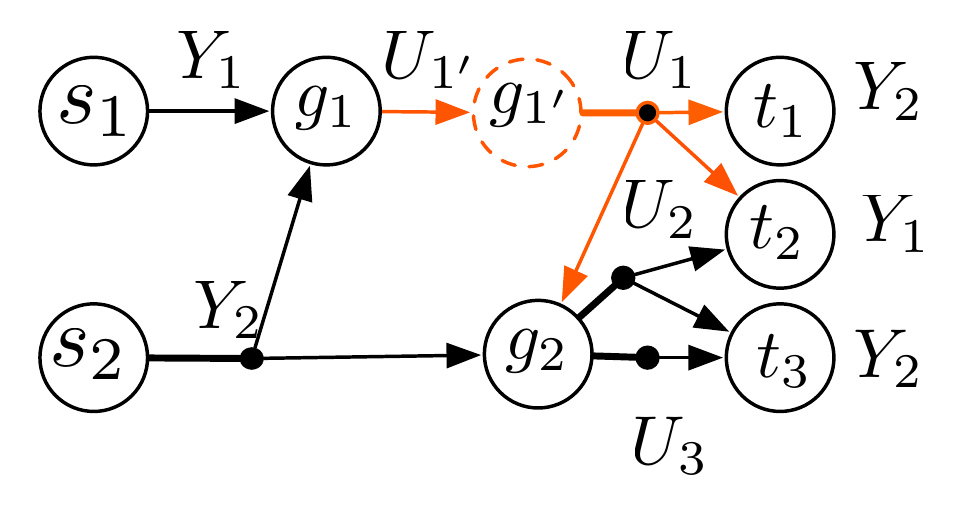} }\hspace{2mm}
\subfloat [\label{fig:C10exampleec}C10: sink $t_3$ has no access to $s_1$ but demands $Y_1$, so the only way to satisfy it is $s_1$ is sending no information.]{\includegraphics[scale=0.4]{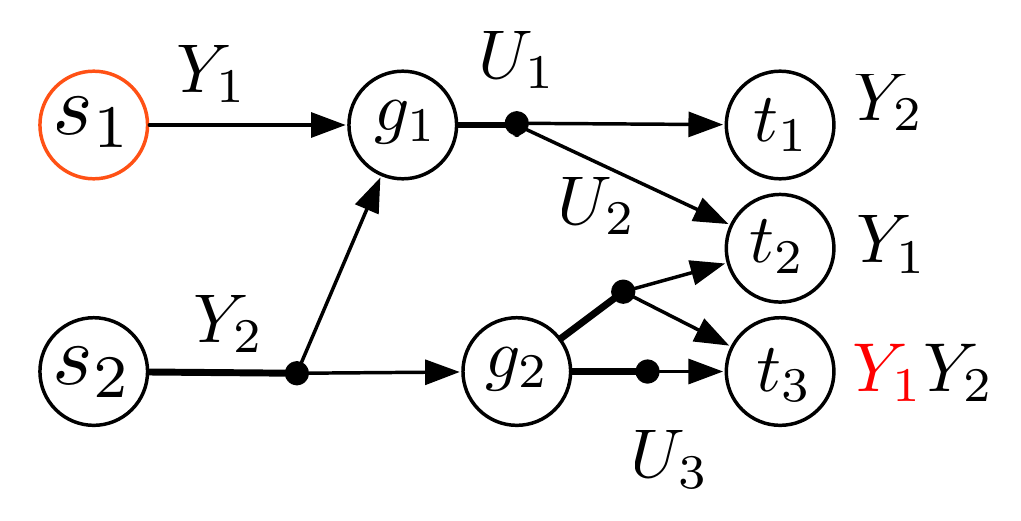} }\hspace{2mm}
\subfloat [\label{fig:C11exampleec}\ref{c11}: sinks $t_1,t_2$ have exactly the same input and thus can be combined into one sink node.]{\includegraphics[scale=0.4]{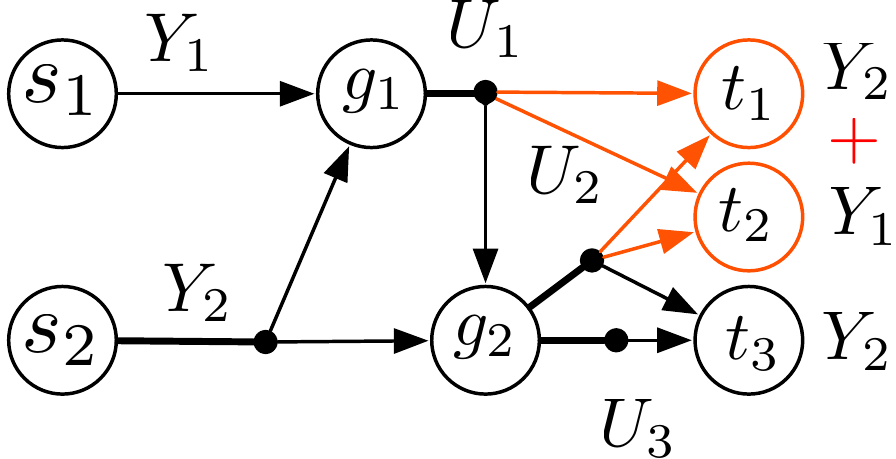} }\hspace{2mm}
\subfloat [\label{fig:C12exampleec}\ref{c12}: $t_1$ decodes $Y_2$ from $U_1$, hence $t_2$ also can decode $Y_2$, thus there is no need to list $Y_2$ in $\beta(t_2)$ .]{\includegraphics[scale=0.4]{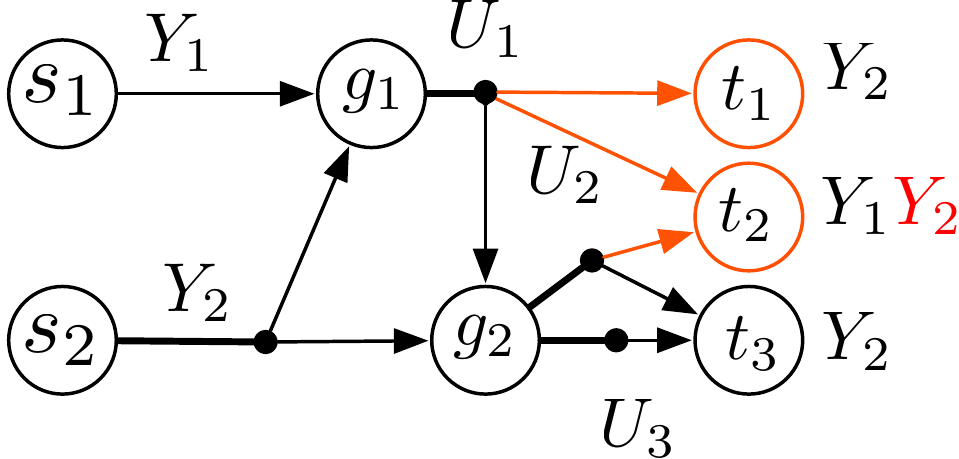} }\hspace{2mm}
\subfloat [\label{fig:C13exampleec}\ref{c13}: $t_1$ decodes $Y_2$ from $U_1$, thus $t_2$ also can decode $Y_2$, thus there is no need to keep direct access of $t_2$ to $Y_2$.]{\includegraphics[scale=0.4]{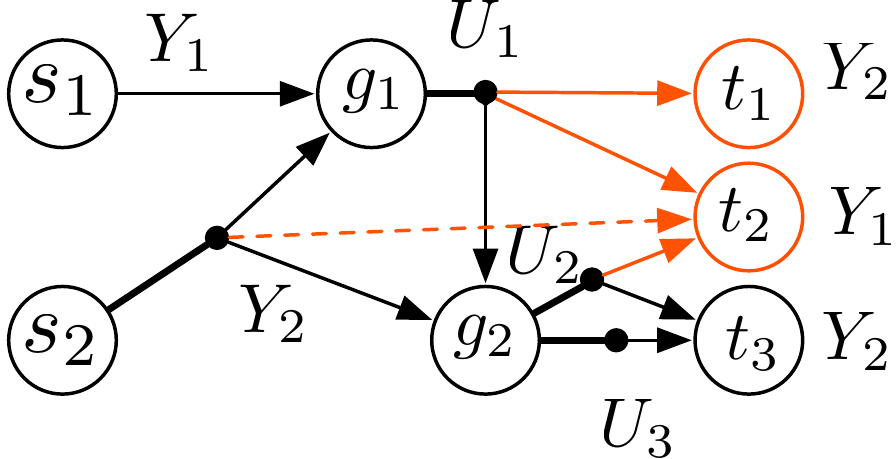} }\hspace{2mm}
\subfloat [\label{fig:C14exampleec}\ref{c14}: each connected component can be viewed as a separate network instance.]{\includegraphics[scale=0.35]{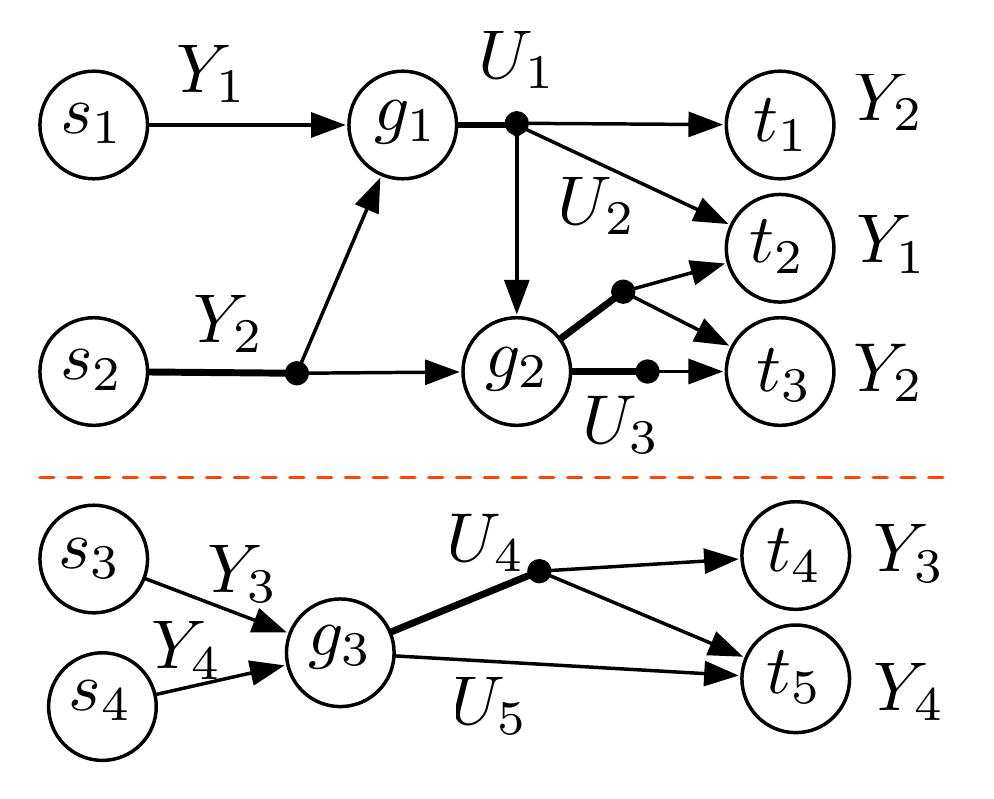} }\caption{Examples to demonstrate the minimality conditions \ref{c1}--\ref{c14}.}
\label{fig:minimalconditions} 
\end{figure}

\begin{theorem}\label{thm:minimality}
Suppose a network instance $\Asf=(\Smc,\Gmc,\Tmc,\Emc,\beta)$, with rate region and bounds $\Rmc_l(\Asf),\ l \in\{c,\ast,q,(s,q),o\}$, is a reduction from another network $\Asf'=(\Smc',\Gmc',\Tmc',\Emc',\beta')$, with rate region bounds $\Rmc_l(\Asf'), l \in \{c,\ast,q,(s,q),o\}$, by removing one of the redundancies specified in \ref{c1}--\ref{c14} in Def. \ref{def:minimal}.  Then, by defining $\Rbf_{\setminus \Amc} = \textrm{Proj}_{\setminus \Amc} \Rbf$ to be the projection of $\Rbf$ excluding coordinates associated with $\Amc$, and $\omega_s$ be the source rate of $s$, we have the following.
\begin{enumerate}
\item[] {\em Source minimality}:
\item [\namedlabel{D1}{(\textbf{D1})}] If $\exists s'\in\Smc',{\rm Hd}({\rm Out}(s')) \cap \Gmc' = \emptyset$, $\Asf$ will be $\Asf'$ with $s'$ removed and
\begin{equation}
\Rmc_{l}(\Asf') := \left\{\Rbf | \Rbf_{\setminus s'}\in \Rmc_l(\Asf),\ \omega_{s'} = 0  \right\} \quad \forall l \in\{c,\ast,q,(s,q),o \}
\end{equation}
if $\exists t' \in \mathcal{T}'$ such that $s \in \beta(t')$ and $s\notin \textrm{In}(t')$, while
\begin{equation}
\Rmc_{l}(\Asf') := \left\{\Rbf | \Rbf_{\setminus s}\in \Rmc_l(\Asf),\ \omega_s \geq 0 \right\} \quad \forall l \in\{c,\ast,q,(s,q),o \}
\end{equation}
otherwise.  Furthermore,
\begin{equation}
\Rmc_{l}(\Asf) = \textrm{Proj}_{\setminus s} \Rmc_{l}(\Asf') \quad \forall l \in\{c,\ast,q,(s,q),o \}.
\end{equation}

\item [\namedlabel{D2}{(\textbf{D2})}] if $\exists s'\in\Smc',t'\in\Tmc'$, such that $t'\in {\rm Hd}({\rm Out}(s'))$ and $s'\in \beta(t')$, $\Asf$ will be $\Asf'$ with ${\rm In}(t)={\rm In}(t')\setminus s'$ and $\beta(t)=\beta(t')\setminus s'$.   Further, $\Rmc_l(\Asf')=\Rmc_l(\Asf)$ for all $l \in\{c,\ast,q,(s,q),o \}$.

\item [\namedlabel{D3}{(\textbf{D3})}] if $\exists s'\in \Smc'$, such that $\forall\,t'\in \Tmc'$, $Y_{s'}\notin\beta(t')$, $\Asf$ will be $\Asf'$ with removal of the redundant source $s'$ and 
\begin{equation}
\Rmc_l(\Asf') := \left\{\Rbf | \Rbf_{\setminus s'} \in \Rmc_l(\Asf), H(Y_{s'}) \geq 0  \right\} \quad \forall l \in\{c,\ast,q,(s,q),o \}
\end{equation}
and
\begin{equation}
\Rmc_l(\Asf) = \textrm{Proj}_{\setminus s'} \Rmc_l(\Asf') \quad \forall l \in\{c,\ast,q,(s,q),o \}.
\end{equation}

\item [\namedlabel{D4}{(\textbf{D4})}] if $\exists s,s' \in\Smc'$ such that ${\rm Hd}({\rm Out}(s))={\rm Hd}({\rm Out}(s'))$ and $\gamma(s)=\gamma(s')$, $\Asf$ will be $\Asf'$ with sources $s,s'$ merged and 
\begin{equation}
\Rmc_l(\Asf')=  \left\{ \Rbf | [\Rbf_{\setminus\{s,s'\}}^T,H(Y_{s})+H(Y_{s'})]^T \in \Rmc_l(\Asf)  \right\} \quad \forall l \in\{c,\ast,q,o \}.
\end{equation}
i.e., replace $H(Y_{s})$ in $\Rmc_l(\Asf)$ with $H(Y_s)+H(Y_{s'})$ to get $\Rmc_l(\Asf')$, $l \in \{\ast,q,o\}$.  Furthermore,
\begin{equation}
\Rmc_l(\Asf) = \left\{ \Rbf_{\setminus\{s\}}\left| \Rbf \in \Rmc_l(\Asf'), \omega_{s'} = 0 \right. \right\} \quad \forall l \in\left\{c,\ast,q,(s,q),o \right\}.
\end{equation}

\item[] {\em Node minimality}:
\item [\namedlabel{D5}{(\textbf{D5})}] If $\exists\,k',j'\in \Gmc$ such that ${\rm In}(k')={\rm In}(j')$, $\Asf$ will be $\Asf'$ with $k',j'$ merged so that ${\rm In}(k)={\rm In}(k')={\rm In}(j'), {\rm Out}(k)={\rm Out}(k')\cup {\rm Out}(j')$, and $\Gmc=\Gmc'\setminus j'$.  Further, $\Rmc_l(\Asf)=\Rmc_l(\Asf')$ for all $l \in\{c,\ast,q,(s,q),o \}$.

\item [\namedlabel{D6}{(\textbf{D6})}] If $\exists g'\in\Gmc$ such that ${\rm In}(g)=\emptyset$, or ${\rm Out}(g)=\emptyset$, $\Asf$ will be $\Asf'$ with removal of the redundant node(s) $g'$ and $\Rmc_l(\Asf')=\Rmc_l(\Asf)$ for all $l \in\{c,\ast,q,(s,q),o \}$.

Similarly, if $\exists t'\in\Tmc$ such that ${\rm In}(t')=\emptyset$, $\Asf$ will be $\Asf'$ with removal of the redundant node(s) $t'$ and the deletion of any sources it demands, 
%\begin{equation}
$\Rmc_l(\Asf')= \left\{ \Rbf | \Rbf_{\setminus \beta(t')} \in \Rmc_l(\Asf), \omega_s = 0 \forall s \in \beta(t') \right\}$, and 
%\begin{equation}
$\Rmc_l(\Asf) = \textrm{Proj}_{\setminus \beta(t')} \Rmc_l(\Asf')
$%\end{equation} 
for all  $l \in\{c,\ast,q,(s,q),o \}$.

\item[] {\em Edge minimality}:
\item [\namedlabel{D7}{(\textbf{D7})}] If $\exists e'\in\Emc'$ such that $\text{Hd}(e')=\emptyset$, $\Asf$ will be $\Asf'$ with removal of edge $e'$,
\begin{equation}
\Rmc_l(\Asf')=\{\Rbf|\Rbf_{\setminus e'}\in \Rmc_l(\Asf),R_{e'}\geq 0\},
\end{equation}
and $\Rmc_l(\Asf) = \textrm{Proj}_{\setminus e'} \Rmc_l(\Asf')$ for all  $l \in\{c,\ast,q,(s,q),o \}$.

\item [\namedlabel{D8}{(\textbf{D8})}] If $\exists e,e'\in\Emc'$ with ${\rm Tl}(e)={\rm Tl}(e')$,  ${\rm Hd}(e)={\rm Hd}(e')$, $\Asf$ will be $\Asf'$ with edges $e,e'$ merged as $e$ and
\begin{equation}
\Rmc_l(\Asf')=  \left\{ \Rbf | [\Rbf_{\setminus\{e,e'\}}^T,R_e+R_{e'}]^T \in \Rmc_l(\Asf)  \right\}; \quad \forall l \in\{c,\ast,q,o \}
\end{equation}
i.e., replace $R_e$ in $\Rmc_*(\Asf)$ with $R_e +R_{e'}$ to get $\Rmc_*(\Asf')$.  Furthermore,
\begin{equation}
\Rmc_l(\Asf) = \left\{ [\Rbf_{\setminus\{e,e'\}}^T,R_e]^T | \Rbf \in \Rmc_l(\Asf'), R_{e'} = 0 \right\} \quad \forall l \in\{c,\ast,q,(s,q),o \}.
\end{equation}

\item [\namedlabel{D9}{(\textbf{D9})}]  If $\exists e,e'\in\Emc',g'\in\Gmc'$ such that ${\rm In}(g')=e$,  ${\rm Hd}(e)=g'$, ${\rm Out}(g')=e'$, then $\Asf$ will be $\Asf'$ with the node $g'$ removed and a new edge $e_{i}$ replacing $e,e'$ by directly connecting ${\rm Tl}(e)$ and ${\rm Hd}(e')$.  Further, 
\begin{equation}
\Rmc_l(\Asf')=  \left\{ \Rbf | [\Rbf_{\setminus\{e,e'\}}^T,\min \{R_e,R_{e'} \} ]^T \in \Rmc_l(\Asf)  \right\};
\end{equation}
i.e., replace $R_e$ in $\Rmc_l(\Asf)$ with $\min\{R_{e},R_{e'}\}$ to get $\Rmc_l(\Asf')$.  Accordingly,
\begin{equation}
\Rmc_l(\Asf ) = \left\{ [\Rbf_{\setminus\{e,e'\}}^T,\min \{R_e,R_{e'} \} ]^T \left| \Rbf \in \Rmc_l(\Asf') \right.\right\}.
\end{equation}
for all $l \in\{c,\ast,q,(s,q),o \}$.

\item[] {\em Sink minimality}:
\item [\namedlabel{D10}{(\textbf{D10})}] If $\exists t'\in\Tmc', s'\in\Smc'$, such that $s'\in\beta(t')$ but $s'\notin \sigma(t')$, then $\Asf$ will be $\Asf'$ with $s'$ deleted, 
\begin{equation}
\Rmc_l(\Asf')=\{\Rbf|\Rbf_{\setminus s'}\in\Rmc_l(\Asf), H(Y_{s'})=0\},
\end{equation}
and
\begin{equation}
\Rmc_l(\Asf) = \textrm{Proj}_{\setminus s'} \Rmc_l(\Asf')
\end{equation}
for all $l \in\{c,\ast,q,(s,q),o \}$.

\item [\namedlabel{D11}{(\textbf{D11})}] If $\exists t,t'\in\Tmc',t\neq t'$, such that $\mathrm{In}(t)=\mathrm{In}(t')$, then $\Asf$ will be $\Asf'$ with sinks $t,t'$ merged and $\Rmc_l(\Asf')=\Rmc_l(\Asf)$ for all $l \in\{c,\ast,q,(s,q),o \}$.

\item [\namedlabel{D12}{(\textbf{D12})}]  If $\exists t,t'$ such that $\mathrm{In}(t)\subseteq \mathrm{In}(t')$ and $\beta(t)\cap\beta(t')\neq\emptyset$, then $\Asf$ will be $\Asf'$ with removal of $\beta(t)\cap\beta(t')$ from $\beta(t')$ and $\Rmc_l (\Asf')=\Rmc_l(\Asf)$ for all $l \in\{c,\ast,q,(s,q),o \}$.
%\item [(\textbf{C13}.)] $\nexists t_i\in \Tmc,t_{\Jmc}\subseteq (\Tmc\setminus t_i)$, such that $\mathrm{In}(t_k)=\cup_{j\in\Jmc}\mathrm{In}(t_j)$ and $\beta(t_k)=\cup_{j\in\Jmc}\beta(t_j)$;
%\item [(\textbf{C13}.)] If $|\Smc|>1$, $\exists t\in\Tmc$ such that $\beta(t)\subsetneq \Smc$.
\item [\namedlabel{D13}{(\textbf{D13})}] If $\exists t,t',s'\in\beta(t)$ such that $\mathrm{In}(t)\subseteq \mathrm{In}(t')$ and $t'\in \textrm{Hd}(\textrm{Out}(s'))$, then $\Asf$ will be $\Asf'$ with removal of $s'$ from ${\rm In}(t')$ and $\Rmc_l(\Asf')=\Rmc_l(\Asf)$ for all $l \in\{c,\ast,q,(s,q),o \}$.

\item[] {\em Connectivity}:
\item [\namedlabel{D14}{(\textbf{D14})}] if $\Asf'$ is not weakly connected and $\Asf_1,\Asf_2$ are two weakly disconnected components, then $\Rmc_l(\Asf')=\Rmc_l(\Asf_1)\times \Rmc_l(\Asf_2)$ for all $l \in\{c,\ast,q,(s,q),o \}$.
\end{enumerate}
\end{theorem}
\begin{IEEEproof}
In the interest of conciseness, for all but \ref{D4} and \ref{D8} we will only briefly sketch the proof for the expressions determining $\Rmc_*(\Asf')$ from $\Rmc_*(\Asf)$, as the map in the opposite direction and the other rate region bounds follow directly from parallel arguments.

\ref{D1} holds because $s'$ is not communicating with any nodes other than possibly sinks.  If there is a sink that demands it that does not have direct access to it, then this sink can not successfully receive any information from it, since $s'$ does not communicate with any intermediate nodes.  Hence, in this case $\omega_{s'}=0$ and every other rate is constrained according to $\mathcal{R}_{\ast}(\Asf)$ because the remainder of the network has no interaction with $s'$.  Alternatively, if every sink that demands $s'$ has direct access to it, any non-negative source rate can be supported for $s'$, and the remainder of the network is constrained as by $\mathcal{R}_{\ast}(\Asf)$ because no other part of the network interacts with $s'$.

\ref{D2} holds because the demand of $s'$ at sink $t'$ is trivially satisfied if it has direct access to $s'$.  The constraint has no impact on the rate region of the network.

In \ref{D3} if a source is not demanded by anyone, it can trivially support any rate. 

When two sources have exactly the same connections and are demanded by same sinks as under \ref{D4}, they can be simply viewed as a combined source for $\Rmc_{l}$ with $l\in\{c,\ast,q,o\}$, since the exact region and these bounds enable simple concatenation of sources.  Since the source entropies are variables in the rate region expression, it is equivalent to make $s$
as the combined source, which since the previous sources were independent, will have an entropy which is the sum of their entropies.  Moving from $\Rmc_{l}(\Asf')$ to $\Rmc_{l}(\Asf)$ is then accomplished for any $l \in \{c,\ast,q,(s,q), o\}$ by observing that $\Asf$ can be viewed as $\Asf'$ with $\omega_{s'} = 0 $.

An intermediate node can only utilize its input hyperedges to produce its output hyperedges, hence when two intermediate nodes have the same input edges, their encoding capabilities are identical, and thus for pursuing minimality of representation of a network, these two nodes having the same input should be represented as one node.  Thus, \ref{D5} is necessary and the merge of nodes with same input does not impact the coding on edges or the rate region, as the associated constraints $\mathcal{L}_{\Asf}= \mathcal{L}_{\Asf'}$.

If the input or output of an intermediate node is empty, as in \ref{D6} it is incapable of affecting the capacity region.  If, as in the second case covered by \ref{D6} the input to an sink node is empty, any sources which it demands can only be reliably decoded if they have zero entropy.

\ref{D7} is clear because an edge to nowhere can not effect the rest of the capacity region and is effectively unconstrained itself.

\ref{D8} can be shown as follows.  If $\Rmc_*(\Asf)$ is known and when edge $e$ in $\Asf$ is represented as two parallel edges $e,e'$ so that the network becomes  $\Asf'$, then the constraint on $e,e'$ in $\Asf'$ is simply to make sure the total capacity $R_e+R_{e'}$ can allow the information to be transmitted from the tail node to head nodes.  Simple concatenation of the messages among the two edges will achieve this for those bounds $l \in\{c,\ast,q,o\}$ allowing such concatenation.  Therefore, replace the $R_e$ in $\Rmc_l(\Asf)$ with $R_e+R_{e'}$ will obtain the rate region $\Rmc_l(\Asf')$ for any $l \in\{c,\ast,q,o\}$.  Moving from $\Rmc_l(\Asf')$ to $\Rmc_l(\Asf)$ is accomplished by recognizing that $\Asf$ is effectively $\Asf'$ with $R_e' = 0$.

Under the condition in \ref{D9}, an intermediate node $g'$ has exactly one input edge $e$ and exactly one output hyperedge $e'$, and the input $e$ is an edge (i.e. $g'$ is its only destination).  The rate coming out of this node can be no larger than the rate coming in since the single output hyperedge must be a deterministic function of the input edge.  It suffices to treat these two edges as one hyperedge connecting the tail of $e$ to the head of $e'$ with the rate the minimum of the rates on the two links.

If a sink demands a source that it does not have access to, the only way to satisfy this network constraint is the source entropy is $0$.  Hence, \ref{D10} holds.  The removal of this redundant source does not impact the rate region of the network with remaining variables.

\ref{D11}, similar to \ref{D3}, observes that two sink nodes with same input yield the exact same constraints $\mathcal{L}_{\Asf'}$ as $\mathcal{L}_{\Asf}$ with the two sink nodes merged. 

\ref{D12} is easy to understand because the decoding ability of $\beta(t)$ at sink node $t$ is implied by sink $t'$.  The non-necessary repeated decoding constraints will not affect the rate region for this network.  

\ref{D13}, similar to \ref{D12}, observes that the ability of $t$ to decode $s'$ implies that $t'$ can decode it as well, and hence, adding or removing the direct access to $s'$ at $t'$ will not affect the rate region. 

\ref{D14} is obviously true since the weakly disconnected components can not influence each others rate regions. 
\end{IEEEproof}

Fig.\,\ref{fig:minimalconditions} contains examples illustrating these reductions.  In general, we can define a minimality operator $\Asf=\mathrm{minimal}(\Asf')$ on networks, which checks the minimality conditions \ref{c1}--\ref{c14} on $\Asf'$ one by one, in the order \ref{c1}, \ref{c2}, \ref{c6}, \ref{c5}, \ref{c3}, \ref{c4}, \ref{c7}--\ref{c14}.  If any of the conditions encountered is not satisfied, the network is immediately reduced it according to the associated reduction in Theorem \ref{thm:minimality}, and the resulting reduced network is checked again for minimality by starting again at condition \ref{c1}, if needed, until all minimality conditions are satisfied.  Furthermore, define the associated rate region operator $\Rmc_*(\Asf')=\mathrm{minimal}_{\Asf' \leftarrow \Asf}(\Rmc_*(\Asf))$ which moves through each of the reduction steps applied by $\mathrm{minimal}(\Asf')$ to the network $\Asf'$ in reverse order, utilizing the expression for the rate region change under each reduction, thereby obtaining the rate region of $\Asf'$ from $\Asf$.  Accordingly, let $\Rmc_*(\Asf)=\mathrm{minimal}_{\Asf' \rightarrow \Asf}(\Rmc_*(\Asf'))$ be the rate region operator which moves through each of the reduction steps applied by $\mathrm{minimal}(\Asf')$ to the network $\Asf'$ in order, utilizing the expression for the rate region change under each reduction, thereby obtaining the rate region of $\Asf$ from $\Asf'$.  This network minimality operator and its associated rate region operators will come in use later in the paper.  However, we next discuss the enumeration of minimal networks of a particular size.

\section{Enumeration of Non-isomorphic Minimal Networks}\label{sec:enumeration}

\begin{figure}
\centering
\captionsetup{justification=centering}
\includegraphics[scale=0.5]{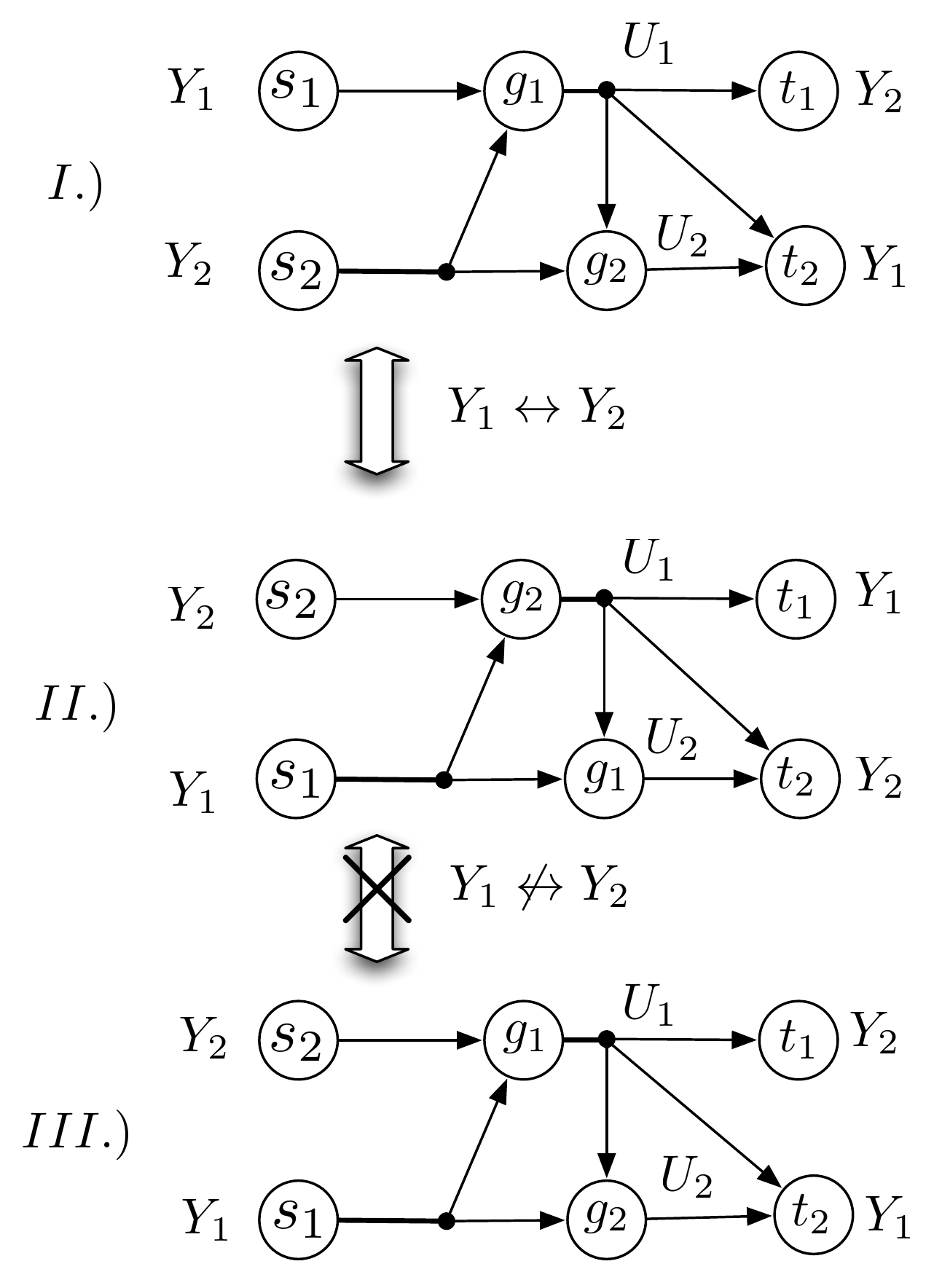}
\caption{\label{fig:isodemo} A demonstration of the equivalence between network coding problems via isomorphism: networks $I$ and $II$ are equivalent because $II$ can be obtained by permuting $Y_1,Y_2$.  However, network $III$ is not equivalent to $I$ or $II$, because the demands at the sinks do not reflect the same permutation of $Y_1,Y_2$ as is necessary on the source side.}
\end{figure}

Even though the notion of network minimality (\S \ref{sec:minimality}) reduces the set of network coding problem instances by removing parts of a network coding problem which are inessential, much more needs to be done to group network coding problem instances into appropriate equivalence classes.  Although we have to use label sets to describe the edges and sources in order to specify a network coding problem instance (identifying a certain source as source number one, another as source number two, and so on), it is clear that the essence of the underlying network coding problem is insensitive to these labels.  For instance, it is intuitively clear that the first two problems depicted in Fig. \ref{fig:isodemo} should be equivalent even though their labeled descriptions differ, while the third problem should not be considered equivalent to the first two.

In a certain sense, having to label network coding problems in order to completely specify them obstructs our ability to work efficiently with a class of problems.  This is because one unlabeled network coding problem equivalence class typically consists of many labeled network coding problems.  In principle, we could go about investigating the unlabeled problems by exhaustively listing labeled network coding problem obeying the minimality constraints, testing for equivalence under relabeling of the source and node or edge indices, and grouping them together into equivalence classes.   %A na\"{i}ve such network coding problem enumeration algorithm is described in Appendix \ref{sec:sperner}.  
%\begin{figure}
%\centering
%\captionsetup{justification=centering}
%\includegraphics[scale=0.5]{EquivalentClass.pdf}
%%\vspace{-0.2cm}
%\caption{\label{fig:eqclass} A demonstration of equivalent classes of all networks of a particular size.}
%%\vspace{-0.6cm}
%\end{figure}

%\subsection{Enumeration of Networks Based on Leiterspiel Algorithm}\label{sec:leiterspielAlg}
However, listing networks by generating all variants of the labeled encoding  becomes infeasible rapidly as the problem grows because of the large number of labeled networks in each equivalence class.
As a more feasible alternative, it is desirable to find a method for directly  cataloguing all (unlabeled) network coding problem equivalence classes by generating exactly one representative from each equivalence class directly, without isomorphism (equivalence) testing.

In order to develop such a method, and to explain the connection between its solution and other isomorphism-free exhaustive generation problems of a similar ilk, in this section we first formalize a concise method of encoding a network coding problem instance in \S\ref{net:probInstance}.  With this encoding in hand, in \S \ref{sec:netEquivGroupAct}, the notion of equivalence classes for network coding problem instances can be made precise as orbits in this labeled problem space under an appropriate group action.  The generic algorithm \emph{Leiterspiel} \cite{Schmalz_Bayreuth_92,betten2006error}, for computing orbits within the power set of subsets of some set $\mathcal{X}$ on which a group $\groupG$ acts, can then be applied, together with some other standard orbit computation techniques in computational group theory \cite{Permlib,GAP}, in order to provide the desired non-isomorphic network coding problem list generation method in \S\ref{net:nonisoalg}.

\subsection{Encoding a Network Coding Problem}\label{net:probInstance}
Though, as is consistent with the network coding literature, we have thus far utilized a tuple $\Asf=(\Smc,\Gmc,\Tmc,\Emc,\beta)$ to represent a network instance, this encoding proves to be insufficiently parsimonious to enable easy identification of equivalence classes.  As will be discussed later, the commonly used node representation of a network, a key component of the $\Asf=(\Smc,\Gmc,\Tmc,\Emc,\beta)$ encoding, unnecessarily increases the complexity of enumeration.  
Hence, we represent a network instance in an alternate way for enumeration.  Specifically, a network instance with $K$ sources and $L$ edges that obeys the minimality conditions (\textbf{C1}-\textbf{C14}) is encoded as an ordered pair $(\edgeDefSet,\sinkDefSet)$ consisting of a set $\edgeDefSet$ of edge definitions $\edgeDefSet \subseteq \{(i,\mathcal{A}) | i \in \{K+1,\ldots,K+L\},$ $\ \mathcal{A} \subseteq \{1,\ldots,K+L\} \setminus \{i\}, \ |\mathcal{A} | > 0 \}$, and a set $\sinkDefSet$ of sink definitions $\sinkDefSet \subseteq \left\{ (i,\mathcal{A}) | i \in \{1,\ldots, K\}, \ \mathcal{A} \subseteq \{1,\ldots,K+L\} \setminus \{i\} \right\}$.  
Here, the sources are associated with labels $\{1,...,K\}$ and the edges are associated with labels $\{K+1,\ldots,K+L\}$.  Each $(i,\mathcal{A}) \in \edgeDefSet$ indicates that the edge $i\in \Emc_U$ is encoded exclusively from the sources and edges in $\mathcal{A}$, and hence represents the information that $\mathcal{A}=\textrm{In}(\textrm{Tl}(i))$.  Furthermore, each sink definition $(i,\mathcal{A}) \in \sinkDefSet$ represents the information that there is a sink node whose inputs are $\mathcal{A}$ and which decodes source $i$ as its output.  Note that there are $L$ non-source edges in the network, each of which must have some input according to condition \ref{c6}.  We additionally have the requirement that $|\edgeDefSet| = L$, and, to ensure that no edge is multiply defined, we must have that if $(i,\mathcal{A})$ and $(i',\mathcal{A'})$ are two different elements in $\edgeDefSet$, then $i\neq i'$.  As the same source may be decoded at multiple sinks, there is no such requirement for $\sinkDefSet$.

As is illustrated in Figures \ref{fig:isoexamplewosymmetry} and \ref{fig:isoexamplewsymmetry}, this edge-based definition of the directed hypergraph included in a network coding problem instance can provide a more parsimonious representation than a node-based representation, and as every edge in the network for a network coding problem is associated with a random variable, this representation maps more easily to the entropic constraints than the node representation of the directed acyclic hypergraph does.  Additionally it is beneficial because it is guaranteed to obey several of the key minimality constraints.  In particular, the representation ensures that there are no redundant nodes \ref{c5}, \ref{c11}, since the intermediate nodes are associated directly the elements of the set $\{ \mathcal{A} | \ \exists i,\ (i,\mathcal{A}) \in \edgeDefSet \}$ and the sink nodes are associated directly with $\{ \mathcal{A} | \ \exists i,\ (i,\mathcal{A}) \in \sinkDefSet \}$.  Representing $\edgeDefSet$ as a set (rather than a multi-set) also ensures that \ref{c8} is always obeyed, since such a parallel edge would be a repeated element in $\edgeDefSet$.

\subsection{Expressing Network Equivalence with a Group Action}  \label{sec:netEquivGroupAct}
Another benefit of the representation of the network coding problem as the ordered pair $(\edgeDefSet,\sinkDefSet)$ is that it enables the notion of network isomorphism to be appropriately defined.  In particular, let $\groupG := S_{\{1,2,\ldots,K\}} \times S_{\{K+1,\ldots,K+L\}}$ be the direct product of the symmetric group of all permutations of the set $\{1,2,\ldots,K\}$ of source indices and the symmetric group of all permutations of the set $\{K+1,\ldots,K+L\}$ of edge indices.  The group $\groupG$ acts in a natural manner on the elements of the sets $\edgeDefSet,\sinkDefSet$ of edge and sink definitions.  In particular, let $\pi \in \groupG$ be a permutation in $\groupG$, then the group action maps
\begin{equation}
\pi ( (i,\mathcal{A}) ) \mapsto (\pi(i),\pi(\mathcal{A}))
\end{equation}
with the usual interpretation that $\pi(\mathcal{A}) = \{ \pi (j) | j\in\mathcal{A}\}$.  This action extends to an action on the sets $\edgeDefSet$ and $\sinkDefSet$ in the natural manner
\begin{equation}
\pi( \edgeDefSet ) \mapsto \left\{ \pi( (i,\mathcal{A}) ) | (i,\mathcal{A}) \in \edgeDefSet \right\}.
\end{equation}
This action then extends further still to an action on the network $(\edgeDefSet,\sinkDefSet)$ via
\begin{equation}
\pi ( (\edgeDefSet,\sinkDefSet ) ) = (\pi (\edgeDefSet), \pi(\sinkDefSet) ).
\end{equation}

Two networks $(\edgeDefSet_1,\sinkDefSet_1)$ and $(\edgeDefSet_2,\sinkDefSet_2)$ are said to be \emph{isomorphic}, or in the same equivalence class, if there is some permutation of $ \pi \in \groupG$ such that $\pi((\edgeDefSet_1,\sinkDefSet_1)) = (\edgeDefSet_2,\sinkDefSet_2)$.  In the language of group actions, two such pairs are isomorphic if they are in the same orbit under the group action, i.e. if $(\edgeDefSet_2,\sinkDefSet_2) \in \left\{ \pi((\edgeDefSet_1,\sinkDefSet_1)) \left| \pi \in \groupG \right. \right\}=: \orbit_{(\edgeDefSet_1,\sinkDefSet_1)}$.  In other words, the equivalence classes of networks are identified with the orbits  in the set of all valid minimal problem description pairs $(\edgeDefSet,\sinkDefSet)$ under the action of $\groupG$. 

We elect to represent each equivalence class with its \emph{canonical network}, which is the element in each orbit that is least in a lexicographic sense.  Note that this lexicographic (i.e., dictionary) order is well-defined, as we can compare two subsets $\mathcal{A}$ and $\mathcal{A}'$ by viewing their members in increasing order (under the usual ordering of the integers $\{1,\ldots,L+K\}$) and lexicographically comparing them.  This then implies that we can lexicographically order the ordered pairs $(i,\mathcal{A})$ according to $(i,\mathcal{A}) > (j,\mathcal{A}')$ if $j< i$ or $i=j$ and $\mathcal{A}' < \mathcal{A}$ under this lexicographic ordering.  Since the elements of $\edgeDefSet$ and $\sinkDefSet$ are of the form $(i,\mathcal{A})$, this in turn means that they can be ordered in increasing order, and then also lexicographically compared, enabling comparison of two edge definition sets $\edgeDefSet$ and $\edgeDefSet'$ or two sink definition sets $\sinkDefSet$ and $\sinkDefSet'$.  Finally, one can then use these orderings to define the lexicographic order on the network ordered pairs $(\edgeDefSet,\sinkDefSet)$.  The element in an orbit $\orbit_{(\edgeDefSet,\sinkDefSet)}$ which is minimal under this lexicographic ordering will be the canonical representative for the orbit.

A key basic result in the theory of group actions, the \emph{Orbit Stabilizer Theorem}, states that the number of elements in an orbit, which in our problem is the number of networks that are isomorphic to a given network, is equal to the ratio of the size of the acting group $\groupG$ and its stabilizer subgroup $\groupG_{(\edgeDefSet,\sinkDefSet)}$ of any element selected from the orbit:
\begin{equation}
\left| \left\{\pi((\edgeDefSet,\sinkDefSet)) \left| \pi \in \groupG \right. \right\} \right| = \left| \orbit_{(\edgeDefSet,\sinkDefSet)} \right| =\frac{\left|\groupG\right|}{\left|\groupG_{(\edgeDefSet,\sinkDefSet)}\right|}, \quad 
\groupG_{(\edgeDefSet,\sinkDefSet)} := \left\{ \pi \in \groupG \left| \pi((\edgeDefSet,\sinkDefSet))=(\edgeDefSet,\sinkDefSet) \right. \right\}
\end{equation}
Note that, because it leaves the sets of edges, decoder demands, and topology constraints set-wise invariant, the elements of the stabilizer subgroup $\groupG_{(\edgeDefSet,\sinkDefSet)}$ also leave the set of rate region constraints (\ref{eq:rrcondef1}), (\ref{eq:rrcondef2}), (\ref{eq:Lratefree}), (\ref{eq:rrcondef4}) invariant.  Such a group of permutations on sources and edges is called the \emph{network symmetry group}, and is the subject of a separate investigation \cite{JayantISIT2015,JayantNetCod2015}.  This network symmetry group plays a role in the present study because, as depicted in Figures  \ref{fig:isoexamplewosymmetry} and \ref{fig:isoexamplewsymmetry}, by the orbit stabilizer theorem mentioned above, it determines the number of networks equivalent to a given canonical network (the representative we will select from the orbit).

In particular, Fig. \ref{fig:isoexamplewosymmetry} shows the orbit of a $(2,2)$ network $(\edgeDefSet,\sinkDefSet)$ whose stabilizer subgroup (i.e., network symmetry group) is simply the identity, and hence has only one element. In this instance, the number of isomorphic labeled network coding problems in this equivalence class is then $|\groupG|=\left| \symmGroup_{\{1,2\} } \times \symmGroup_{\{3,4\}} \right| = 4$ in the edge representation, as shown at the left.  Even this tiny example demonstrates well the benefits of encoding a network coding problem via the more parsimonious representation $(\edgeDefSet,\sinkDefSet)$ vs. the encoding via the node representation hypergraph $(\Vmc,\Emc)$ and the sink demands $\beta(\cdot)$.  Namely, because the size of the group acting on the node representation is $|\symmGroup_{\{a,b\}} \times \symmGroup_{\{c,d,e,f,g \} } | = 240$, and, as the stabilizer subgroup in the node representation has the same order ($1$), the number of isomorphic networks represented in the node based representation is $240$.

By contrast, Fig. \ref{fig:isoexamplewsymmetry} shows the orbit of a $(2,2)$ network $(\edgeDefSet,\sinkDefSet)$ whose stabilizer subgroup (i.e., network symmetry group) is the largest possible among $(2,2)$ networks, and has order $4$.  In this instance, the number of isomorphic labeled network coding problems in this equivalence class is $\frac{|\groupG|}{|\groupG_{(\edgeDefSet,\sinkDefSet)} |}= 1$ in the edge representation.  The stabilizer subgroup in the node representation has generators $\langle\{ (a,b)(d,f)(e,g)\}, \{(d,e)(f,g)\} \rangle$, which has the same order of $4$, and hence there are $\frac{240}{4} =60$ isomorphic network coding problems to this one in the node representation.

\begin{figure}
\centering
\captionsetup{justification=centering}
\subfloat [\label{fig:isoexamplewosymmetry} All isomorphisms of a $(2,2)$ network with empty symmetry group]{\includegraphics[width=.9\textwidth]{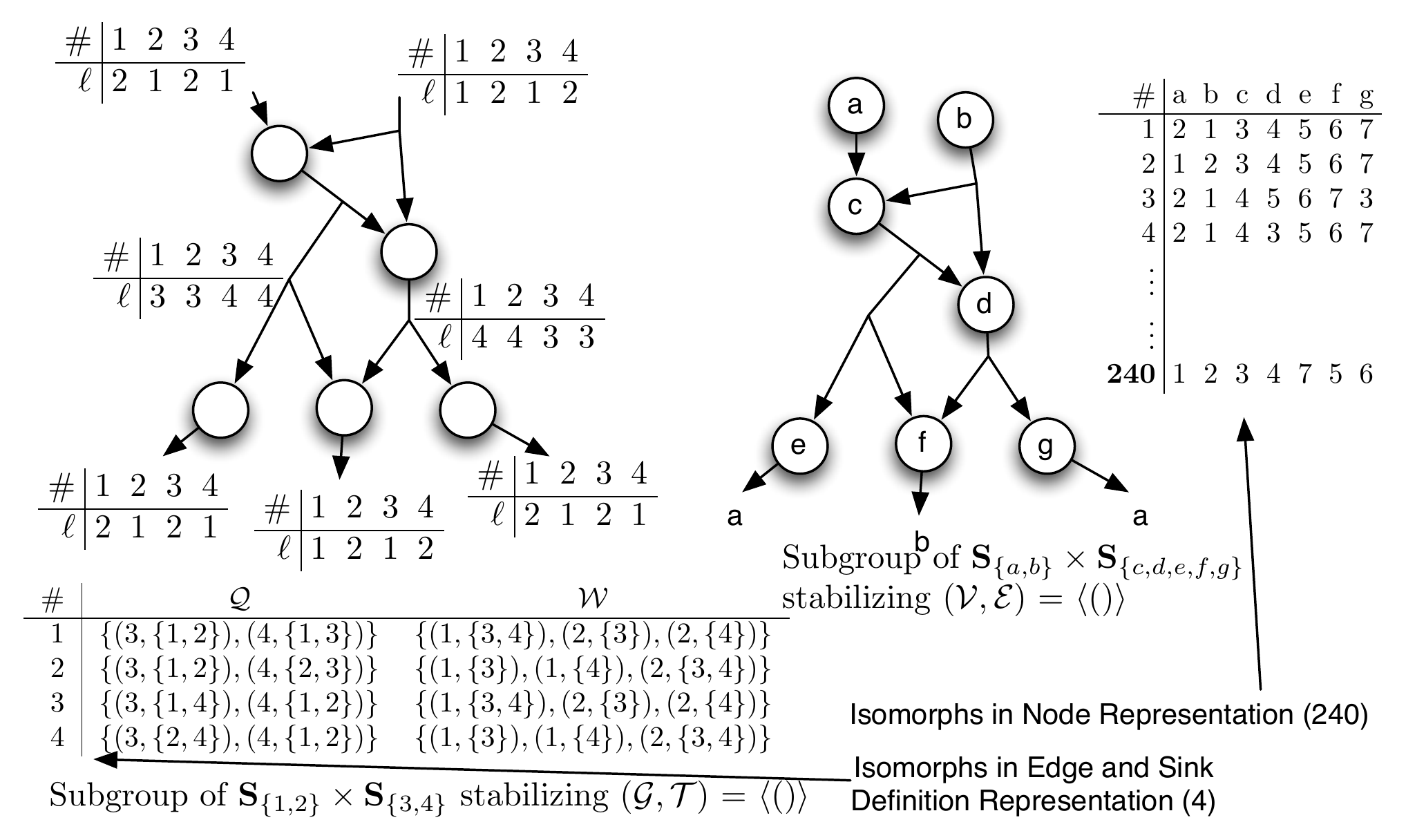}}

\subfloat [\label{fig:isoexamplewsymmetry} All isomorphisms of a $(2,2)$ network with full symmetry group]{\includegraphics[width=.9\textwidth]{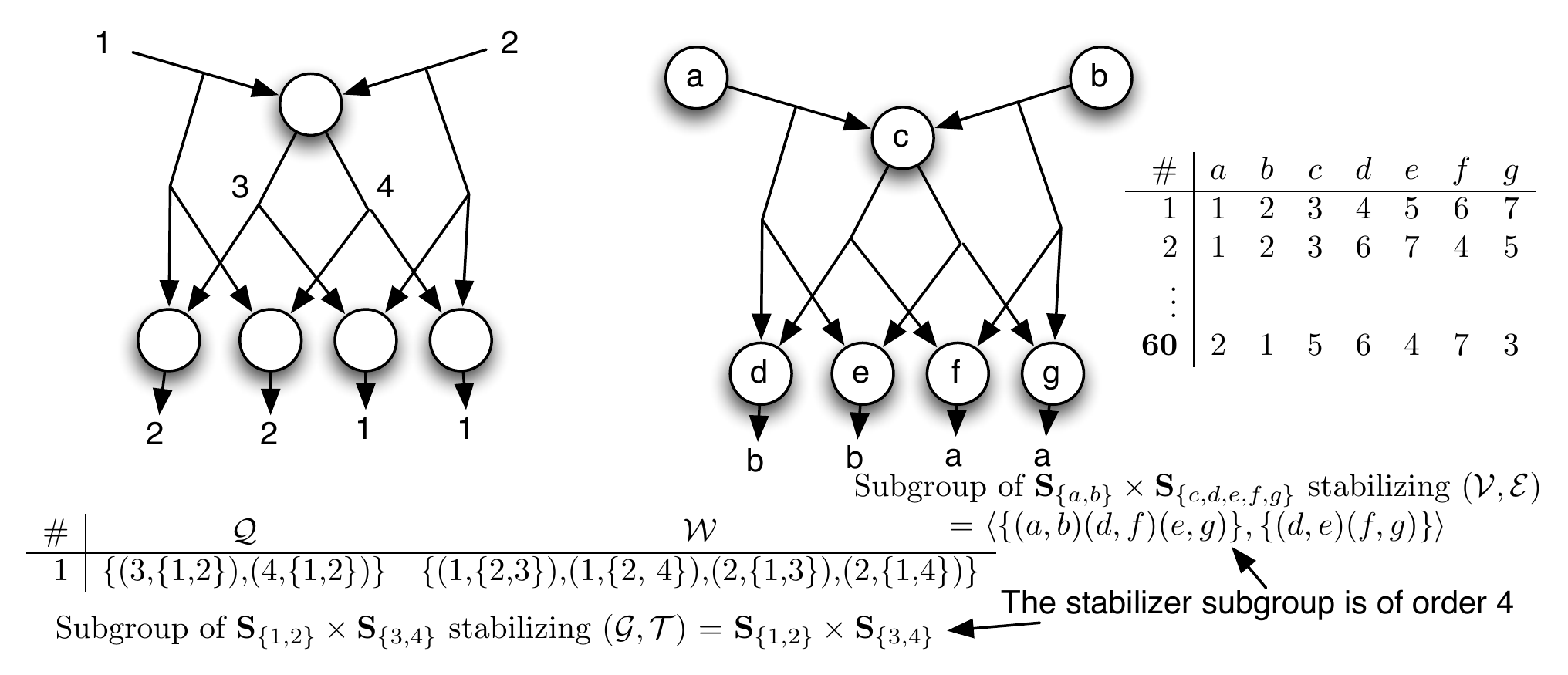}}
\caption{ Examples of $(2,2)$ networks with all edge isomorphisms (left) and all node isomorphisms (right).  The instance indices are marked by \# and the labels are marked by $l$.}
\end{figure}

%With a firm understanding of posing equivalence classes as orbits under group actions on appropriate parameterizations of the set of minimal network coding problems, as well as an appropriate notion of ordering and canonical representatives from these orbits, we are now ready to describe a computational group theory based method of enumerating non-isomorphic network coding problem instances.

\subsection{Network Enumeration/Listing Algorithm}\label{net:nonisoalg}
Formalizing the notion of a canonical network via group actions on the set of minimal $(\edgeDefSet,\sinkDefSet)$ pairs enables one to partly develop a method for directly listing canonical networks based on techniques from computational group theory.

To solve this problem we can harness the algorithm \emph{Leiterspiel}, loosely translated \emph{snakes and ladders} \cite{betten2006error,Schmalz_Bayreuth_92}, which, given an algorithm for computing canonical representatives of orbits, i.e., transversal, on some finite set $\mathcal{X}$ under a group $\groupG$ and its subgroups, provides a method for computing the orbits on the power set $\mathcal{P}_i(\mathcal{X}) = \left\{ \mathcal{B} \subseteq \mathcal{X} |\ |\mathcal{B}| = i \right\}$ of subsets from $\mathcal{X}$ of cardinality $i$, incrementally in $i$.  In fact, the algorithm can also list directly only those canonical representatives of orbits for which some test function $f$ returns $1$, provided that the test function has the property that any subset of a set with $f=1$ also has $f=1$.  This test function is useful for only listing those subsets in $\mathcal{P}_i(\Xmc)$ with a desired set of properties, provided these properties are inherited by subsets of a superset with that property.

To see how to apply and modify Leiterspiel for network coding problem enumeration, let $\mathcal{X}$ be the set of possible edge definitions
\begin{equation}
\mathcal{X} := \left\{ (i,\mathcal{A}) \left| i \in\{K+1,\ldots,K+L\}, \ \mathcal{A} \subseteq \{1,\ldots,K+L\} \setminus \{i\} \right.\right\}
\end{equation}

For small to moderately sized networks, the orbits in $\mathcal{X}$ from $\groupG$ and its subgroups can be readily computed with modern computational group theory packages such as GAP \cite{GAP} or PERMLIB \cite{Permlib}.  Leiterspiel can be applied to first calculate the non-isomorphic candidates for the edge definition set $\edgeDefSet$, as it is a subset of $\mathcal{X}$ with cardinality $L$ obeying certain conditions associated with the definition of a network coding problem and its minimality (c.f. \textbf{C1}--\textbf{C14}).  Next, for each non-isomorphic edge-definition $\edgeDefSet$, a list of non-isomorphic sink-definitions $\mathcal{A}$, also constrained to obey problem definition and minimality conditions (\textbf{C1}--\textbf{C14}), can be created with a second application of Leiterspiel.  The pseudo-code for the resulting generation/enumeration is provided in Alg. \ref{alg:networkEnumeration}

\begin{algorithm}
\SetAlgoLined
 \KwIn{number of sources $K$, number of non-source edges $L$} 
 \KwOut{All non-isomorphic network instances $\Zmc$}
 \textbf{Initialization:} $\Zmc=\emptyset$\;
Let $\Xmc:=\left\{ (i,\mathcal{A}) \left| i \in\{K+1,\ldots,K+L\}, \ \mathcal{A} \subseteq \{1,\ldots,K+L\} \setminus \{i\} \right.\right\}$\;
 Let $f_1$ be the condition that $\nexists (i,\Amc),(i',\Amc')$ such that $i=i'$\;
 Let $f_2$ be the condition of acyclicity\;
Let acting group $\groupG:=\symmGroup_{\{1,\ldots,K\}}\times \symmGroup_{\{K+1,\ldots,K+L\}}$\;
Call Leiterspiel algorithm to incrementally get all candidate transversal up to $L$:
 $T_{L}=Leiterspiel(\groupG,\Pmc_{L}^{f_1,f_2}(\Xmc))$\;
 \For{each $\edgeDefSet\in T_{L}$}{
 \If{$\edgeDefSet$ obeys (\textbf{C1})}{
 Let $\mathcal{X}' := \{ (i,\mathcal{A}) |i\in\{1,\ldots,K\}, \Amc\subseteq \{1,\ldots,K+L\}\setminus \{i\}, \exists$ a directed path in  $\edgeDefSet$  from $i$ to at least one edge in $\mathcal{A}\}$\;
 Let $f'_1$ be the condition (\textbf{C12})\;
 Let $f'_2$ be the condition (\textbf{C13})\;
Let acting group $\groupG:=\symmGroup_{\{1,\ldots,K\}}\times \symmGroup_{\{K+1,\ldots,K+L\}}$\;
Call Leiterspiel algorithm to incrementally get all candidate canonical representatives, i.e., transversals, up to no new element can be added obeying (\textbf{C12},\textbf{C13}):
 $T_{K}=Leiterspiel(\groupG,\Pmc_{K}^{f'_1,f'_2}(\Xmc'))$\;
 \For{each $\sinkDefSet\in T_{K}$}{
 \If{$(\edgeDefSet,\sinkDefSet)$ obeys (\textbf{C3}--\textbf{C7}) and (\textbf{C14}) }{
$\Zmc=\Zmc\cup (\edgeDefSet,\sinkDefSet)$\; 
 }
 }
 }
}

\caption{Enumerate all non-isomorphic $(K,L)$ networks using Leiterspiel algorithm.}
\label{alg:networkEnumeration}
\end{algorithm}

As outlined above, in the first stage of the enumeration/generation algorithm, Leiterspiel is applied to grow subsets from $\mathcal{X}$ of size $i$ incrementally in $i$ until $i=L$.  Some of the network conditions have the appropriate inheritance properties, and hence can be incorporated as constraints into the constraint function $f$ in the Leiterspiel process. These include
\begin{itemize}
\item \textbf{no repeated edge definitions}: If $\mathcal{B}\subseteq \mathcal{C} \subseteq \mathcal{X}$ and $\mathcal{C}$ has the property that no two of its edge definitions $(i,\mathcal{A})$ and $(i',\mathcal{A}')$ have $i=i'$, then so does $\mathcal{B}$.  Hence, the constraint function $f$ in the first application of Leiterspiel incorporates checks to ensure that no two edge definitions in the candidate subset define the same edge.
\item \textbf{acyclicity}: If $\mathcal{B}\subseteq \mathcal{C} \subseteq \mathcal{X}$ and $\mathcal{C}$ is associated with an acyclic hyper graph, then so is $\mathcal{B}$.  Hence, the constraint function $f$ in the first application of Leiterspiel checks to determine if the subset in question is acyclic.
\end{itemize}
At the end of this first Leiterspiel process, some more canonical edge definition sets $\edgeDefSet$ can be ruled as non-minimal owing to (\textbf{C1}), requiring that each source appears in the definition of at least one edge variable.

For each member of the resulting narrowed list of canonical edge definition sets $\edgeDefSet$, we must then build a list of canonical representative sink definitions $\sinkDefSet$.  This is done by first creating the ($\edgeDefSet$-dependent) set of \emph{valid} sink definitions
\begin{equation}
\mathcal{X}' := \left\{ (i,\mathcal{A}) \left| \exists \ \textrm{a directed path in } \edgeDefSet\ \textrm{ from}\ i \ \textrm{to at least one edge in }\ \mathcal{A} \right.\right\}
\end{equation}
which are crafted to obey the minimality conditions (\textbf{C10}) that the created sink (defined by its input which is the set of sources and edges in $\mathcal{A}$ in the sink definition $(i,\mathcal{A})$) must have at least one path in the hyper graph defined by $\edgeDefSet$ to the source $i$ it is demanding, and (\textbf{C2}) that is can not have a direct connection to the source it is demanding.

Leiterspiel is then applied to determine canonical (lexicographically minimal) representatives of sink definition sets $\sinkDefSet$, utilizing the associated stabilizer of the canonical edge definition set $\edgeDefSet$ being extended as the group, with the test function $f$ handling the minimality conditions \ref{c12} and \ref{c13}, during the iterations.
%\begin{itemize}
%\item (\textbf{C12}): No two sink definitions $(i,\mathcal{A})$ and $(i',\mathcal{A}')$ with $i=i'$ can have $\mathcal{A} \subset \mathcal{A}'$.  Clearly, this property is inherited by a subset of a sink definition $\sinkDefSet$ for a minimal network, and hence it can be incorporated in the constraint function $f$ in the Leiterspiel generation.
%\item (\textbf{C13}): No sink may have direct access to a source which it would be capable of decoding due its containment within its in-edges of the in-edges of a second sink.  This property is also inherited over subsets, and hence can be incorporated in the constraint function $f$ in Leiterspiel.
%\end{itemize}

This second application of Leiterspiel to determine the list of canonical sink definition sets $\sinkDefSet$ for each canonical edge definition set $\edgeDefSet$ does not have a definite cap on the cardinality of each of the canonical sink definition sets $\sinkDefSet$.  Rather, subsets of all sizes are determined incrementally until there is no longer any canonical subset that can obey the constraint function associated with (\textbf{C12}) and (\textbf{C13}).  Each of the candidate canonical sink definition sets $\sinkDefSet$ (of all different cardinalities) are then tested together with $\edgeDefSet$ with the remaining conditions, which do not have the inheritance property necessary for incorporation as constraints earlier in the two stages of Leiterspiel processing.
% SW: why are these repeated here? 
%\begin{itemize}
%\item \ref{c3}: each source is demanded by at least one sink.
%\item \ref{c4}: no two sources are available to exactly the same edges and sinks and demanded by exactly the same sinks.
%\item \ref{c7}: Every edge variable has at least one other edge variable or sink that is dependent on it.
%\item \ref{c8}: No parallel edges.
%\item \ref{c9}: No redundant nodes with exactly one directed (non-hyper) edge in and one edge out.
%\item \ref{c14}: The final network coding problem (including both the sink and edge definitions) must yield a weakly connected graph.
%\end{itemize}

Any pair of canonical $(\edgeDefSet,\sinkDefSet)$ surviving each of these checks is then added to the list of canonical minimal non-isomorphic network coding problem instances.

An additional pleasant side effect of the enumeration is that the stabilizer subgroups, i.e., the network symmetry groups \cite{JayantISIT2015}, are directly provided by the second Leiterspiel.  Harnessing these network symmetry groups provides a powerful technique to reduce the complex process of calculating the rate region for a network coding problem instance \cite{JayantNetCod2015}.

Although this method directly generates the canonical representatives from the network coding problem equivalence classes without ever listing other isomorphs within these classes, one can also use the stabilizer subgroups provided by Leiterspiel to directly enumerate the sizes of these equivalence classes of $(\edgeDefSet,\sinkDefSet)$ pairs, as described above via the orbit stabilizer theorem.  Experiments summarized in Table \ref{tab:networkEnum} show that the number of isomorphic cases is substantially larger than the number of canonical representives/equivalence classes, and hence the extra effort to directly list only canonical networks is worthwhile.  It is also worth noting that a node representation, utilizing a node based encoding of the hyper edges, would yield a substantially higher number of isomorphs.

\subsection{Modification to Other Problem Types}
A final point worth noting is that this algorithm is readily modified to handle listing canonical representatives of special network coding problem families contained within our general model, as described in \S \ref{sec:specialclasses}.  For instance, IDSC problems can be enumerated by simply defining $\edgeDefSet$ to have each edge access all of the sources and no other edges, then continuing with the subsequent sink enumeration process.  It is also easily adapted to enumerate only directed edges and match the more restrictive constraints described in the original Yan, Yeung, and Zhang \cite{YanYeungTranIT2012} rate region paper.

\subsection{Enumeration Results for Networks with Different Sizes}
By using our enumeration tool with an implementation of the algorithms above, we obtained the list of canonical minimal network instances for different network coding problem sizes with $N=K+L\leq 5$.
%$K=1,2,3$ and $L=1,2,3$ excepting $(3,3)$.  
While the whole list is available \cite{CongduanTranIT2015data}, we give the numbers of network problem instances in Table \ref{tab:networkEnum}, where $|\Zmc|, |\hat{\Zmc}|,|\hat{\Zmc}_n|$ represent the number of canonical network coding problems (i.e., the number of equivalence classes), the number of edge descriptions of network coding problems including symmetries/equivalences, and the number of node descriptions of network coding problems including the symmetries/equivalences, respectively.  As we can see from the table, the number of possibilities in the node representation of the network coding problems explodes very quickly, with the more than 2 trillion labeled node network coding problems covered by the study only necessitating a list of consisting of roughly 750,000 equivalence classes of network coding problems.  That said, it is also important to note that the number of non-isomorphic network instances increases exponentially fast as network size grows.  For instance, the number of non-isomorphic general network instances grows from $333$ to $485,890$ (roughly, an increase of about $1500$ times), when the network size grows from $(2,2)$ to $(2,3)$.  To provide an illustration of the variety of networks that are encountered, Fig. \ref{fig:all22networks} depicts all $46$ of the $333$ canonical minimal network coding problems of size $(2,2)$ obeying the extra constraint that no sink has direct access to a source. 

\begin{figure}
\centering \includegraphics[width=\textwidth]{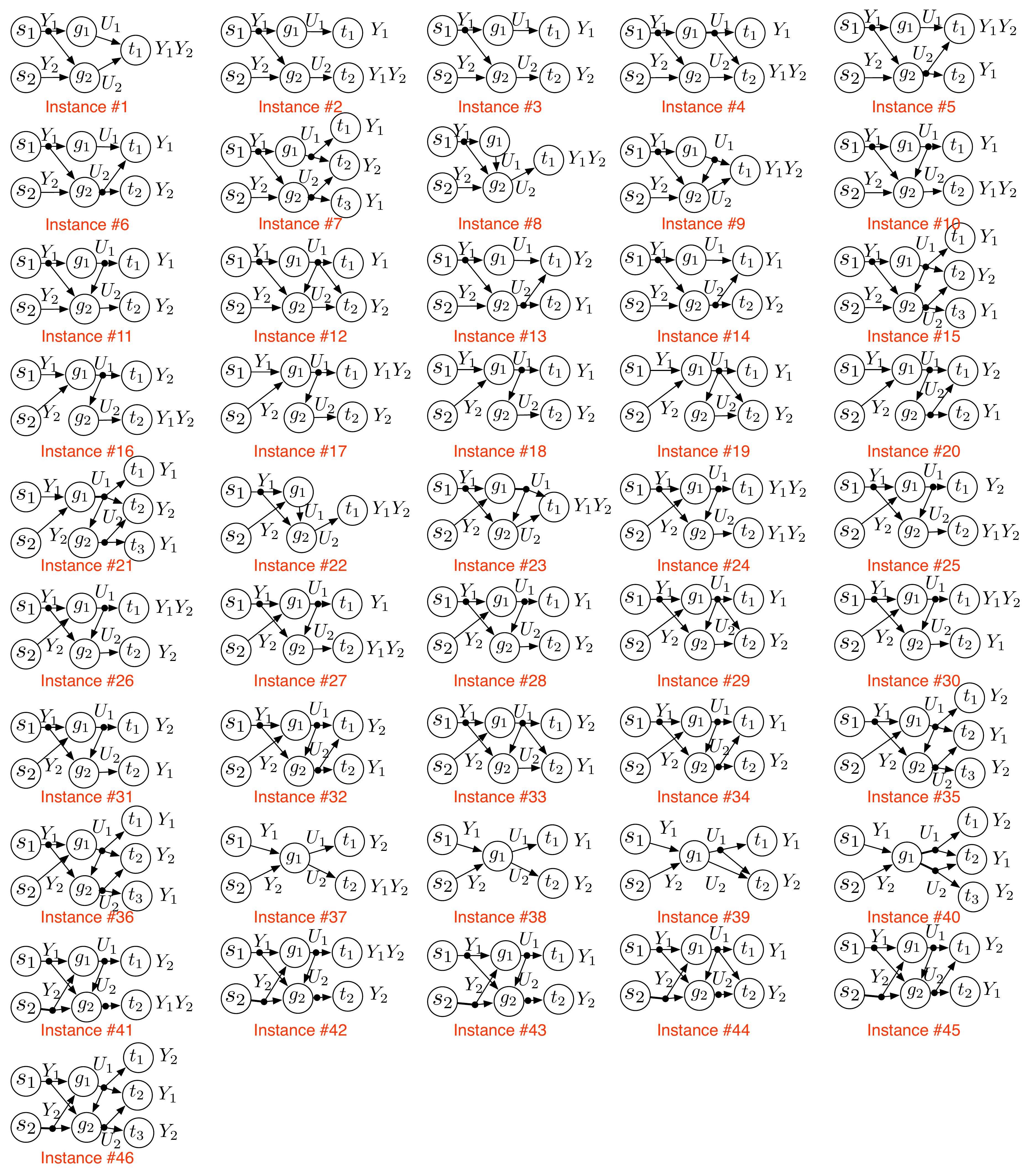} \caption{All 46 non-isomorphic network instances of $(2,2)$ networks with the constraint that sinks do not have direct access to sources.}
\label{fig:all22networks} 
\end{figure}

\begin{table}
\caption{Number of network coding problems of different sizes: $|\Zmc|$ represents the number of non-isomorphic networks, $|\hat{\Zmc}|$ represents the number of isomorphic networks with edge isomorphism, and $|\hat{\Zmc}_n|$ represents the number of isomorphic networks with node isomorphism.}\label{tab:networkEnum}
\centering
%\begin{tabular}{|c|r|r|r|r|r|r|r|r|r|r|}
%\hline
%$(K,L)$ & (1,2) & (1,3) & (1,4)& (2,1) & (2,2) & (2,3) & (3,1) &  (3,2) &(4,1) & Total\\
%\hline
%$|\Zmc|$                & 4 & 132 & 18027& 1 & 333 & 485\ 890 &  9  & 239\ 187 & 536& 744\ 119 \\ \hline
%$|\hat{\Zmc}|$     & 7 & 749 & 420948& 1 & 1\ 270 & 5\ 787\ 074 &  31  & 2\ 829\ 932 &10478& 9\ 050\ 490 \\ \hline
%$|\hat{\Zmc}_n|$ & 39 & 18\ 401& 600\ 067\ 643&6  & 163\ 800 & 2\ 204\ 574\ 267\ 764 & 582 & 176\ 437\ 964\ 418 &12\ 149\ 472& 2\ 381\ 624\ 632\ 119 \\
%\hline
%\end{tabular}
% present numbers vertically
\begin{tabular}{|c|c|c|c|}
\hline
$(K,L)$ & $|\Zmc|$  & $|\hat{\Zmc}|$ &$|\hat{\Zmc}_n|$\\ \hline
(1,2) & 4 &7 &39  \\ \hline
(1,3) & 132 & 749 & 18\ 401 \\ \hline
(1,4)& 18027& 420948& 600\ 067\ 643\\ \hline
(2,1) & 1 &1 & 6 \\ \hline
(2,2) & 333 &1\ 270 & 163\ 800 \\ \hline
(2,3) & 485\ 890 & 5\ 787\ 074 & 2\ 204\ 574\ 267\ 764\\ \hline
(3,1) & 9  & 31  &582\\ \hline
 (3,2) & 239\ 187 & 2\ 829\ 932 &176\ 437\ 964\ 418\\ \hline
(4,1) &  536 & 10478& 12\ 149\ 472\\ \hline
Total & 744\ 119& 9\ 050\ 490 & 2\ 381\ 624\ 632\ 119\\ \hline
\end{tabular}

\end{table}

As a special class of hyperedge multi-source network coding problems, it is easier to enumerate IDSC networks, defined in Example \ref{ex:IDSC} in \S\ref{sec:model}.  Since we assume that all encoders in IDSC have access to all sources, we only need to consider the configurations at the decoders, which additionally are only afforded access to edges from intermediate nodes.  These extra constraints are easily incorporated into Algorithm \ref{alg:networkEnumeration} by removing the edge definitions, restricting to the unique one associated with the IDSC problems, and enumerating exclusively the sink definitions.

We give the enumeration results for $K=2,3$ and $L=2,3$ IDSC networks in Table \ref{IDSClist}, while the full list is available in \cite{CongduanIDSCfile}.  From the table we see that, even for this special type of network, the number of non-isomorphic instances grows very quickly.  For instance, the number of non-isomorphic IDSC instances grows from $33$ to $179$ (roughly, a factor of 6 increase), when the network size grows from $(2,3)$ to $(3,3)$.

\begin{table}
\caption{\label{IDSClist}List of numbers of IDSC configurations.  $|\Zmc_n'|$ is the number of configurations in the node representation including isomorphisms, $|\Zmc'|$ is the number of configurations in the edge representation including isomorphisms, and $|\Zmc|$ is the number of all non-isomorphic configurations.}\vspace{-.3cm}
\begin{center}
\begin{tabular}{|c|c|c|c|c|c|c|}
\hline
$(K,L)$& \multicolumn{3}{c|}{$2$} & \multicolumn{3}{c|}{$3$} \\ \cline{2-7}
 & $|\Zmc_n'|$ & $|\Zmc'|$ & $|\Zmc|$ &  $|\Zmc_n'|$& $|\Zmc'|$ & $|\Zmc|$\\ \hline
$2$ & 54 & 12 & 4 &4970 & 234& 33  \\ \hline
$3$ & 234 & 24& 3 & 443130 & 4752 & 179\\ \hline
\end{tabular}
\end{center}
\end{table}

\section{Rate Region Results for Small Networks}\label{sec:resultssmall}
With the list of minimal canonical network coding problems provided by the algorithm in the previous section in hand, the next step in our computational agenda was to determine each of their rate regions with computational tools.  In this section, we describe a database we have created which contains the exact regions of all general networks with sizes $N=K+L \leq 5$ and all IDSC networks with sizes $K=2,3$ and $L=2,3$.

\subsection{Database of Rate Regions for all networks of size $N=K+L\leq 5$}

\begin{table}
\caption{\label{tab:resultsgeneral}Sufficiency of codes for network instances: Columns 3--8 show the number of instances that the rate region inner bounds match with the Shannon outer bound. } 
\begin{center}
%\begin{tabular}{| p{0.8cm}| p{1.5cm} | p{1.5cm} | p{1.5cm} | p{1.5cm} | p{1.5cm} |p{1.5cm}|p{1.5cm} |}
\begin{tabular}{|c|c|c|c|c|c|c|c|}
\hline
$(K,L)$ &$|\Zmc|$& $\Rmc_{s,2}(\Asf)$ & $\Rmc_{2}^{N+1}(\Asf)$ & $\Rmc_{2}^{N+2}(\Asf)$ & $\Rmc_{2}^{N+3}(\Asf)$& $\Rmc_{2}^{N+4}(\Asf)$ & $\Rmc_{{\rm linear}}^N$\\ \hline
$(1,2)$ & 4 & 4 & 4 & 4 & 4& 4 & 4\\ \hline
$(1,3)$ & 132 & 122 & 132 & 132 & 132 & 132 & 132\\ \hline
$(1,4)$ & 18027 & 13386 &16930 & 17697 & 17928 & 17928& 18027\\ \hline
$(2,1)$ & 1 & 1 & 1 & 1 & 1& 1 & 1\\ \hline
$(2,2)$ & 333 & 301 & 319  & 323 &323 & 333 & 333\\ \hline
$(2,3)$ & 485890 & 341406 & 403883 & 432872 & 434545 & -- & 485890 \\ \hline
$(3,1)$ & 9 & 4 & 4 & 9 & 9 & 9 & 9\\ \hline
$(3,2)$ & 239187 & 118133 & 168761 & 202130 & 211417& -- & 239187\\ \hline
$(4,1)$ & 536 & 99 & 230 & 235 & 476& 476& 536\\ \hline
Total: & 744119 & 473456 & 590264 & 653403 & 664835 & -- & 744119\\ \hline
%$(3,3)$ & ? & 78648 & 81439 & 81803  \\ \hline
\end{tabular}
\end{center}
\end{table}

We begin by describing the experimental results we obtained by running our rate region computation software on all general hyperedge network instances of size $N=K+L \leq 5$.   These problems consist of $744,119$ canonical minimal networks, representing $9,050,490$ networks in the edge $(\mathcal{Q},\mathcal{W})$ encoding and $2\ 381\ 624\ 632\ 119$ networks in the standard node representation, as indicated in Table \ref{tab:networkEnum}.  For each non-isomorphic network instance, we calculated several bounds on its rate region: the Shannon outer bound $\Rmc_o$,  the scalar binary representable matroid inner bound $\Rmc_{s,2}$, the vector binary representable matroid inner bounds $\Rmc_{2}^{N+1},\ldots,\Rmc_{2}^{N+4}$, and linear inner bound $\Rmc_{{\rm linear}}^N$.  As indicated in \S\ref{sec:model}, if the outer bound on the rate region matches with an inner bound, we not only obtain the exact rate region, but also know the codes that suffice to achieve any point in it.  The general code constructions from representable matroids follow a similar process in \cite{CongduanNetCod2013,CongduanTranIT2014}, where rate regions and achieving codes are investigated for MDCS.

Though it is infeasible to list each of the $744,119$  rate regions in this paper, a summary of results on the matches of various bounds is shown in Table \ref{tab:resultsgeneral}.  The full list of rate region bounds can be obtained at \cite{CongduanTranIT2015data} and can be re-derived using \cite{EntVecSoft}.

Several key observations kay be made from Table \ref{tab:resultsgeneral}.  First of all the Shannon outer bound is proved to be tight for all networks of size $N=K+L \leq 5$.  Additionally, the results show that linear codes are sufficient to exhaust the entire capacity region for all of them, as indicated in column $2$ and $8$ in Table \ref{tab:resultsgeneral}.  Furthermore, we investigate the number of networks whose rate regions are achievable by simple linear codes, e.g., binary codes (columns 3--7 in Table \ref{tab:resultsgeneral}), and find that simple binary codes are capable of exhausting most of the capacity regions.  

For all $(1,2)$ and $(2,1)$ networks, scalar binary codes suffice.  However, this is not true in general even when there are only one or two edge variables. For example, there are some instances in $(3,1)$, $(4,1)$, $(2,2)$ and $(3,2)$ networks for which scalar binary codes do not suffice.   As we can see from  Table \ref{tab:resultsgeneral}, as the vector binary inner bounds get tighter and tighter (i.e., as we move to the right in Table \ref{tab:resultsgeneral}), the exact rate region is established for more and more instances.  That is, with tighter and tighter binary inner bounds, more and more instances are found for which binary codes suffice.  

In order to provide a sample of the sorts of results available in the database \cite{CongduanTranIT2015data}, the following example shows the various inner bounds on the rate region of a representative $(3,3)$ problem. 
\begin{example} \label{ex:33network} % the example is case 77775 in (3,3) networks w/o direct access between sources and sinks
%results.topology(:,:,results.topologyindexwoperm(77775))

%ans =

%     1     1     1     0     0     0
%     1     0     1     1     0     0
 %    0     0     1     1     1     0
% results.finalwoperm(:,:,77775)

%ans =

%     0     0     0     0
%     0     0     0     0
%     0     0     0     0
%     0     0     0     0
%     0     0     0     3
%     0     0     5     6
%     0     0     0     0
A 3-source 3-encoder hyperedge network instance $\Asf$ with block diagram and rate region $\Rmc_*(\Asf)$ shown in Fig.\ \ref{fig:33example}.%\end{example}
%\vspace{-.3cm}
\begin{figure}
\centering \includegraphics[scale=.6]{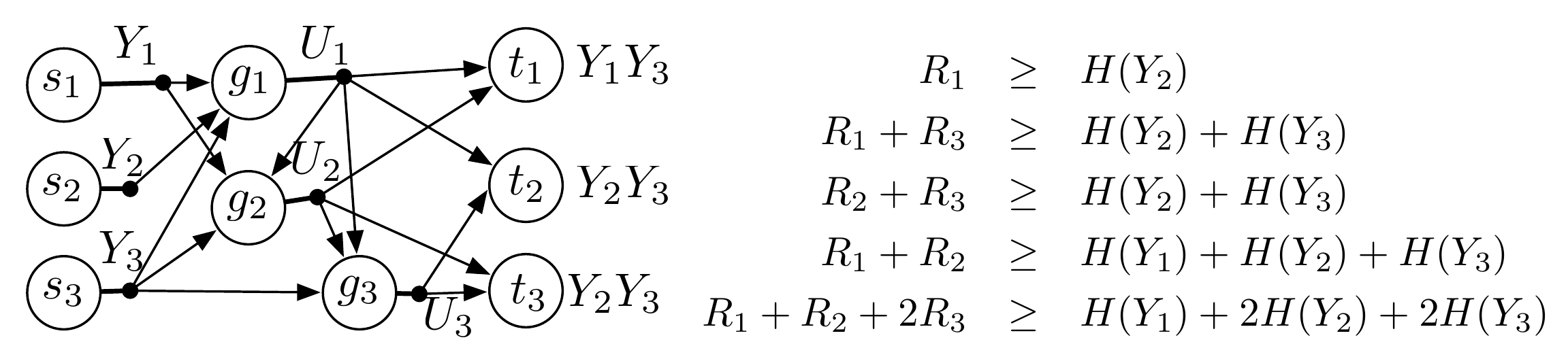} \caption{ Block diagram and rate region $\Rmc_*(\Asf)$ for the $(3,3)$ network instance $\Asf$ in Example \ref{ex:33network}.}
\label{fig:33example} 
\end{figure}

\begin{figure}
\centering \includegraphics[scale=.36]{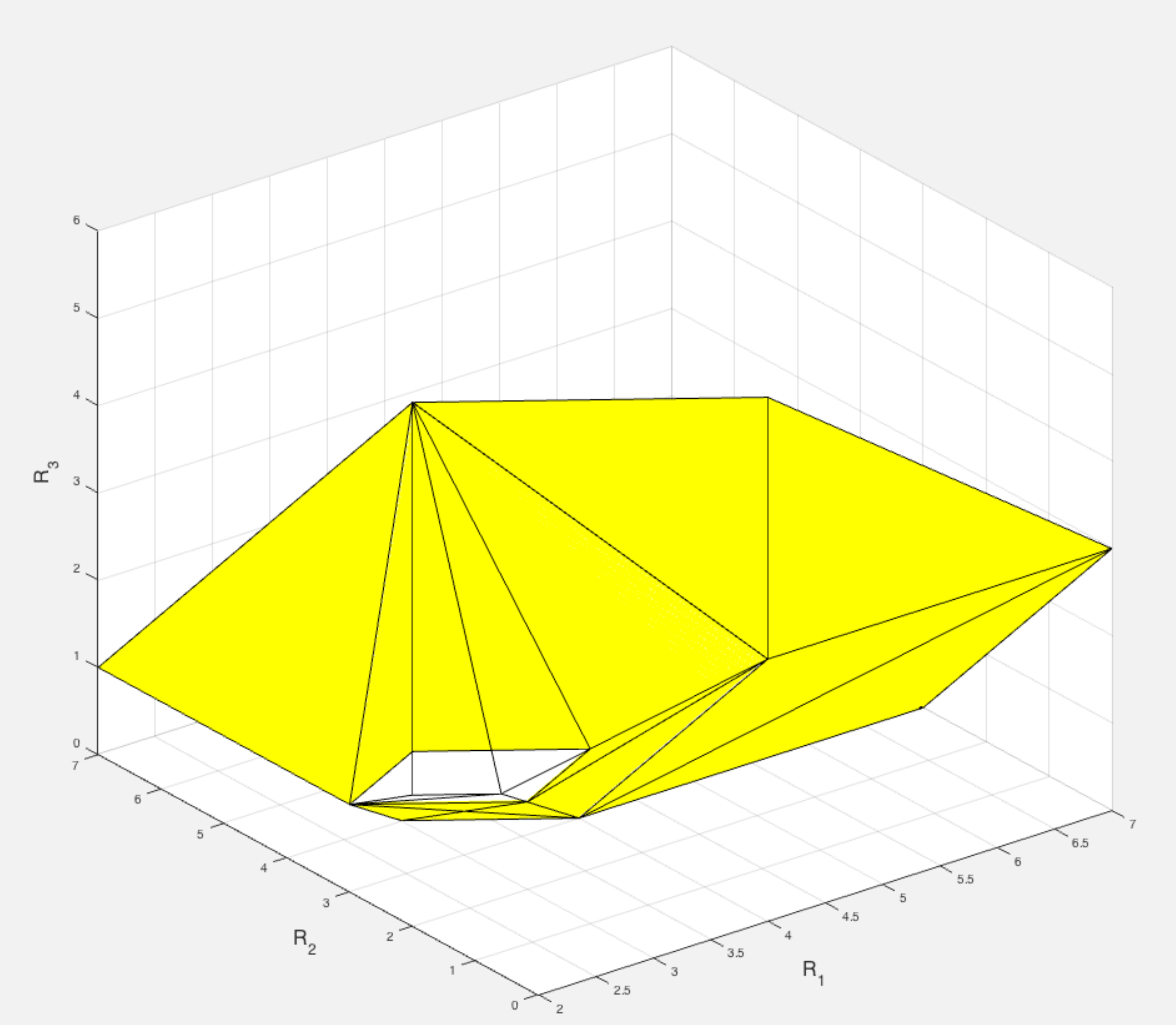} \caption{Comparison of rate regions $\Rmc_*(\Asf)$ (which equals to $\Rmc_o(\Asf)$) and $\Rmc_{2}^7(\Asf)$ for the $(3,3)$ network instance $\Asf$ in Example \ref{ex:33network}, when source entropies are $(H(Y_1),H(Y_2),H(Y_3))=(1,2,1)$ and the cone is capped by $R_1+R_2+R_3\leq 10$: the white part is the portion that scalar binary codes cannot achieve.  The ratio of  $\Rmc_{2}^7(\Asf)$ over $\Rmc_*(\Asf)$ is about $99.57\%$ for this choice of $(H(Y_1),H(Y_2),H(Y_3))$.}
\label{fig:33networkexample_sbo} 
\end{figure}

\begin{figure}
\centering \includegraphics[scale=.4]{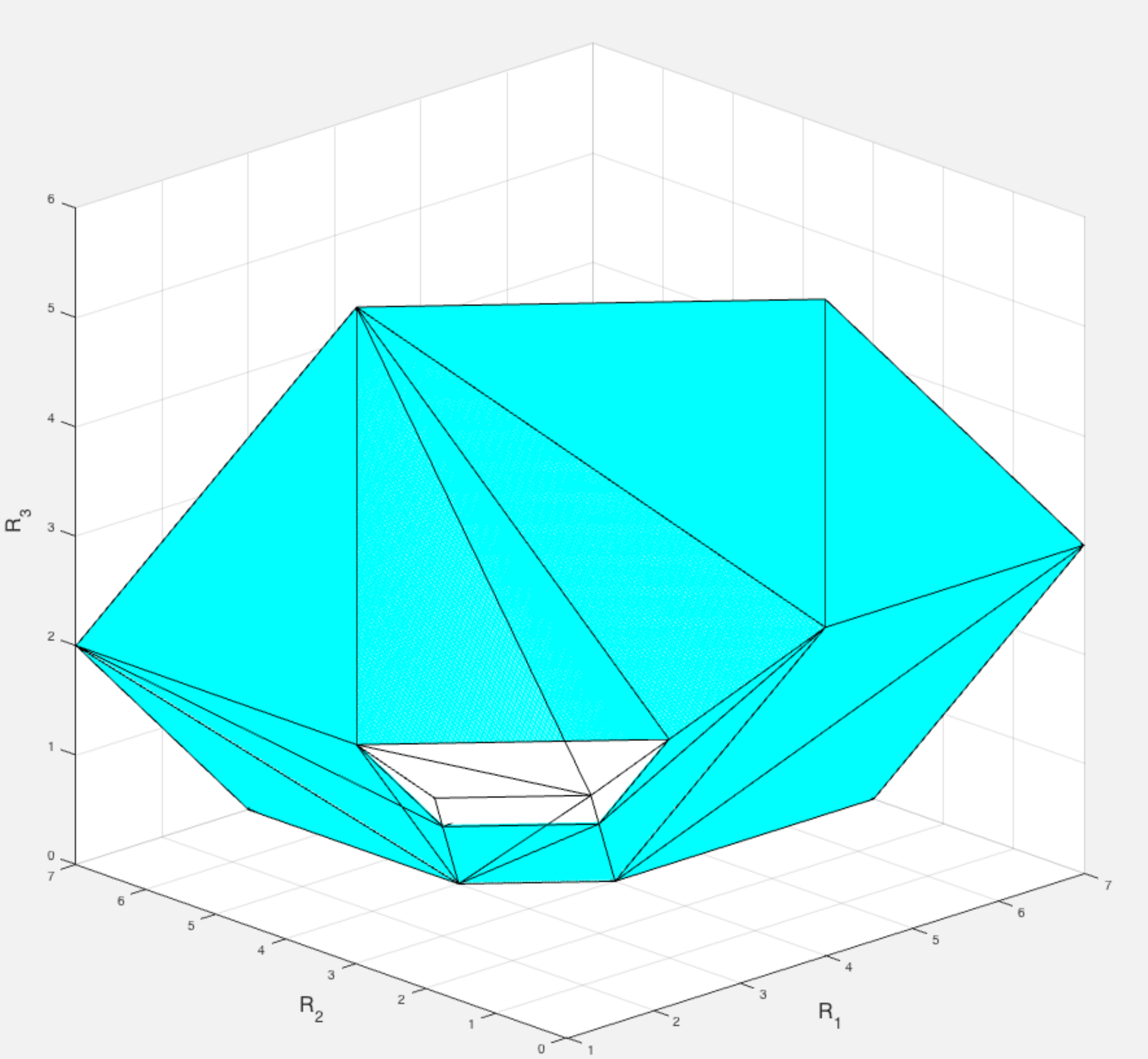} \caption{Comparison of rate regions $\Rmc_{2}^7(\Asf)$ and $\Rmc_{s,2}(\Asf)$ for the $(3,3)$ network instance $\Asf$ in Example \ref{ex:33network} when source entropies are $[H(Y_1),H(Y_2),H(Y_3)]=[1,1,2]$ and the cone is capped by $R_1+R_2+R_3\leq 10$: the white part is the portion that scalar binary codes cannot achieve.  The ratio of  $\Rmc_{s,2}(\Asf)$ over $\Rmc_{2}^7(\Asf)$ is about $99.41\%$ for this choice of $[H(Y_1),H(Y_2),H(Y_3)]$.}
\label{fig:33networkexample_sbvb} 
\end{figure}

First, scalar binary codes do not suffice for this network. The scalar binary coding rate region is
\begin{equation}
\Rmc_{s,2}=\Rmc_*(\Asf)\cap\left\{
\begin{array}{c}
R_1+R_2+R_3  \geq  H(Y_1)+2H(Y_2)+H(Y_3)\\
R_1+R_2+R_3  \geq  H(Y_1)+H(Y_2)+2H(Y_3)
\end{array}\right\}.
\end{equation}
One of the extreme rays in the Shannon outer bound on rate region is $[R_1,R_2,R_3,H(Y_1),H(Y_2),H(Y_3)]=[1,1,1,0,0,2]$.  This extreme ray cannot be achieved by scalar binary codes because no scalar code can encode a source with entropy of two into a variable with entropy of at most one.  Fig.\,\ref{fig:33networkexample_sbo} illustrates the gap between $\Rmc_*(\Asf)$ and $\Rmc_{s,2}(\Asf)$ with a particular source entropy assignment.  When source entropies are $[H(Y_1),H(Y_2),H(Y_3)]=[1,2,1]$ and the cone is capped by $R_1+R_2+R_3\leq 10$, there is a clear gap between the two polytopes, though the inner bound occupies more than $99\%$ of the exact rate region for this choice of $[H(Y_1),H(Y_2),H(Y_3)]$.  

Second, vector binary codes from $7$ bits do not suffice for this network either.  The vector binary coding rate region is 
\begin{equation}
\Rmc_{2}^7=\Rmc_*(\Asf)\cap\left\{
R_1+R_2+R_3  \geq  H(Y_1)+2H(Y_2)+H(Y_3)
\right\}.
\end{equation}
One of the extreme rays in the Shannon outer bound is $[R_1,R_2,R_3,H(Y_1),H(Y_2),H(Y_3)]=[2,1,1,1,2,0]$.  This extreme ray cannot be achieved by binary codes from $7$ bits because the empty source $Y_3$ takes one bit as well when we assign bits to variables in general.  Hence, at least $8$ bits are necessary (2+1+1+1+2+1=8), as will be shown later.  In our inner bound calculation, every variable in the network needs to have at least one associated element from the representable matroid, even though its entropy can be zero, like $Y_3$ in this case.  Though this inner bound is still loose in the sense of matching with the exact rate region, it is tighter than the scalar binary inner bound $\Rmc_{s,2}(\Asf)$.  This is illustrated in Fig.\,\ref{fig:33networkexample_sbvb} by choosing a particular source entropy tuple.  When source entropies are $[H(Y_1),H(Y_2),H(Y_3)]=[1,1,2]$ and the cone is capped by $R_1+R_2+R_3\leq 10$, there is a clear gap between the two polytopes, though the scalar inner bound takes more than $99\%$ space of the tighter vector binary inner bound for this choice of $[H(Y_1),H(Y_2),H(Y_3)]$.

However, vector binary codes from $8$ bits suffice for this network and thus $\Rmc_2^{8}(\Asf)=\Rmc_*(\Asf)$. One can construct vector binary codes to achieve all extreme rays in the Shannon outer bound on the rate region. For instance, the extreme ray $[R_1,R_2,R_3,H(X),H(Y),H(Z)]=[2,1,1,1,2,0]$ can be achieved by the vector binary code as follows: $U_1=[Y_1+Y_2^{(2)},\ Y_2^{(1)}],U_2=Y_2^{(1)}+Y_2^{(2)},U_3=Y_2^{(2)}$, where $Y_2^{(1)},Y_2^{(2)}$ are the two bits in source $Y_2$.

\end{example}

\subsection{Database of Rate Regions for small IDSC instances}
\begin{table}
\caption{\label{tab:IDSCresults}Sufficiency of codes for IDSC instances: Columns 3 and 4 show the number of instances that the rate region inner bounds match with the Shannon outer bound.}\vspace{-.3cm}
\begin{center}
%\begin{tabular}{| p{0.8cm}| p{1.5cm} | p{1.5cm} | p{1.5cm} | p{1.5cm} | p{1.5cm} |p{1.5cm}|p{1.5cm} |}
\begin{tabular}{|c|c|c|c|c|}
\hline
$(K,L)$ &$|\Zmc|$& $\Rmc_{s,2}(\Asf)$ & $\Rmc_{2}^{N+1}(\Asf)$ \\ \hline
$(2,2)$ & 4 & 4 & 4 \\ \hline
$(2,3)$ & 33 & 26 & 33 \\ \hline
$(3,2)$ & 3 & 3 & 3 \\ \hline
$(3,3)$ & 179 & 143 & 179 \\ \hline
\end{tabular}
\end{center}
\end{table}

Here, experimental results on thousands of IDSC (defined in Example \ref{ex:IDSC} in \S\ref{sec:specialclasses}) instances are presented separately.  We investigated rate regions for $219$ non-isomorphic minimal IDSC instances representing $5130$ isomorphic ones.  These include the cases when $(K,L)=(2,2),(2,3),(3,2),(3,3)$.  Similarly, for the rate region of each non-isomorphic IDSC instance, we calculated its Shannon outer bound $\Rmc_o$,  scalar binary inner bound $\Rmc_{s,2}$, and the vector binary inner bounds $\Rmc_{2}^{N+1}$, where $N=K+L=K+|\Emc|$.  

A summary of results on the number of instances for which the various bounds agree is shown in Table \ref{tab:IDSCresults}.  The exact rate regions, their converses, and the codes that achieve them for all 219 non-isomorphic cases can be obtained at \cite{CongduanIDSCfile} and can be re-derived using \cite{EntVecSoft}.
For the non-isomorphic IDSC instances we considered, the Shannon outer bound is always tight on the rate regions, and the exact rate regions are obtained.
Scalar binary codes also only suffice for the instances with $L=2$ but not for all instances with $L=3$. 
However, vector binary codes from binary matroids on $N+1$ variables suffice for all the 219 instances.  Thus, for the IDSC problems up to $K\leq 3, L\leq 3$, vector binary codes suffice.

After obtaining these massive databases of all rate regions for small networks, our next question is how to learn from them, and further, how to use them to solve more (larger) networks.  For this purpose, we will develop in the following two sections notions of network hierarchy that enable us to relate networks of different sizes, their rate regions, and their properties with one another.

\section{Network Embedding Operations}\label{sec:embedding}
In this section, we propose a series of embedding operations relating smaller networks to larger networks in a manner such that one can directly obtain the rate region of the smaller network from the rate region of the larger network.  These operations will be selected in a manner that, due to this mapping, properties of the larger network can be considered to be inherited from small networks embedded within it.  In particular, we will show that if a certain class of codes is insufficient to exhaust the rate region of a small network embedded in a larger one, then this class of codes will be insufficient to exhaust the rate region of the larger one as well.  

\subsection{Definition of embedding operations}\label{subsec:embeddingdef}
The first operation is source deletion. When a source is deleted or removed, the source does not exist in the new network and the decoders that previously demanded it will no longer demand it after deletion.  Fig.\ \ref{fig:sd} illustrates the deletion of a source. When source $k$ is deleted, $t$ will no longer require $k$. A particular example is shown in Fig.\,\ref{fig:sourcedeletionexample}.
After deleting the source, the minimality conditions are checked to make sure the obtained network is minimal, and if not, an associated minimal network is found via a series of reductions according to Thm \ref{thm:minimality}.
\begin{definition}[Source Deletion ($\Asf\backslash k$)]\label{def:sourdel}
Fix network $\Asf=(\Smc,\Gmc,\Tmc,\Emc,\beta)$. If source $k\in\Smc$ is deleted, then, the new network is $\mathrm{minimal}(\Asf')$, where $\Asf'=(\Smc',\Gmc,\Tmc,\Emc,\beta')$ with $\Smc'=\Smc\setminus k$ and $\beta'=(\beta(t)\setminus k,t\in\Tmc)$.
\end{definition}

\begin{figure}
\centering 
\subfloat [\label{fig:sd}Source deletion: when source $k$ is deleted, it sends nothing to the network. Decoders that previously required $Y_k$ will no longer require it.]{\includegraphics[scale=0.5]{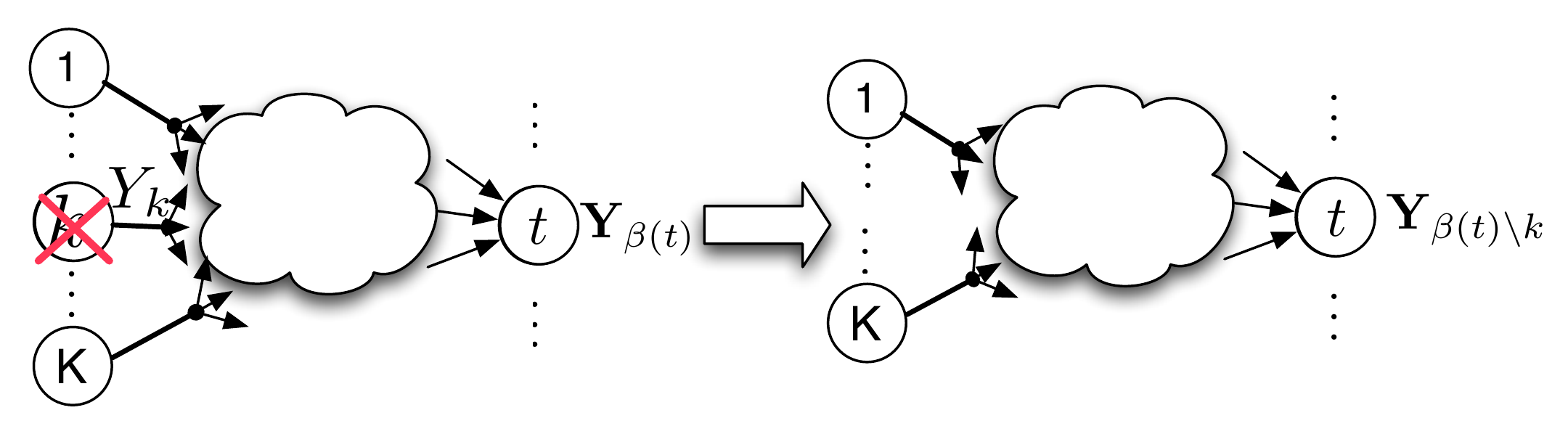} }\hspace{.2cm}
\subfloat [\label{fig:ec}Edge contraction: when $e$ is contracted, the head nodes directly have access to input of ${\rm Tl}(e)$.]{\includegraphics[scale=0.5]{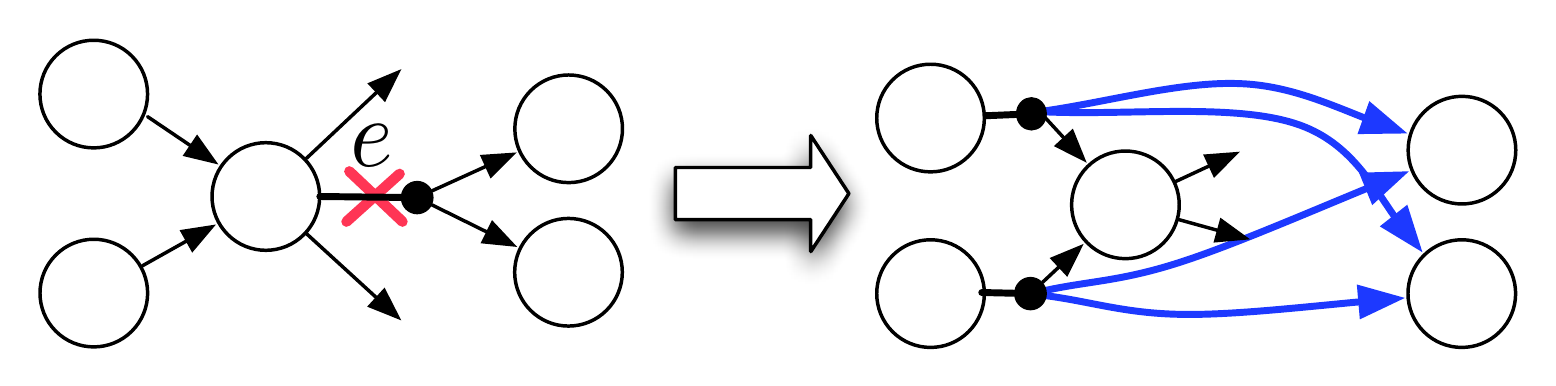} }\hspace{.2cm}
\subfloat [\label{fig:ed}Edge deletion: when delete $e$, its head nodes no longer receive information from $e$.]{\includegraphics[scale=0.5]{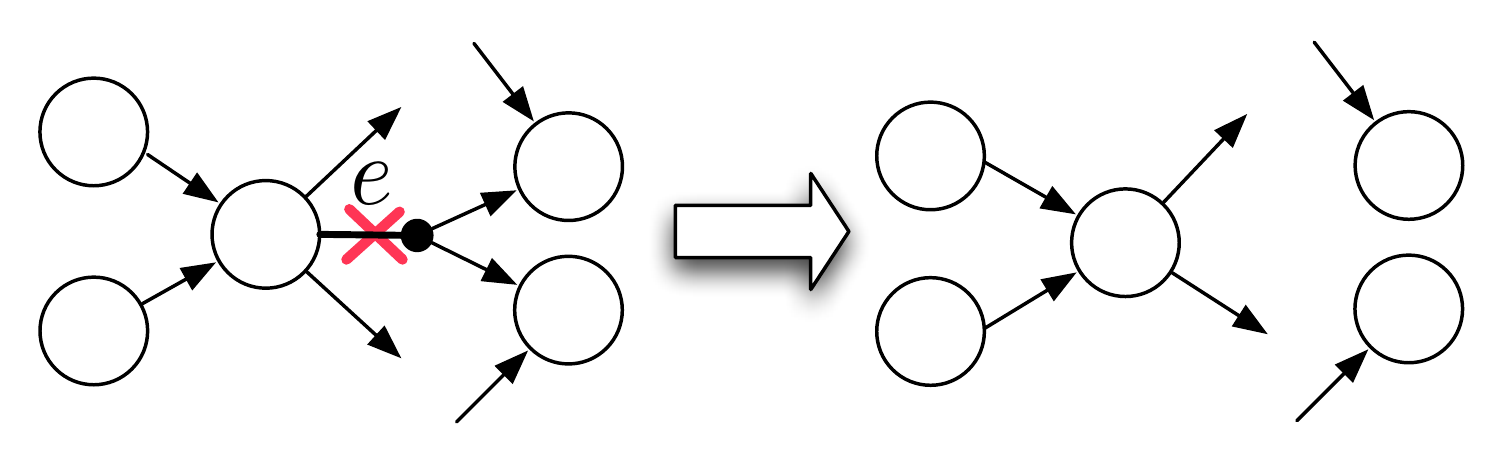} }
\caption{Definitions of embedding operations on a network}
\label{fig:embed} 
\end{figure}

\begin{figure}
\centering 
\subfloat [\label{fig:sourcedeletionexample}Source deletion example: when $s_1$ is deleted, its hyperedge is removed, and the sink $t_2$ which previously $Y_1,Y_2$ will now demand only $Y_2$.  When minimality is considered, it will be observed that the new sink $t_2$s ability to decode $Y_2$ has been implied by $t_1$.  Thus, $t_2$ is removed as well.  At this point $U_2$ and $U_3$ have become parallel edges, which are then merged.]{\includegraphics[scale=0.6]{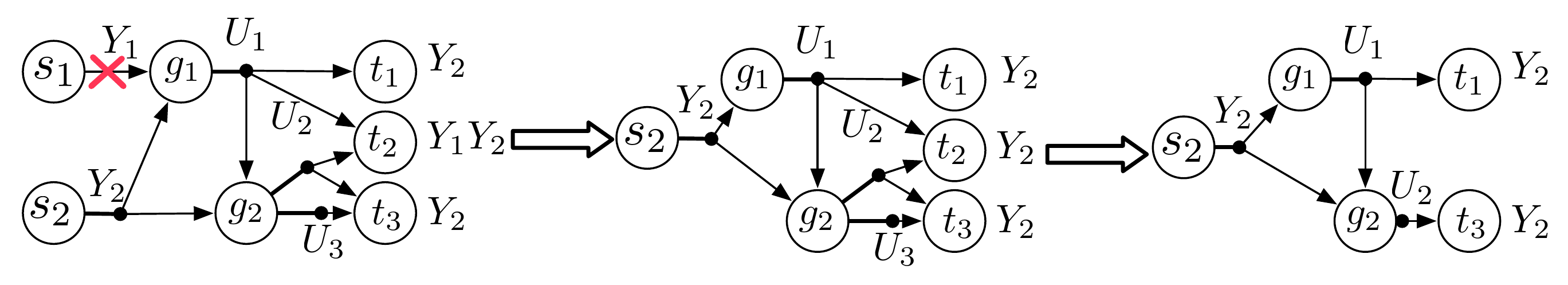} }

\subfloat[\label{fig:edgecontractionexample}Demonstration of edge contraction on a network: when $e_3$ is contracted, the input of $g_2$ will be directly available to $t_3$.  When minimality is considered, $Y_2$ is now trivially decoded at $t_3$ due to direct access to it, and thus $Y_2$ is removed from $\beta(t_3)$.  In addition, $g_2$ is removed.]{ \includegraphics[scale=0.6]{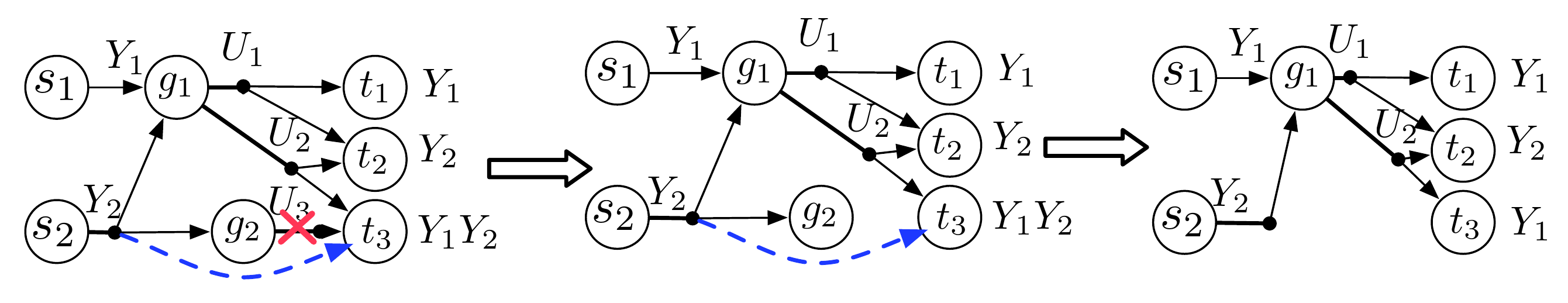} }

\subfloat[\label{fig:edgedeletionexample}Demonstration of edge deletion on a network: when $e_2$ is deleted, $t_2,t_3$ have no access to $U_2$.  Then $t_1,t_2$ are combined since they have the same input after deleting $U_2$.]{\includegraphics[scale=0.6]{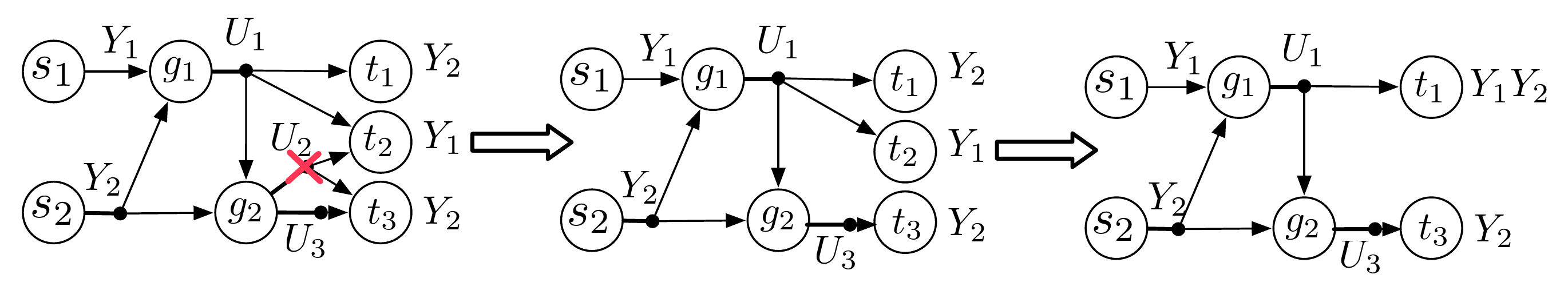} }
\caption{Examples to show the embedding operations on a network}
\end{figure}

The next operation we consider is edge contraction.  When an edge is contracted, the edge will be removed, and its head nodes are given direct access to all the inputs of the tail node.  Fig.\ \ref{fig:ec} demonstrates the contraction of an edge. As it shows, when edge $e$ is contracted, the head nodes it connects to will directly have access to all the input of its tail node.   Minimality conditions need to be checked after this operation.  A particular example is shown in Fig.\,\ref{fig:edgecontractionexample}.  

\begin{definition}[Edge Contraction $(\Asf\slash e)$]\label{def:edgcont}
Fix network $\mathsf{A}=(\Smc,\Gmc,\Tmc,\Emc,\beta)$.  If edge $e$ is contracted, then, the new network is $\mathrm{minimal}(\Asf')$, with $\Asf' = (\Smc,\Gmc,\Tmc,\Emc',\beta)$ where $\Emc'= \Emc\setminus \left( e \cup \textrm{In}(\textrm{Tl}(e))  \right) \bigcup_{e' \in \textrm{In}(\textrm{Tl}(e))} \{{\rm Tl}(e'), {\rm Hd}(e) \cup {\rm Hd}(e')\}$.

\end{definition}

Finally, we define edge deletion.  When an edge is deleted, it is simply removed from the graph, and the resulting graph is then checked and, if necessary, further reduced, for minimality.  Fig.\ \ref{fig:ed} demonstrates the deletion of an edge, and Fig. \ref{fig:edgedeletionexample} gives a particular example of the operation.  With consideration of minimality conditions, we formally define the edge deletion.

\begin{definition}[Edge Deletion $(\Asf \backslash e)$]\label{def:edgdel}
Fix network $\mathsf{A}=(\Smc,\Gmc,\Tmc,\Emc,\beta)$.  If edge $e$ is deleted, then, the smaller network instance $\mathrm{minimal}(\Asf')$, with $\Asf'=(\Smc,\Gmc,\Tmc,\Emc',\beta)$ where $\Emc'=\Emc\setminus e$.
\end{definition}

Based on these operations, we make precise the notion of an {\it embedded} network, or a network minor.
\begin{definition}[Embedded Network] \label{def:embedded}
A network $\Asf'$ is said to be {\it embedded} in another network $\Asf$, or is a {\it minor} of $\Asf$, denoted as $\Asf' \prec \Asf$, if $\Asf'$ can be obtained by a series of operations of source deletion, edge deletion/ contraction on $\Asf$.  Similarly, we say that $\Asf$ is an extension of $\Asf'$, denoted $\Asf \succ \Asf'$.
\end{definition}

With this definition in hand, we set out in the next subsection on determining the relationship between the rate region and properties of a large network and the rate region and properties of a small network embedded within it.

\subsection{Inheritance of Rate Regions \& their Properties Under Embedding Operations}
In this section we will prove a series of theorems that explain both how to obtain the rate region of an embedded network, under the operators defined in the previous subsection, from that of a larger extension network, as well as how certain properties of the rate region can be viewed as inherited under embedding operations.  A particularly interesting rate region property we will consider is the sufficiency of a class of linear codes to exhaust the entire capacity region.  Note that in each of the theorems below, the network $\Asf'$ will refer to the network in the definition of the associated operator (source deletion, edge contraction, and edge deletion) before the $\textrm{minimal}(\cdot)$ operator is applied.

% need to introduce notation for what happens to the rate region under minimality reduction and use that in the following

%Because the rate region of a non-minimal network can be easily obtained from its minimal equivalent, and vice-versa, we will 
%Since the rate region of a non-minimal network can be easily derived from the rate region of its minimal representation, according to Theorem \ref{thm:minimality}, in the following discussions, without loss of generality, we will assume that the obtained smaller networks after embedding operations need no further minimality reduction.  That is, the smaller network obtained from each operation is only one element less than the original network.

\begin{theorem} \label{thm:SrcDel}
Suppose a network $\Asf''=\textrm{minimal}(\Asf')$ is a minimal form of a network $\Asf'$ created by deleting source $k$ from another network $\Asf=(\Smc,\Gmc,\Tmc,\Emc,\beta)$, i.e., $\Asf'=\Asf\setminus k$.  Then for every $ l\in\{*,q,(s,q),o\}$
\begin{equation}\label{eq:srcdeleq1}
\Rmc_l(\Asf'')=\textrm{minimal}_{\Asf' \rightarrow \Asf''} \left( {\rm Proj}_{\boldsymbol{\omega}\setminus H(Y_k),\boldsymbol{r}}\left(\left\{\Rbf\in \Rmc_l(\Asf)\left| H(Y_k)=0\right.\right\}\right) \right).
\end{equation}
%\begin{eqnarray}\label{eq:srcDeleq1}
%\Rmc(\Asf')={\rm Proj}_{\Xbf_{\setminus k},\Rbf_{\Emc'}}\left(\left\{\Rbf\in \Rmc(\Asf)\left| H(X_k)=0\right.\right\}\right),\\
%\label{eq:srcDeleq2}
%\Rmc_q(\Asf')={\rm Proj}_{\Xbf_{\setminus k},\Rbf_{\Emc'}}\left(\left\{\Rbf\in \Rmc_q(\Asf)\left| H(X_k)=0\right.\right\}\right),\\
%\label{eq:srcDeleq3}
%\Rmc_{s,q}(\Asf')={\rm Proj}_{\Xbf_{\setminus k},\Rbf_{\Emc'}}\left(\left\{\Rbf\in \Rmc_{s,q}(\Asf)\left| H(X_k)=0\right.\right\}\right).
%\end{eqnarray}
\end{theorem}
\begin{IEEEproof}
We will prove $\Rmc_l(\Asf') = {\rm Proj}_{\boldsymbol{\omega}\setminus H(Y_k),\boldsymbol{r}}\left(\left\{\Rbf\in \Rmc_l(\Asf)\left| H(Y_k)=0\right.\right\}\right)$, since the remainder of the theorem holds from the minimality reductions in Thm. \ref{thm:minimality}.

Select any point $\Rbf'\in \Rmc_*(\Asf')$.  Then there exists a conic combination of some points in $\Rmc_*(\Asf')$ that are associated with entropic vectors in $\Gamma_{N'}^*$ such that $\Rbf'=\sum\limits_{\rbf'_j \in\Rmc_*(\Asf')}\alpha_j \rbf'_j$, where $\alpha_j\geq 0, \forall j$.  For each $\rbf'_j$, there exist random variables $\Ybf^{(j)}_{\setminus k},U^{(j)}_i , i\in \Emc$, where $\Ybf^{(j)}_{\setminus k}=\left[ Y_i^{(j)} \left| i\in \Smc \setminus k \right.\right]$, such that the entropy vector 
\begin{equation*}
\hbf^{(j)'}=\left[H(\Amc) \left| \Amc\subseteq \left\{Y_s^{(j)},U_e^{(j)} \left| s\in\Smc\setminus k,e\in\Emc\right.\right\}\right.\right]
\end{equation*}
 is in $\Gamma_{N'}^*$, where $N'=N-1$ is the number of variables in $\Asf'$.  Furthermore, their entropies satisfy all the constraints determined by $\Asf'$.  Define $Y^{(j)}_k$ to be the empty sources, $H(Y^{(j)}_k)=0$. Then the entropies of random variables $\{\Ybf^{(j)}_{\setminus k},U^{(j)}_i,i\in \Emc'\}\cup Y^{(j)}_k$ will satisfy the constraints in $\Asf$ with $H(Y^{(j)}_k)=0$ and the entropy vector $\hbf^{(j)}=\left[H(\Amc) \left| \Amc\subseteq \left\{Y_s^{(j)},U_e^{(j)} \left| s\in\Smc,e\in\Emc\right.\right\}\right.\right]$ will be in $\Gamma_N^*$ since adding an empty variable does not make an entropic vector to be non-entropic.  Denote $\rbf_j=[\rbf'_j, H(Y_k^{(j)})=0]$, then $\rbf_j\in \Rmc_*(\Asf)$. Hence, by using the same conic combination, we have an associated rate point $\Rbf=\sum\limits_{\rbf_j \in\Rmc_*(\Asf)}\alpha_j \rbf_j\in \left\{\Rbf\in \Rmc_*(\Asf)\left| H(Y_k)=0\right.\right\}$.  Thus, we have $\Rmc_*(\Asf')\subseteq {\rm Proj}_{\boldsymbol{\omega}\setminus H(Y_k),\boldsymbol{r}}(\{\Rbf\in \Rmc_*(\Asf)|H(Y_k)=0\})$.  If $\Rbf'$ is achievable by $\Fbb_q$ codes (scalar or vector), there exists a construction of some basic $\Fbb_q$ codes (scalar or vector) to achieve it.  Since letting $Y_s^{(j)}$ be empty does not affect the other sources and codes, the same construction of basic $\Fbb_q$ codes will also achieve the point $\Rbf$ with $H(Y_s)=0$.  Thus, $\Rmc_l(\Asf')\subseteq{\rm Proj}_{\boldsymbol{\omega}\setminus H(Y_k),\boldsymbol{r}}( \{\Rbf\in \Rmc_l(\Asf)|H(Y_k)=0\}),l\in\{q,(s,q)\}$.

On the other hand, if we select any point $\Rbf\in \{\Rbf\in \Rmc_*(\Asf)|H(Y_k)=0\}$, then, there exists a conic combination of some points in $\Rmc_*(\Asf)\cap \{H(Y_k)=0\}$ associated with entropic vectors in $\Gamma_N^*$, i.e., $\Rbf=\sum\limits_{\rbf_j \in\Rmc_*(\Asf)\cap \{H(Y_k)=0\}}\alpha_j \rbf_j, \ \alpha_j\geq 0,\ \forall j$.  For each $\rbf_j$, there exist random variables $\left\{\Ybf^{(j)}_{\Smc},\Ubf^{(j)}_{\Emc}\right\}$ such that their entropies satisfy all the constraints  determined by $\Asf$.   Furthermore, since $\alpha_j\geq 0$, the only conic combination makes $H(Y_k)=0$ is the case that $H(Y_k^{(j)})=0$.  We can drop $H(Y_k^{(j)})$ because  the entropies of $\left\{\Ybf^{(j)}_{\setminus k},\Ubf^{(j)}_{\Emc} \right\}$ satisfy all constraints determined by $\Asf'$ and the entropic vector projecting out $Y_k^{(j)}$ is still entropic.  Using the same conic combination, $\Rbf'={\rm Proj}_{\setminus H(Y_k)}\sum\limits_{\rbf_j \in\Rmc_*(\Asf)}\alpha_j \rbf_j ={\rm Proj}_{\boldsymbol{\omega}\setminus H(Y_k),\boldsymbol{r}} \ \Rbf\in \Rmc_*(\Asf')$.  Thus, we have ${\rm Proj}_{\boldsymbol{\omega}\setminus H(Y_k),\boldsymbol{r}}(\{\Rbf\in \Rmc_*(\Asf)|H(Y_k)=0\})\subseteq\Rmc_*(\Asf')$.  If $\Rbf$ is achievable by $\Fbb_q$ code $\Cbb$, then the code to achieve $\Rbf'$ could be the code $\Cbb$ with deletion of rows associated with source $Y_k$, i.e., $\Cbb'=\Cbb_{\setminus Y_k,:}$. Thus, ${\rm Proj}_{\boldsymbol{\omega}\setminus H(Y_k),\boldsymbol{r}}(\{\Rbf\in \Rmc_l(\Asf)|H(Y_k)=0\})\subseteq\Rmc_l(\Asf'),l\in\{q,(s,q)\}$.

Furthermore, for any point $\Rbf'\in\Rmc_o(\Asf')$, there exists an associated point $\hbf' \in\Gamma_{N'}$ and a rate vector $\rbf'=[R_e | e\in\Emc]$ such that $\Rbf'={\rm Proj}_{\boldsymbol{\omega}\setminus H(Y_k),\boldsymbol{r}}\  [\hbf',\rbf']\cap \Lmc_{\Asf'}$.  Clearly, if we increase the dimension of $\hbf'$ by adding a variable $Y_k$ with zero entropy, i.e., $H(Y_k)=0$, we have the new entropy vector in $\Gamma_N$.  That is, if we define $\hbf=\left[h'_{\Amc\cap \{Y_s,U_e|s\in\Smc',e\in\Emc\}}|\Amc\subseteq \{Y_s,U_e|s\in\Smc,e\in\Emc\}\right]$, then $\hbf \in\Gamma_N$.  Since $H(Y_k)=0$, the network constraints in $\Asf$ will be  satisfied given that the zero entropy does not break the conditional entropies associated with network constraints.  Hence, there exists an associated point $\Rbf\in\Rmc_o(\Asf)$ with $H(Y_k)=0$.  Therefore, we have $\Rmc_o(\Asf')\subseteq {\rm Proj}_{\boldsymbol{\omega}\setminus H(Y_k),\boldsymbol{r}} (\{\Rbf\in\Rmc_o(\Asf)|H(Y_k=0)\})$.  Reversely, suppose a point $\Rbf\in\Rmc_o(\Asf)$ is picked with $H(Y_k)=0$.  There exists a vector $\hbf\in\Gamma_N$ and a rate vector $\rbf=[R_e | e\in\Emc]$ such that $\Rbf={\rm Proj}_{\boldsymbol{\omega},\boldsymbol{r}}\ [\hbf,\rbf]\cap \Lmc_{\Asf}$.  Since the network constraints $\Lmc_\Asf$ with $H(Y_k)=0$ will be $\Lmc_{\Asf'}$, and ${\rm Proj}_{\boldsymbol{\omega}\setminus H(Y_k),\boldsymbol{r}} \ [\hbf,\rbf] \in\Gamma_{N'}\cap \Lmc_{\Amc'}$, we have ${\rm Proj}_{\boldsymbol{\omega}\setminus H(Y_k),\boldsymbol{r}} \Rbf\in\Rmc_o(\Amc')$.  Therefore, we have ${\rm Proj}_{\boldsymbol{\omega}\setminus H(Y_k),\boldsymbol{r}} (\{\Rbf\in\Rmc_o(\Asf)|H(Y_k=0)\})\subseteq \Rmc_o(\Asf')$.
\end{IEEEproof}

\begin{theorem} \label{thm:EncCon}
Suppose a network $\Asf''=\textrm{minimal}(\Asf')$ is a minimal form of a network $\Asf'$ obtained by contracting $e$ from another network $\Asf=(\Smc,\Gmc,\Tmc,\Emc,\beta)$, i.e., $\Asf'=\Asf\slash e$.  Then
\begin{eqnarray}
\Rmc_l(\Asf'')&=&\textrm{minimal}_{\Asf' \rightarrow \Asf''} \left( {\rm Proj}_{\boldsymbol{\omega},\boldsymbol{r}\setminus R_e}\Rmc_l(\Asf) \right),\  l\in\{*,q,o\} \label{eq:encConeq1}\\
\Rmc_{s,q}(\Asf'')&\supseteq&\textrm{minimal}_{\Asf' \rightarrow \Asf''} \left( {\rm Proj}_{\boldsymbol{\omega},\boldsymbol{r}\setminus R_e}\Rmc_{s,q}(\Asf) \right), \label{eq:encDeleq3}
\end{eqnarray}
%\begin{eqnarray}
%\Rmc(\Asf')&=&{\rm Proj}_{H(X_k),k\in\Smc,\Rbf_{\Emc'}}\Rmc(\Asf), \label{eq:encDeleq1}\\
%\Rmc_q(\Asf')&=&{\rm Proj}_{H(X_k),k\in\Smc,\Rbf_{\Emc'}}\Rmc_q(\Asf), \label{eq:encDeleq2}\\
%\Rmc_{s,q}(\Asf')&\supseteq&{\rm Proj}_{H(X_k),k\in\Smc,\Rbf_{\Emc'}}\Rmc_{s,q}(\Asf), \label{eq:encDeleq3}
%\end{eqnarray}
\end{theorem}
\begin{IEEEproof}
We will prove $\Rmc_l(\Asf') = {\rm Proj}_{\boldsymbol{\omega},\boldsymbol{r}\setminus R_e}\left(\left\{\Rbf\in \Rmc_l(\Asf)\right\}\right)$ for $ l\in\{*,q,o\}$, and for the scalar case, $\Rmc_{s,q}(\Asf'')\supseteq \textrm{minimal}_{\Asf' \rightarrow \Asf''} \left( {\rm Proj}_{\boldsymbol{\omega},\boldsymbol{r}\setminus R_e}\Rmc_{s,q}(\Asf) \right)$, since the remainder of the theorem holds from the minimality reductions in Thm. \ref{thm:minimality}.

Select any point $\Rbf'\in \Rmc_*(\Asf')$.  Then there exists a conic combination of some points in $\Rmc_*(\Asf')$ that are associated with entropic vectors in $\Gamma_{N'}^*$ such that $\Rbf'=\sum\limits_{\rbf'_j \in\Rmc_*(\Asf')}\alpha_j \rbf'_j$, where $\alpha_j\geq 0, \forall j$.  For each $\rbf'_j$, there exist random variables $\Ybf^{(j)}_{\Smc},U^{(j)}_i , i\in \Emc\setminus e$, such that the entropy vector 
\begin{equation*}
\hbf^{(j)'}=\left[H(\Amc) \left| \Amc\subseteq \left\{Y_s^{(j)},U_i^{(j)} \left| s\in\Smc,i\in\Emc\setminus e\right.\right\}\right.\right]
\end{equation*}
 is in $\Gamma_{N'}^*$, where $N'=N-1$ is the number of variables in $\Asf'$.  Furthermore, their entropies satisfy all the constraints determined by $\Asf'$.  In the network $\Asf$, define $U^{(j)}_e$ to be the concatenation of all inputs to the tail node of $e$, $U^{(j)}_e=\Ubf^{(j)}_{{\rm In}({\rm Tl}(e))}$. Then the entropies of random variables $\left\{\Ybf^{(j)}_{\Smc},\Ubf^{(j)}_{\Emc}\right\}$ will satisfy the constraints in $\Asf$, and additionally obey $H(U^{(j)}_e)=H(\Ubf^{(j)}_{{\rm In}({\rm Tl}(e))})$. Hence, $\hbf^{(j)}=\left[H(\Amc)\left|\Amc\subseteq \left\{Y_s^{(j)},,U_i^{(j)}  \left| s\in\Smc,i\in\Emc\right.\right\}\right.\right]\in\Gamma_N^*$.  That is, $\rbf_j=[\rbf'_j,R_e\geq H(\Ubf^{(j)}_{{\rm In}({\rm Tl}(e))})]\in \Rmc_*(\Asf)$.  By using the same conic combination, we have an associated rate point $\Rbf=\sum\limits_{\rbf_j \in\Rmc_*(\Asf)}\alpha_j \rbf_j\in \left\{\Rbf\in \Rmc_*(\Asf)\left| R_e\geq H(\Ubf^{(j)}_{{\rm In}({\rm Tl}(e))})\right.\right\}$. Thus, we have 
\begin{equation} 
\Rmc(\Asf')\subseteq
{\rm Proj}_{\boldsymbol{\omega},\boldsymbol{r}\setminus R_e}( \{\Rbf\in \Rmc(\Asf)|R_e\geq H(\Ubf_{{\rm In}({\rm Tl}(e))})\})
\subseteq {\rm Proj}_{\boldsymbol{\omega},\boldsymbol{r}\setminus R_e}\Rmc(\Asf).\label{eq:pfeq1encDel1}
\end{equation}

If $\Rbf'$ is achievable by general $\Fbb_q$ codes, since concatenation of all input is a valid $\Fbb_q$ vector code, we have
\begin{equation}
\Rmc_q(\Asf')\subseteq
{\rm Proj}_{\boldsymbol{\omega},\boldsymbol{r}\setminus R_e}(\{\Rbf\in \Rmc_q(\Asf)|R_e\geq H(\Ubf_{{\rm In}({\rm Tl}(e))})\})\subseteq {\rm Proj}_{\boldsymbol{\omega},\boldsymbol{r}\setminus R_e}\Rmc_q(\Asf).\label{eq:pfeq1encDel2}
\end{equation}

However, we cannot establish same relationship when scalar $\Fbb_q$ codes are considered, because for the point $\Rbf'$, the associated $\Rbf$ with $H(U_e)$ may not be scalar $\Fbb_q$ achievable. 

On the other hand, if we select any point $\Rbf\in \{\Rbf\in \Rmc_*(\Asf)\}$, then, there exists a conic combination of some points in $\Rmc_*(\Asf)$ associated with entropic vectors in $\Gamma_N^*$, i.e., $\Rbf=\sum\limits_{\rbf_j \in\Rmc_*(\Asf)}\alpha_j \rbf_j, \ \alpha_j\geq 0,\ \forall j$.  For each $\rbf_j$, there exist random variables $\left\{\Ybf^{(j)}_{\Smc},\Ubf^{(j)}_{\Emc}\right\}$ such that their entropies satisfy all the constraints  determined by $\Asf$.  
 Since the entropies of $\left\{\Ybf^{(j)}_{\Smc},U^{(j)}_i | i\in \Emc\setminus e\right\}$  satisfy all constraints determined by $\Asf'$ (because they are a subset of the constraints from $\Asf$) and the entropic vector projecting out $U_e$ is still entropic.  Thus, by letting $R_e^{(j)}$ to be unconstrained, we have 
${\rm Proj}_{\boldsymbol{\omega},\boldsymbol{r}\setminus R_e}\rbf_j\in\Rmc_*(\Asf').$  Further, by using the same conic combination, $\Rbf'={\rm Proj}_{\setminus R_e}\sum\limits_{\rbf_j \in\Rmc_*(\Asf)}\alpha_j \rbf_j ={\rm Proj}_{\boldsymbol{\omega},\boldsymbol{r}\setminus R_e} \ \Rbf\in \Rmc_*(\Asf')$.  Thus, we have ${\rm Proj}_{\boldsymbol{\omega},\boldsymbol{r}\setminus R_e}(\{\Rbf\in \Rmc_*(\Asf)\})\subseteq\Rmc_*(\Asf')$.

If $\Rbf\in \Rmc_*(\Asf)$ is achievable by $\Fbb_q$ code $\Cbb$, either scalar or vector, then the code to achieve $\Rbf'={\rm Proj}_{\boldsymbol{\omega},\boldsymbol{r}\setminus R_e}\Rbf\in \Rmc_*(\Asf')$ could be the code $\Cbb$ with deletion of columns associated with edge $e$, i.e., $\Cbb'=\Cbb_{:,\setminus U_e}$, because the code on edge $e$ is not of interest.  Thus, we have 
${\rm Proj}_{\boldsymbol{\omega},\boldsymbol{r}\setminus R_e}\Rmc_l(\Asf)\subseteq\Rmc_l(\Asf'),\ l\in\{q,(s,q)\}.$

Furthermore, for any point $\Rbf'\in\Rmc_o(\Asf')$, there exists an associated point $\hbf'\in\Gamma_{N'}$ and a rate vector $\rbf'=[R_i|i\in\Emc\setminus e]$ such that  $\Rbf'={\rm Proj}_{\boldsymbol{\omega},\boldsymbol{r}\setminus R_e}\  [\hbf',\rbf']\cap \Lmc_{\Asf'}$.  Clearly, if we increase the dimension of $\hbf'$ by adding a variable $U_e$ which is the vector of all input variables to the tail node of $e$, i.e., $U_e=[U_i|i\in{\rm In}({\rm Tl}(e))]$ and $H(U_e)=H(U_i,i\in {\rm In}({\rm Tl}(e)))$, we have the new vector in $\Gamma_N$.  That is, if we define 
\begin{equation}
\hbf=\left\{\begin{array}{cc} h'_{\Amc\cap \{Y_s,U_i|s\in\Smc,i\in\Emc'\}}, & U_e \notin \Amc\\
h'_{\Amc\cap \{Y_s,U_i|s\in\Smc,i\in\Emc'\}\cup \{U_i | i\in {\rm In}({\rm Tl}(e))\}} & U_e \in \Amc
\end{array}
\right.
\end{equation}
for $\Amc\subseteq \{Y_s,U_i|s\in\Smc,i\in\Emc\}$, then $\hbf \in\Gamma_N$.  Further, we let $R_e$ to be unconstrained, i.e., $R_e=\infty$.  Since $H(U_e)\leq R_e$, the network constraints in $\Asf$ will be  satisfied given that the other constraints will not be affected.  Hence, there exists an associated point $\Rbf\in\Rmc_o(\Asf)$ with $H(U_e)\leq R_e$, where $R_e$ is unconstrainted.  Therefore, we have $\Rmc_o(\Asf')\subseteq {\rm Proj}_{\boldsymbol{\omega},\boldsymbol{r}\setminus R_e} (\{\Rbf\in\Rmc_o(\Asf)\})$.  Reversely, suppose a point $\Rbf\in\Rmc_o(\Asf)$ is picked with $R_e$ unconstrained.  There exists an associated vector $\hbf\in\Gamma_N$ and a rate vector $\rbf=[R_i | i\in\Emc]$ such that $\Rbf={\rm Proj}_{\boldsymbol{\omega},\boldsymbol{r}} \ [\hbf,\rbf]\cap \Lmc_{\Asf}$.  Since $R_e$ is unconstrained, we will have $H(U_e)$ unconstrained as well.  Since the network constraints $\Lmc_\Asf$ with $R_e$ unconstrained will be $\Lmc_{\Asf'}$, and ${\rm Proj}_{\boldsymbol{\omega},\boldsymbol{r}\setminus R_e} \ [\hbf,\rbf] \in\Gamma_{N'}\cap \Lmc_{\Asf'}$, we have ${\rm Proj}_{\boldsymbol{\omega},\boldsymbol{r}\setminus R_e} \Rbf\in\Rmc_o(\Asf')$.  Therefore, we have ${\rm Proj}_{\boldsymbol{\omega},\boldsymbol{r}\setminus R_e} (\{\Rbf\in\Rmc_o(\Asf)\})\subseteq \Rmc_o(\Asf')$.
\end{IEEEproof}

\begin{theorem} \label{thm:edgdel}
Suppose a network $\Asf''=\textrm{minimal}(\Asf')$ is a minimal form of a network $\Asf'$ obtained by deleting $e$ from another network $\Asf=(\Smc,\Gmc,\Tmc,\Emc,\beta)$, i.e., $\Asf'=\Asf\setminus e$.  Then
\begin{equation}
\Rmc_l(\Asf')=\textrm{minimal}_{\Asf' \rightarrow \Asf''} \left({\rm Proj}_{\boldsymbol{\omega},\boldsymbol{r}\setminus R_e}(\{\Rbf\in \Rmc_l(\Asf)|R_e=0\})\right),\ l\in\{*,q,(s,q),o\}
\label{eq:encConteq1}%\\
%\label{eq:encConteq2}
%\Rmc_q(\Asf')={\rm Proj}_{H(X_\Smc),\Rbf_{\Emc'}}(\{\Rbf\in \Rmc_q(\Asf)|R_e=0\}),\\
%\label{eq:encConteq3}
%\Rmc_{s,q}(\Asf')={\rm Proj}_{H(X_\Smc),\Rbf_{\Emc'}}(\{\Rbf\in \Rmc_{s,q}(\Asf)|R_e=0\}).
\end{equation}
\end{theorem}
\begin{IEEEproof}
We will prove $\Rmc_l(\Asf') = {\rm Proj}_{\boldsymbol{\omega},\boldsymbol{r}\setminus R_e}\left(\left\{\Rbf\in \Rmc_l(\Asf)\left| R_e=0\right.\right\}\right)$, since the remainder of the theorem holds from the minimality reductions in Thm. \ref{thm:minimality}.

Select any point $\Rbf'\in \Rmc_*(\Asf')$.  Then there exists a conic combination of some points in $\Rmc_*(\Asf')$ that are associated with entropic vectors in $\Gamma_{N'}^*$ such that $\Rbf'=\sum\limits_{\rbf'_j \in\Rmc_*(\Asf')}\alpha_j \rbf'_j$, where $\alpha_j\geq 0, \forall j$.  For each $\rbf'_j$, there exist random variables $\Ybf^{(j)}_{\Smc},U^{(j)}_i , i\in \Emc\setminus e$, such that the entropy vector 
\begin{equation*}
\hbf^{(j)'}=\left[H(\Amc) \left| \Amc\subseteq \left\{Y_s^{(j)},U_i^{(j)} \left| s\in\Smc,i\in\Emc\setminus e\right.\right\}\right.\right]
\end{equation*}
 is in $\Gamma_{N'}^*$, where $N'=N-1$ is the number of variables in $\Asf'$.  Furthermore, their entropies satisfy all the constraints determined by $\Asf'$.   Let $U^{(j)}_e$ be the empty set or encoding all input with the all-zero vector, $U^{(j)}_e=\emptyset$ and further let $R_e=0$. Then the entropies of random variables $\left\{\Ybf^{(j)}_{\Smc},\Ubf^{(j)}_{\Emc}\right\}$ will satisfy the constraints in $\Asf$, and additionally obey $H(U^{(j)}_e)\leq R_e=0$.  Furthermore, the vector $\hbf^{(j)}=\left[H(\Amc)\left| \Amc\subseteq \left\{Y_s^{(j)},U_i^{(j)}\left| s\in\Smc,i\in\Emc\right.\right\}\right.\right]\in\Gamma_{N}^*$.  That is, $\rbf_j=[\rbf'_j,R_e=0]\in \Rmc_*(\Asf)$.  By using the same conic combination, we have an associated rate point $\Rbf=\sum\limits_{\rbf_j \in\Rmc_*(\Asf)}\alpha_j \rbf_j\in \left\{\Rbf\in \Rmc_*(\Asf)\left| R_e=0\right.\right\}$.  Thus, we have $\Rmc_*(\Asf')\subseteq {\rm Proj}_{\boldsymbol{\omega},\boldsymbol{r}\setminus R_e}(\{\Rbf\in \Rmc_*(\Asf)|R_e=0\})$.  

If $\Rbf'$ is achievable by general $\Fbb_q$ linear vector or scalar codes, there exists a construction of basic linear codes to achieve it.  Since all-zero code is a valid $\Fbb_q$ vector and scalar linear code, we have 
$\Rmc_l(\Asf')\subseteq{\rm Proj}_{\boldsymbol{\omega},\boldsymbol{r}\setminus R_e}(\{\Rbf\in \Rmc_l(\Asf)|R_e=0\}),\ l\in\{q,(s,q)\}.
$

On the other hand, if we select any point $\Rbf\in \{\Rbf\in \Rmc_*(\Asf)|R_e=0\}$, then, there exists a conic combination of some points in $\Rmc_*(\Asf)\cap \{R_e=0\}$ associated with entropic vectors in $\Gamma_N^*$, i.e., $\Rbf=\sum\limits_{\rbf_j \in\Rmc_*(\Asf)\cap \{R_e=0\}}\alpha_j \rbf_j, \ \alpha_j\geq 0,\ \forall j$.  For each $\rbf_j$, there exist random variables $\left\{\Ybf^{(j)}_{\Smc},\Ubf^{(j)}_{\Emc}\right\}$ such that their entropies satisfy all the constraints  determined by $\Asf$.   Furthermore, since $\alpha_j\geq 0$, the only conic combination makes $R_e=0$ is the case that $R_e^{(j)}=0$ and further $H(U_e^{(j)})=0$.  We can drop $H(U_e^{(j)})$, i.e., $R_e$, because  the entropies of $\left\{\Ybf^{(j)}_{\Smc},\Ubf^{(j)}_{\setminus e} \right\}$ satisfy all constraints determined by $\Asf'$ and the entropic vector projecting out $U_e^{(j)}$ is still entropic.  Using the same conic combination, $\Rbf'={\rm Proj}_{\setminus R_e}\sum\limits_{\rbf_j \in\Rmc_*(\Asf)}\alpha_j \rbf_j ={\rm Proj}_{\boldsymbol{\omega},\boldsymbol{r}\setminus R_e} \ \Rbf\in \Rmc_*(\Asf')$.  Thus, we have ${\rm Proj}_{\boldsymbol{\omega},\boldsymbol{r}\setminus R_e}(\{\Rbf\in \Rmc_*(\Asf)|R_e=0\})\subseteq\Rmc_*(\Asf')$.  If $\Rbf$ is achievable by $\Fbb_q$ code $\Cbb$, then the code to achieve $\Rbf'$ could be the code $\Cbb$ with deletion of columns associated with edge $U_e$, i.e., $\Cbb'=\Cbb_{:,\setminus U_e}$. Thus, ${\rm Proj}_{\boldsymbol{\omega},\boldsymbol{r}\setminus R_e}(\{\Rbf\in \Rmc_l(\Asf)|R_e=0\})\subseteq\Rmc_l(\Asf'),l\in\{q,(s,q)\}$.

Furthermore, for any point $\Rbf'\in\Rmc_o(\Asf')$, there exists an associated point $\hbf'\in\Gamma_{N'}$ and a rate vector $\rbf'=[R_i | i\in\Emc\setminus e]$ such that $\Rbf'={\rm Proj}_{\boldsymbol{\omega},\boldsymbol{r}\setminus R_e}\ [\hbf',\rbf']\cap \Lmc_{\Asf'}$.  Clearly, if we increase the dimension of $\hbf'$ by adding a variable $U_e$ with zero entropy, i.e., $H(U_e)=0$, we have the new entropy vector in $\Gamma_N$.  That is, if we define $\hbf=\left[h'_{\Amc\cap \{Y_s,U_i|s\in\Smc,i\in\Emc'\}}|\Amc\subseteq \{Y_s,U_i|s\in\Smc,i\in\Emc\}\right]$, then $\hbf \in\Gamma_N$.  Further, we let $R_e=0$.  Since $H(U_e)=R_e=0$, the network constraints in $\Asf$ will be  satisfied given that the zero entropy (capacity) does not break the conditional entropies associated with network constraints.  Hence, there exists an associated point $\Rbf\in\Rmc_o(\Asf)$ with $H(U_e)=R_e=0$.  Therefore, we have $\Rmc_o(\Asf')\subseteq {\rm Proj}_{\boldsymbol{\omega},\boldsymbol{r}\setminus R_e} (\{\Rbf\in\Rmc_o(\Asf)|R_e=0\})$.  Reversely, suppose a point $\Rbf\in\Rmc_o(\Asf)$ is picked with $R_e=0$.  There exists a vector $\hbf\in\Gamma_N$ and a rate vector $\rbf=[R_i | i\in\Emc]$ such that $\Rbf={\rm Proj}_{\boldsymbol{\omega},\boldsymbol{r}}\ [\hbf,\rbf]\cap \Lmc_{\Asf}$.  Since $R_e=0$, we will have $H(U_e)=0$.  Since the network constraints $\Lmc_\Asf$ with $H(U_e)=R_e=0$ will be $\Lmc_{\Asf'}$, and ${\rm Proj}_{\boldsymbol{\omega},\boldsymbol{r}\setminus R_e} \ [\hbf,\rbf] \in\Gamma_{N'}\cap \Lmc_{\Asf'}$, we have ${\rm Proj}_{\boldsymbol{\omega},\boldsymbol{r}\setminus R_e} \Rbf\in\Rmc_o(\Asf')$.  Therefore, we have ${\rm Proj}_{\boldsymbol{\omega},\boldsymbol{r}\setminus R_e} (\{\Rbf\in\Rmc_o(\Asf)|R_e=0\})\subseteq \Rmc_o(\Asf')$.
\end{IEEEproof}

\begin{corollary} \label{cor:embedded}
Consider two networks $\Asf,\Asf'$, with rate regions $\Rmc_*(\Asf),\Rmc_*(\Asf')$, such that $\Asf'\prec\Asf$.  If $\Fbb_q$ vector (scalar) linear codes suffice, or Shannon outer bound is tight for $\Asf$, then same statements hold for $\Asf'$. Equivalently, if $\Fbb_q$ vector (scalar) linear codes do not suffice, or Shannon outer bound is not tight for $\Asf'$, then same statements hold for $\Asf$. Equivalently, if $\Rmc_l(\Asf)=\Rmc_*(\Asf)$, then $\Rmc_l(\Asf')=\Rmc_*(\Asf')$, for some $l\in\{o,q,(s,q)\}$.
\end{corollary}
\begin{IEEEproof}
From Definition \ref{def:embedded} we know that $\Asf'$ is obtained by a series of operations of source deletion, edge deletion, edge contraction. Theorems \ref{thm:SrcDel}--\ref{thm:edgdel} indicate that sufficiency of linear codes, vector or scalar, and the tightness of Shannon outer bound are preserved for each single embedding operation.  For vector case, if $\Rmc_q(\Asf)=\Rmc_*(\Asf)$, \eqref{eq:srcdeleq1}, \eqref{eq:encConeq1}, \eqref{eq:encConteq1} directly give $\Rmc_q(\Asf')=\Rmc_*(\Asf')$ for source deletion, edge contraction, and edge deletion, respectively.  Similar arguments work for the tightness of the Shannon outer bound.  For scalar code sufficiency,  \eqref{eq:srcdeleq1} and \eqref{eq:encConteq1}
indicate the same preservation of sufficiency of scalar codes for source and edge deletion, respectively.  For edge contraction and assumption of if $\Rmc_{s,q}(\Asf)=\Rmc_*(\Asf)$, \eqref{eq:encConeq1} and \eqref{eq:encDeleq3} indicate $\Rmc_*(\Asf')\subseteq \Rmc_{s,q}(\Asf')$.  Together with the straightforward fact that $\Rmc_{s,q}(\Asf')\subseteq \Rmc_*(\Asf')$, since scalar $\Fbb_q$ codes achievable rate region must be subset of the entire rate region, we can see $\Rmc_{s,q}(\Asf')=\Rmc_*(\Asf')$ holds for edge contraction as well. 
\end{IEEEproof}

Having introduced the embedding operations, which give smaller networks from larger networks, we next introduce some combination operations to get larger networks from smaller ones.

\section{Network Combination Operations}\label{sec:combination}
In this section we propose a series of combination operations relating smaller networks with larger networks in a manner such that the rate region of the larger network can be easily derived from those of the smaller ones.  In addition, the sufficiency of a class of linear network codes is inherited in the larger network from the smaller one.   Throughout the following, the network $\Asf=(\Smc,\Gmc,\Tmc,\Emc,\beta)$ is a combination of two \emph{disjoint} networks $\Asf_i=(\Smc_i,\Gmc_i,\Tmc_i,\Emc_i,\beta_i)$, $i\in\{1,2\}$, meaning $\Smc_1\cap\Smc_2=\emptyset$, $\Gmc_1\cap\Gmc_2=\emptyset$, $\Tmc_1\cap\Tmc_2=\emptyset$, $\Emc_1\cap\Emc_2=\emptyset$, and $\beta_1(t_1)\cap \beta_2(t_2)=\emptyset,\forall t_1\in \Tmc_1,t_2\in\Tmc_2$.  

\subsection{Definition of Combination Operations}
%The first operation is direct sum. When two networks are directly summed, they are simply concatenated with same topologies.
%\vspace{-0.2cm}
%\begin{definition}[Direct Sum $(\Asf_1+\Asf_2)$]\label{def:dirsum}
%Suppose two networks $\Asf_1=(\Smc_1,\Gmc_1,\Tmc_1,\Emc_1,(\beta_1(t),t\in\Tmc_1))$ and $\Asf_2=(\Smc_2,\Gmc_2,\Tmc_2,\Emc_2,(\beta_2(t),t\in\Tmc_2))$ are concatenated by direct sum. In the new network $\Asf=(\Smc,\Gmc,\Tmc,\Emc,\beta)$, we will have:
%\begin{enumerate}
%\item $\Smc=\Smc_1\cup\Smc_2$;
%\item $\Gmc=\Gmc_1\cup\Gmc_2$;
%\item $\Tmc=\Tmc_1\cup\Tmc_2$;
%\item $\Emc=\Emc_1\cup\Emc_2$;
%\item $\beta(t)=\left\{\begin{array}{cc}
%\beta_1(t), & if\  t\in\Tmc_1;\\
%\beta_2(t), & if\  t\in\Tmc_2.
%\end{array}\right.$
%\end{enumerate}
%\vspace{-0.2cm}
%\end{definition}

\begin{figure}
\centering 
%\subfloat [\label{fig:exampleds}Demonstration of direct sum of two networks: two networks are simply concatenated.]{\includegraphics[scale=0.35]{ConcatenateNetworks} }
\subfloat [\label{fig:examplesi}Sources merge: the merged source will serve for the new larger network.]{\includegraphics[scale=0.5]{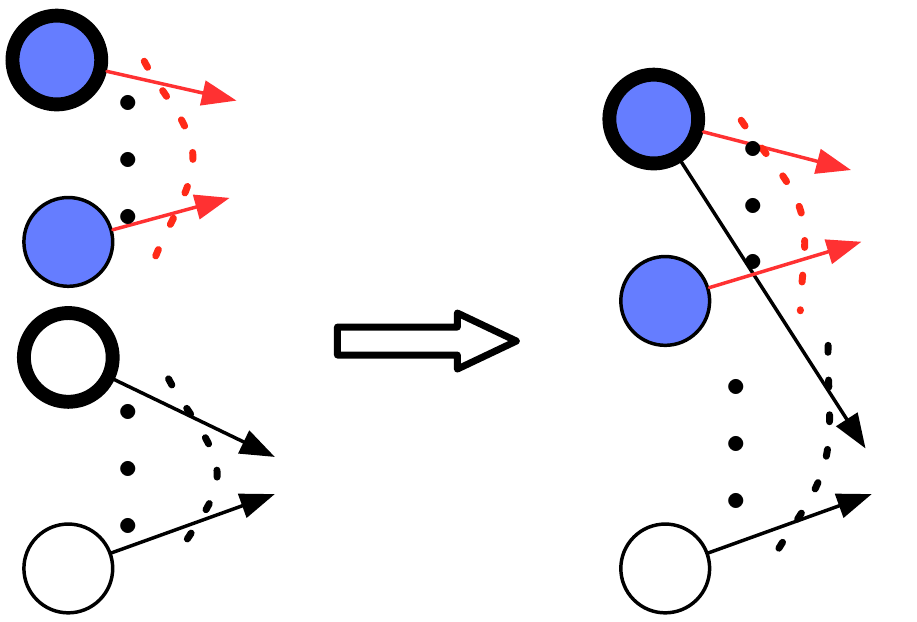} }\hspace{2mm}
\subfloat [\label{fig:examplesc}Sinks merge: input and output of the sinks are unioned, respectively.]{\includegraphics[scale=0.5]{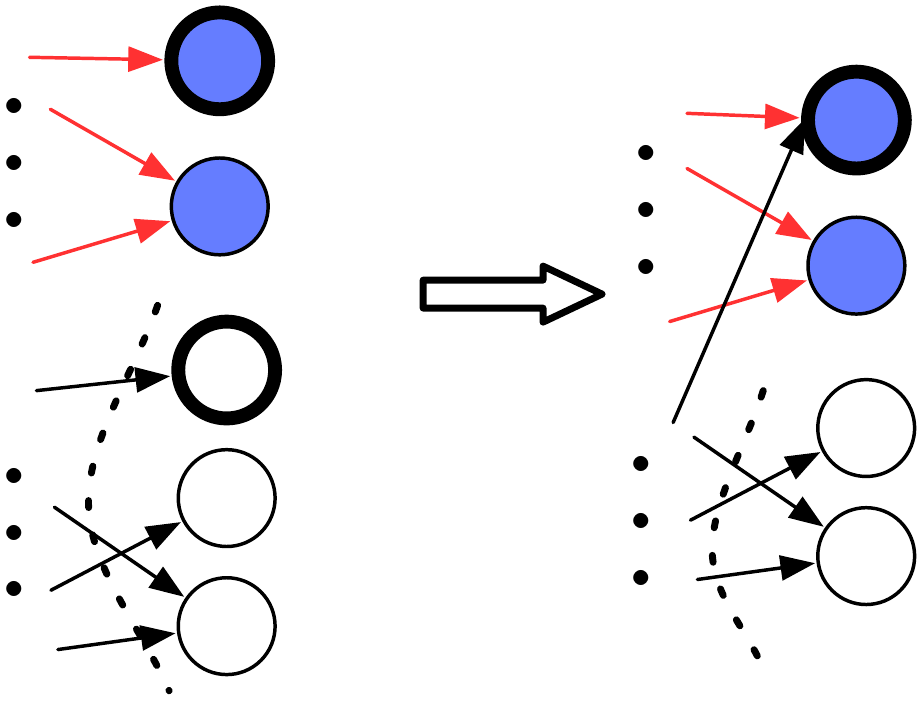} }

\subfloat [\label{fig:examplenc}Intermediate nodes merge: input and output of the nodes are unioned, respectively.]{\includegraphics[scale=0.5]{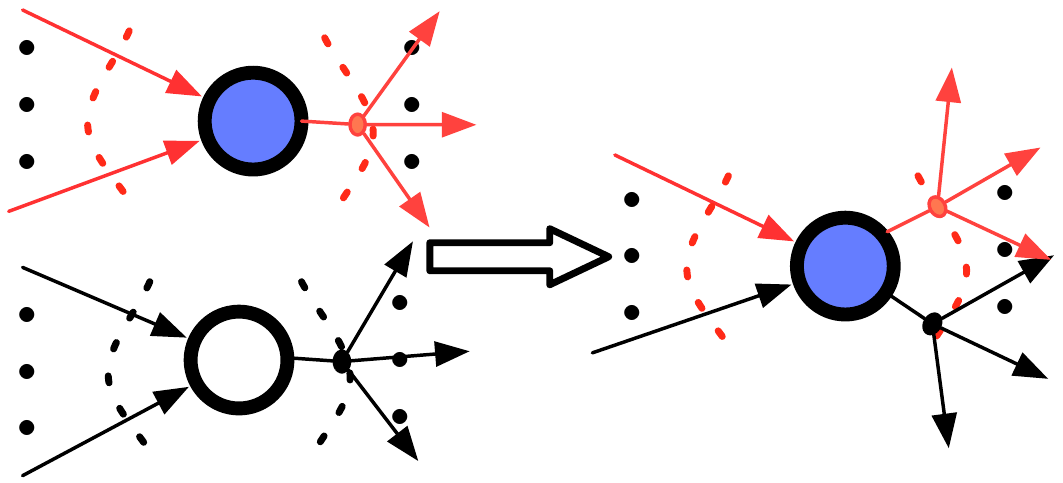} }\hspace{2mm}
\subfloat [\label{fig:exampleec}Edges merge: one extra node and four associated edges are added to replace the two edges.]{\includegraphics[scale=0.5]{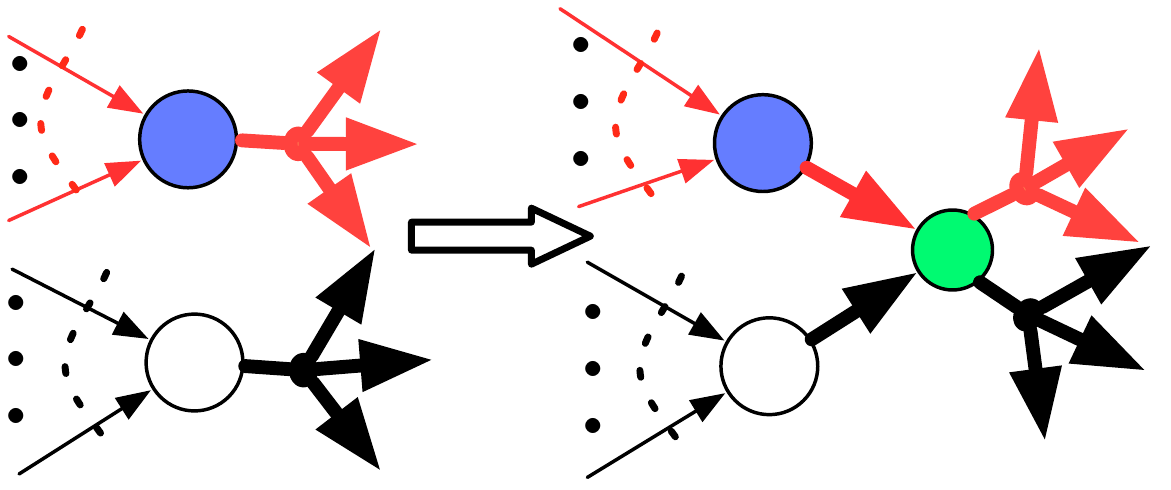} }
\caption{Combination operations on two smaller networks to form a larger network.  Thickly lined nodes (edges) are merged.}
\label{fig:embedexamplegeneral} 
%\vspace{-.6cm}
\end{figure}

%Fig.\ \ref{fig:exampleds} demonstrates the direct sum of two small networks.

The operations we will define will merge network elements, i.e., sources, intermediate nodes, sink nodes, edges, etc, and are depicted in Fig. \ref{fig:embedexamplegeneral}. Since each merge will combine one or several pairs of elements, with each pair containing one element from $\Asf_1$ and the other from $\Asf_2$, each merge definition will involve a bijection $\pi$ indicating which element from the appropriate set of $\Asf_2$ is paired with its argument in $\Asf_1$.

We first consider the sources merge operation, in which the merged sources will function as identical sources for both sub-networks, as shown in Fig.\,\ref{fig:examplesi}.  A sink requiring sources involved in the merge will require the merged source instead.

\begin{figure}
\centering 
\subfloat[\label{fig:combinationexamplesrm}Demonstration of source merge on two networks: sources $s_1,s_3$ are merged to $s_1$, so $s_1$ will send information to both sub-networks.]{\includegraphics[scale=0.4]{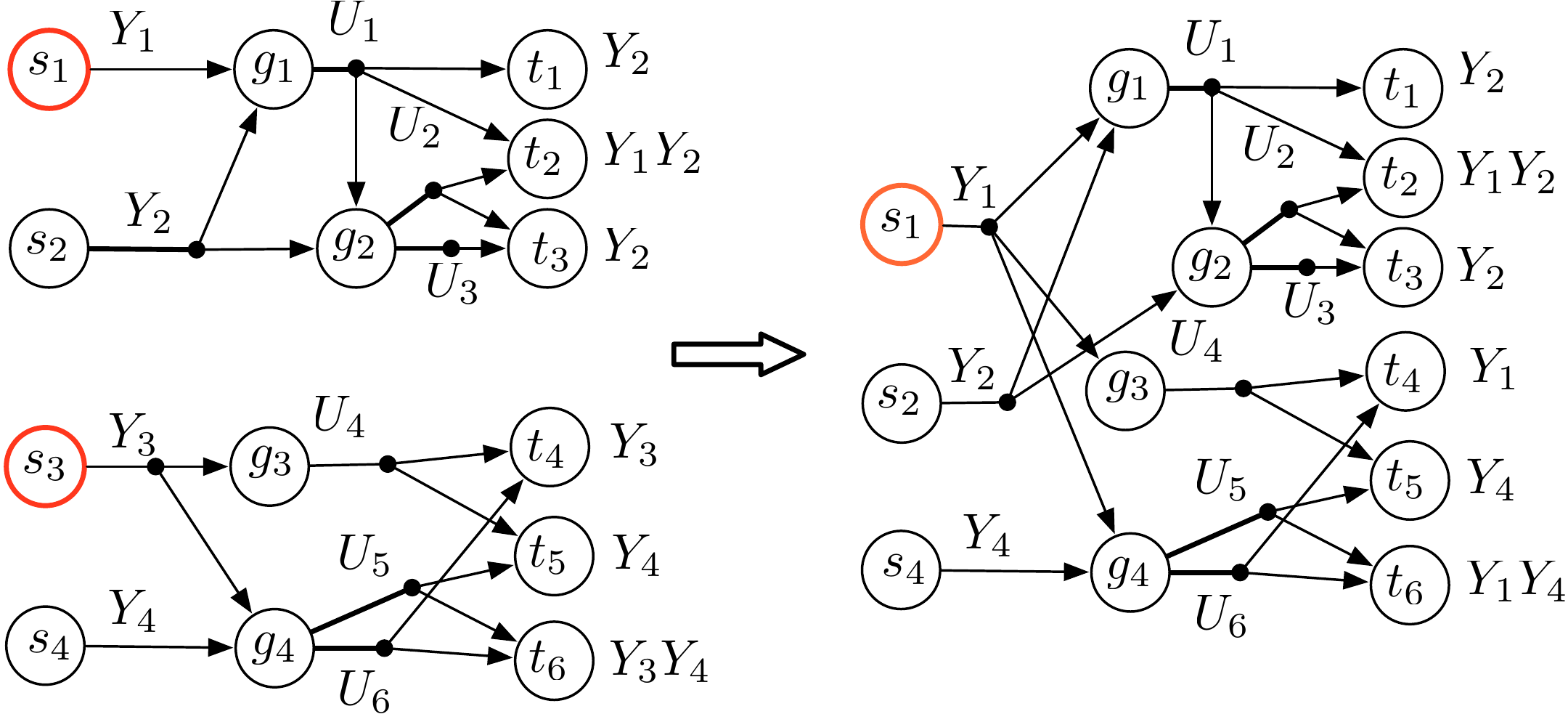} }

\subfloat[\label{fig:combinationexampleskm} Demonstration of sink merge on two networks: sinks $t_2,t_4$ are merged to $t_2$, so their input and demands are combined.]{\includegraphics[scale=0.4]{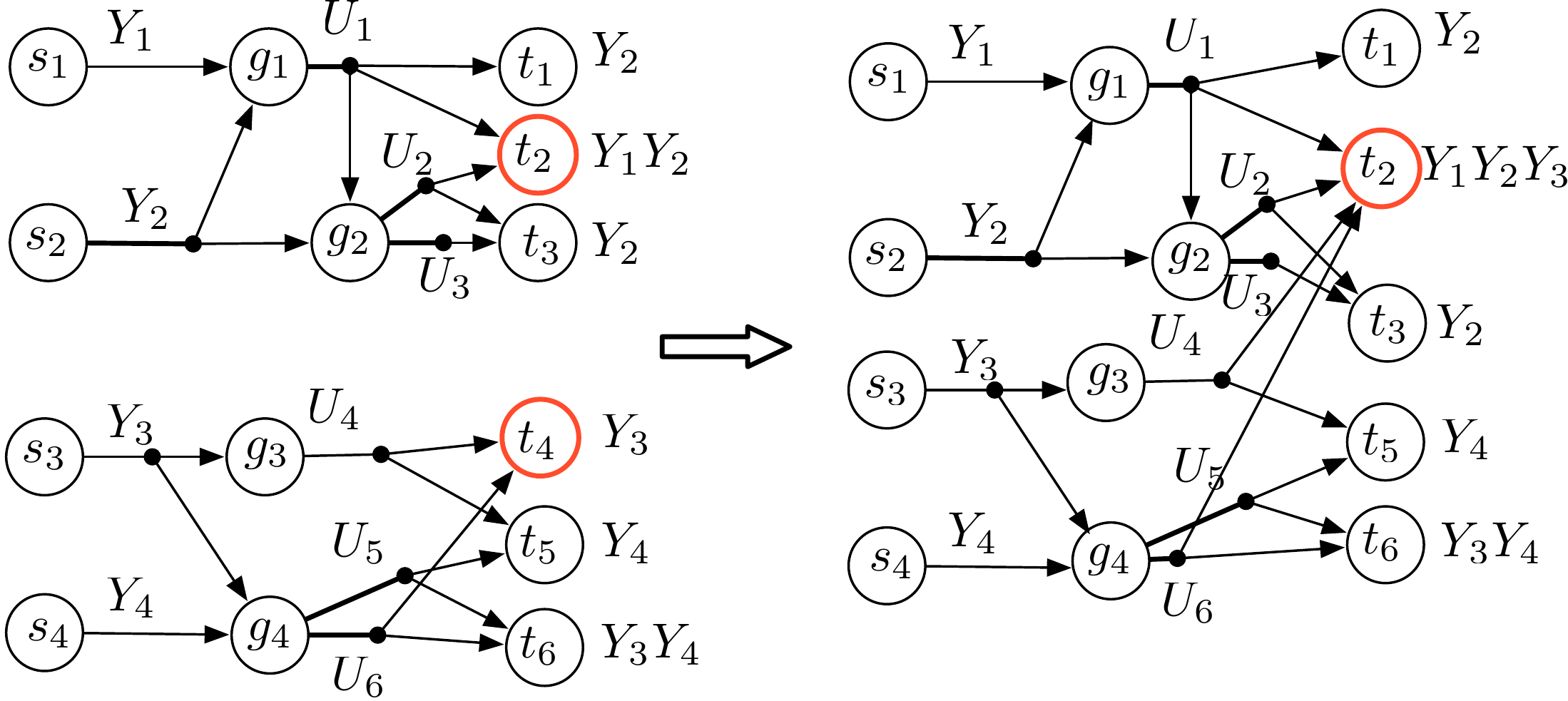} }

\subfloat[\label{fig:combinationexamplend} Demonstration of node merge on two networks: nodes $g_2$ and $g_4$ are merged to node $g_2$, so their input and output are also combined.]{\includegraphics[scale=0.4]{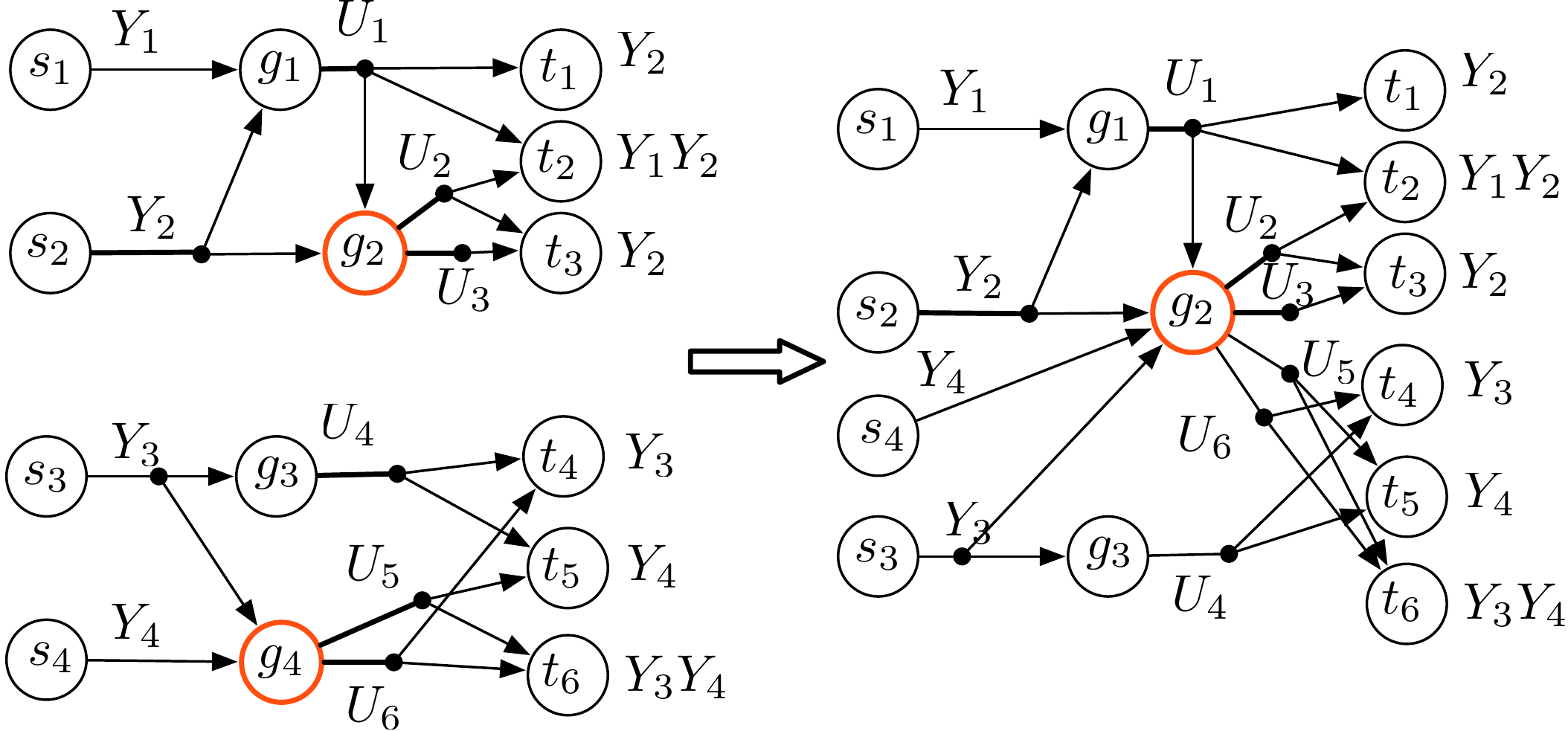} }

\subfloat[\label{fig:combinationexampleegm} Demonstration of edge merge on two networks: when $U_1,U_4$ are merged, one extra node and four edges are added to replace $U_1,U_4$ in the two networks, respectively.]{\includegraphics[scale=0.4]{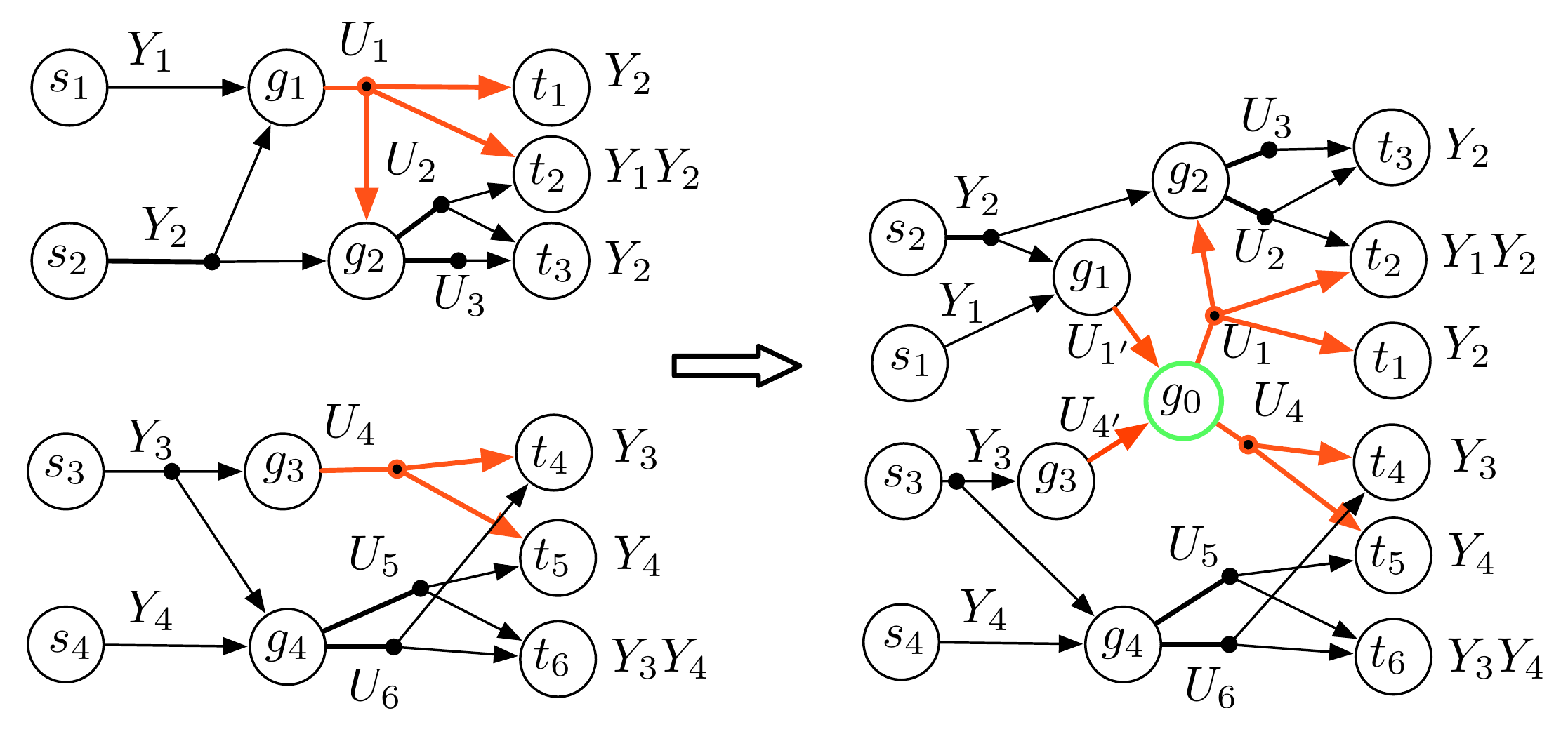} }
\caption{Example to demonstrate combinations of two networks.}
\end{figure}

%Since the sources are assumed to be independent, one may consider merging pairs of sources in two networks.

\begin{definition}[Source Merge $(\Asf_1.\hat{\Smc}=\Asf_2.\pi(\hat{\Smc}))$ --  Fig.\,\ref{fig:examplesi}]
Merging the sources $\hat{\Smc}\subseteq \Smc_1$ from network $\Asf_1$ with the sources $\pi(\hat{\Smc})\subseteq \Smc_2$ from a disjoint network $\Asf_2$, will produce a network $\Asf$ with $i)$ merged sources $\Smc=\Smc_1\cup\Smc_2\setminus\pi(\hat{\Smc})$, $ii)$ $\Gmc=\Gmc_1\cup\Gmc_2$, $iii)$ $\Tmc=\Tmc_1\cup\Tmc_2$, $iv)$ $\Emc=(\Emc_1\cup\Emc_2\setminus \Amc) \cup \Bmc$, where $\Amc=\{e\in\Emc_1\cup\Emc_2|\text{Tl}(e)\in\hat{\Smc}\cup\pi(\hat{\Smc})\}$ includes the edges connected with the sources involved in the merge, $\Bmc=\{(s,\Fmc_1\cup\Fmc_2)|s\in\hat{\Smc},(s,\Fmc_1)\in\Emc_1,(\pi(s),\Fmc_2)\in\Emc_2\}$ includes the new edges connected with the merged sources, and $v)$ updated sink demands
\begin{equation*}
\beta(t) = \left\{ \begin{array}{cc} \beta_1(t) & t \in \mathcal{T}_1 \\
\left(\beta_2(t)\setminus \pi(\hat{S}) \right) \cup \pi^{-1}\left( \pi(\hat{S}) \cap \beta_2(t) \right)  & t\in\mathcal{T}_2 \end{array} \right. .
\end{equation*}
\end{definition}

Fig.\,\ref{fig:combinationexamplesrm} demonstrates the source merge in a network example.

Similar to source merge, we can merge sink nodes of two networks, as demonstrated in Fig.\ \ref{fig:examplesc}.  When two sinks are merged into one sink, we simply union their input and demands as the input and demands of the merged sink. % Formally, we can have definition of sink merge as follows.
\begin{definition}[Sink Merge $(\Asf_1.\hat{\Tmc}+\Asf_2.\pi(\hat{\Tmc}))$ -- Fig.\ \ref{fig:examplesc}.]
Merging the sinks $\hat{\Tmc}\subseteq \Tmc_1$ from network $\Asf_1$ with the sinks $\pi(\hat{\Tmc})\subseteq \Tmc_2$ from the disjoint network $\Asf_2$ will produce a network $\Asf$ with $i)$ $\Smc=\Smc_1\cup\Smc_2$; $\Gmc=\Gmc_1\cup\Gmc_2$, $ii)$ $\Tmc=\Tmc_1\cup\Tmc_2\setminus \pi(\hat{\Tmc})$, $iii)$ $\Emc=\Emc_1\cup\Emc_2\cup\Amc\setminus \Bmc$, where $\Amc=\{(g_2,\Fmc_1\cup\Fmc_2)|g_2\in\Gmc_2,\Fmc_1\subseteq\hat{\Tmc},\Fmc_2\subseteq \Tmc_2,(g_2,\pi(\Fmc_1)\cup\Fmc_2)\in\Emc_2\}$ updates the head nodes of edges in $\Asf_2$ with new merged sinks, $\Bmc=\{(g_2,\Fmc_2)\in\Emc_2|\Fmc_2\cap\pi(\hat{\Tmc})\neq\emptyset\}$ includes the edges connected to sinks in $\pi(\hat(\Tmc))$, and $v)$ updated sink demands
\begin{equation}
\beta(t) = \left\{ \begin{array}{cc} \beta_i(t)  & t\in \Tmc_i \setminus \hat{\Tmc}, i \in\{1,2\} \\
\beta_1(t) \cup \beta_2(\pi(t)) & t \in \hat{\Tmc} \end{array} \right. .
\end{equation}
\end{definition}

Fig.\,\ref{fig:combinationexampleskm}  demonstrates the sink merge in a network example. 

Next, we define intermediate nodes merge.  When two intermediate nodes are merged, we union their incoming and outgoing edges as the incoming and outgoing edges of the merged node, respectively, as illustrated in Fig.\,\ref{fig:examplenc}.  
\begin{definition}[Intermediate Node Merge $(\Asf_1.g+\Asf_2.\pi(g))$ -- Fig.\,\ref{fig:examplenc}]\label{def:NodCom}
Merging the intermediate node $g\in \Gmc_1$ from network $\Asf_1$ with the intermediate node $\pi(g) \in \Gmc_2$ from the disjoint network $\Asf_2$ will produce a network $\Asf$ with $i)$ $\Smc=\Smc_1\cup\Smc_2$, $ii)$ $\Gmc=\Gmc_1\cup\Gmc_2\setminus \pi(g)$, $iii)$ $\Tmc=\Tmc_1\cup\Tmc_2$, $iv)$ $\Emc=\Emc_1\cup\Emc_2\cup\Amc\cup\Bmc\setminus\Cmc\setminus\Dmc$, where $\Amc=\{(g_2,\Fmc_2\setminus \pi(g)\cup g)|g_2\in\Gmc_2, (g_2,\Fmc_2\cup \pi(g))\in\Emc_2\}$ updates the head nodes of edges in $\Asf_2$ that have $\pi(g)$ as head node, $\Bmc=\{(g,\Fmc_2)|(\pi(g),\Fmc_2)\in\Emc_2\}$ updates the tail node of edges in $\Asf_2$ that have $\pi(g)$ as tail node, $\Cmc=\{e\in\Emc_2|\text{Tl}(e)=\pi(g)\}$ includes the edges in $\Asf_2$ that have $\pi(g)$ as tail node, $\Dmc=\{e\in\Emc_2| \pi(g)\in\text{Hd}(e)\}$ includes the edges in $\Asf_2$ that have $\pi(g)$ as head node; and $v)$ updated sink demands
\begin{equation}\label{eq:refMe}
\beta(t) = \left\{ \begin{array}{cc} \beta_1(t) & t\in\Tmc_1 \\ \beta_2(t) & t\in\Tmc_2 \end{array} \right.
\end{equation}
\end{definition}
Fig.\,\ref{fig:combinationexamplend} demonstrates the node merge in a network example.
%The essence of this definition is to union the inputs and outputs of the two nodes when they are merged, as is illustrated in Fig.\,\ref{fig:examplenc}.  
%Observe that we only defined the merging of one intermediate node from each of the two disjoint networks, as merging multiple pairs of nodes from the two networks may not enable the resulting rate region to be calculated in terms of the two smaller networks rate regions, and, additionally, may not preserve the sufficiency of classes of codes or the tightness of outer bounds.
%To merge multiple pairs of intermediate nodes in two smaller networks can be decomposed to two steps: I) merge one pair of nodes to form a larger network; II) then merge pairs of intermediate nodes in the larger network obtain by step I) such that two nodes in a pair are from the two smaller networks respectively.  Theorem \ref{thm:NodCom} shows that step I) can preserve the sufficiency of a class of codes.  Whether Step II) preserves the sufficiency of a class of codes or not is under question since there may be new paths for a sink node after merging several nodes.  The definition of merging multiple intermediate nodes is omitted.  For operations inside a network that preserve the sufficiency of a class of codes, interested readers are referred to \cite{CongduanAllerton2014,CongduanTranIT2014}.

Finally, we define edge merge.   As demonstrated in Fig.\ \ref{fig:exampleec}, when two edges are merged, one new node and four new edges will be added to create a "cross" component so that the transmission will be in the new component instead of the two edges being merged. % With accurate math notations, the definition of edge merge is as follows.

\begin{definition}[Edge Merge $(\Asf_1.e+\Asf_2.\pi(e))$ -- Fig.\ \ref{fig:exampleec}]\label{def:EdgCom}
Merging edge $e\in\Emc_1$ from network $\Asf_1$ with edge $\pi(e)\in\Emc_2$ from disjoint network $\Asf_2$ will produce a network $\Asf$ with $i)$  $\Smc=\Smc_1\cup\Smc_2$, $ii)$  $\Gmc=\Gmc_1\cup\Gmc_2\cup g_0$, where $g_0\notin \Gmc_1,g_0\notin\Gmc_2 $, $iii)$ $\Tmc=\Tmc_1\cup\Tmc_2$, $iv)$ $\Emc=(\Emc_1\setminus e ) \cup(\Emc_2\setminus\pi(e)) \cup\{(\text{Tl}(e), g_0),  (\text{Tl}(\pi(e)),g_0), (g_0,\text{Hd}(e)), (g_0,\text{Hd}(\pi(e)))\}$; and $v)$ updated sink demands given by (\ref{eq:refMe}).
%\begin{enumerate}
%\item $\Smc=\Smc_1\cup\Smc_2$;
%\item $\Gmc=\Gmc_1\cup\Gmc_2\cup g_0$, where $g_0\notin \Gmc_1,g_0\notin\Gmc_2 $;
%\item $\Tmc=\Tmc_1\cup\Tmc_2$;
%\item $\Emc=(\Emc_1\setminus e ) \cup(\Emc_2\setminus\pi(e)) \cup\{(\text{Tl}(e),g_0),  (\text{Tl}(\pi(e)),g_0), (g_0,\text{Hd}(e)), (g_0,\text{Hd}(\pi(e)))\}$;
%\item If $t\in\Tmc_1$, $\beta(t)=\beta_1(t)$. If $t\in\Tmc_2$, $\beta(t)=\beta_2(t)$.
%\end{enumerate}
\end{definition}

It is not difficult to see that this edge merge operation can be thought of as a special node merge operation. Suppose the edges being merged are $\Asf_1.e$, $\Asf_2.\pi(e)$.  If two virtual nodes $g_1,g_2$ are added on $e,\pi(e)$, respectively, splitting them each into two edges, so that $e,\pi(e)$ go into and flow out $g_1,g_2$, respectively, then, the merge of $g_1,g_2$ gives the same network as merging $e,\pi(e)$.  Fig.\,\ref{fig:combinationexampleegm} demonstrates the edge merge in a network example.

%All these operations defined above can be used to combine two smaller networks into a larger network. Next, we would like to show that these operations are able to preserve some properties, such as sufficiency of linear codes and tightness of Shannon outer bound.
%\vspace{-2mm}

%\section{Rate Regions Resulting from the Operations}\label{sec:theorems}
%\vspace{-2mm}
\subsection{Preservation Properties of Combination Operations}

Here we prove that the combination operations enable the rate regions of the small networks to be combined to produce the rate region of the resulting large network, and also preserve sufficiency of classes of codes and tightness of other bounds.

%Here we show that the sufficiency of a class of linear network coding and the tightness of Shannon outer bound are inherited in the larger network from the smaller network under the combination operations presented in the previous section.  

%\vspace{-0.2cm}
\begin{theorem} \label{thm:SrcIdn}
Suppose a network $\Asf$ is obtained by merging $\hat{\Smc}$ with $\pi(\hat{\Smc})$, i.e., $\Asf_1.\hat{\Smc}=\Asf_2.\pi(\hat{\Smc})$. Then
\begin{equation}
\Rmc_l(\Asf)={\rm Proj}( (\Rmc_l(\Asf_1)\times \Rmc_l(\Asf_2)) \cap \Lmc_0), \  l \in\{\ast,q,(s,q),o\}\label{eq:SrcIdneq1}
\end{equation}
with $
\Lmc_0=\left\{H(Y_{s})=H(Y_{\pi(s)}),\forall s\in\hat{\Smc}\right\}$,
and the dimensions kept in the projection are $(H(Y_s),s\in\Smc)$ and $(R_{e},e\in\Emc)$, where $\Smc,\Emc$ represent the source and edge sets of the merged network $\Asf$, respectively.
\end{theorem}
\begin{remark}\label{rem1}
The inequality description of the polyhedral cone ${\rm Proj}( (\mathcal{P}_1\times \mathcal{P}_2) \cap \Lmc_0)$ for two polyhedral cones $\mathcal{P}_j,j\in\{1,2\}$ can be created by concatenating the inequality descriptions for $\mathcal{P}_1$ and $\mathcal{P}_2$, then replacing the variable $H(X_{\pi(s)})$ with the variable $H(X_s)$ for each $s\in\hat{\Smc}$.
\end{remark}

%\begin{IEEEproof}
\noindent \emph{Proof:} Select any point $\Rbf\in \Rmc_*(\Asf)$.  Then there exists a conic combination of some points in $\Rmc_*(\Asf)$ that are associated with entropic vectors in $\Gamma_{N}^*$ such that $\Rbf=\sum\limits_{\rbf_j \in\Rmc_*(\Asf)}\alpha_j \rbf_j$, where $\alpha_j\geq 0, \forall j$.  For each $\rbf_j$, there exist random variables $\Ybf^{(j)}_{\Smc},U^{(j)}_i , i\in \Emc\setminus e$, such that the entropy vector 
\begin{equation*}
\hbf^{(j)}=\left[H(\Amc) \left| \Amc\subseteq \left\{Y_s^{(j)},U_i^{(j)} \left| s\in\Smc,i\in\Emc\right.\right\}\right.\right]
\end{equation*}
 is in $\Gamma_{N}^*$, where $N$ is the number of variables in $\Asf$.  Furthermore, their entropies satisfy all the constraints determined by $\Asf$. When decomposing $\Asf$ into $\Asf_1,\Asf_2$, let i.i.d. copies of variables $Y_s^{(j)},s\in\hat{\Smc}$ work as sources $\pi(\hat{\Smc})\subseteq\Smc_2$.  The associated edges connecting $\hat{\Smc}$ and nodes in $\Gmc_2$ will then connect $\pi(\hat{\Smc})$ and nodes in $\Gmc_2$.  Then the random variables $\{Y_s^{(j)},U_i^{(j)} | s\in\Smc_1,i\in\Emc_1\},\{Y_s^{(j)},U_i^{(j)}| s\in\Smc_2,i\in\Emc_2\}$ will satisfy the network constraints determined by $\Asf_1,\Asf_2$, and also $\mathcal{L}_0$.  Thus, $\Rbf\in{\rm Proj}( (\Rmc_*(\Asf_1)\times \Rmc_*(\Asf_2)) \cap \Lmc_0)$. %, so
%\begin{equation} 
%\Rmc(\Asf)\subseteq{\rm Proj}( (\Rmc(\Asf_1)\times \Rmc(\Asf_2)) \cap \Lmc_0).\label{eq:pfeq1SrcIdn1}
%\end{equation}
Similarly, if $\Rbf$ is achievable by $\Fbb_q$ codes, vector or scalar, the same code applied to the part of $\Asf$ that is $\Asf_1,\Asf_2$ will achieve $\Rbf_1,\Rbf_2$, respectively.  Putting these together, we have
%\begin{equation}
$\Rmc_l(\Asf)\subseteq
{\rm Proj}( (\Rmc_l(\Asf_1)\times \Rmc_l(\Asf_2)) \cap \Lmc_0),  l\in\{\ast,q,(s,q)\}.$ %\label{eq:pfeq1SrcIdn2}
%\end{equation}

Next, if we select two points $\Rbf_1\in \Rmc_*(\Asf_1)$, $\Rbf_2\in \Rmc_*(\Asf_2)$ such that $H(Y_{s})=H(Y_{\pi(s)}),\forall s\in\hat{\Smc}$, then there exist conic combinations $\Rbf_1=\sum\limits_{\rbf_{i,j} \in\Rmc_*(\Asf_i)}\alpha_{i,j} \rbf_{i,j}$ for $i=1,2$, and for each $\rbf_{i,j}$ there exist a set of variables associated with sources and edges.  Since $H(Y_{s})=H(Y_{\pi(s)})$ and sources are independent and uniformly distributed, we can let the associated variables $Y_{s}^{(j)}$ and $Y_{\pi(s)}^{(j)}$ be the same variables.  Then, after combination, the entropy vector of all variables $\{Y_{s}^{(j)},U_e^{(j)}| s\in\Smc,e\in\Emc\}$ will be in $\Gamma_N^*$.  Furthermore, their entropies, together with the rate vectors from $\Rbf_1,\Rbf_2$, will satisfy all network constraints of $\Asf$, and there will be an associated point $\rbf=\rbf_1\times\rbf_2$ with $\Lmc_0$.  Using the same conic combination, we will find the associated point $\Rbf=\Rbf_1\times\Rbf_2 \cap \Lmc_0$.  Hence,  ${\rm Proj}( (\Rmc_*(\Asf_1)\times \Rmc_*(\Asf_2)) \cap \Lmc_0)\subseteq \Rmc_*(\Asf)$.  Now suppose there exists a sequence of network codes for $\Asf_1$ and $\Asf_2$ achieving $\Rbf_1,\Rbf_2$.  By using the same source bits as the source inputs for $s$ in $\Asf_1$ and $\pi(s)$ in $\Asf_2$ for each $s\in\hat{S}$, we have the same effect as using these source bits as the inputs for $s$ in the source merged $\Asf$ and achieving the associated rate vector $\Rbf$, implying $\Rbf \in \Rmc_l(\Asf),\ l\in\{q,(s,q)\}$, and hence  $\Rmc_{s,q}(\Asf)\supseteq{\rm Proj}( (\Rmc_{s,q}(\Asf_1)\times \Rmc_{s,q}(\Asf_2)) \cap \Lmc_0)$ and $\Rmc_{q}(\Asf)\supseteq{\rm Proj}( (\Rmc_{q}(\Asf_1)\times \Rmc_{q}(\Asf_2)) \cap \Lmc_0)$.  Together with the statements above, this proves (\ref{eq:SrcIdneq1}) for $l\in\{\ast,q,(s,q)\}$.

Furthermore, by (\ref{eq:regionoutfree}), any point $\Rbf \in \mathcal{R}_o(\Asf)$, is the projection of some point $[\hbf,\rbf] \in \Gamma_N \cap \mathcal{L}_{\Asf}$, where $\hbf\in\Gamma_N$ and $\rbf=[R_e | e\in\Emc]$. Because the Shannon inequalities and network constraints in $\Gamma_N \cap \mathcal{L}_{\Asf}$ form a superset (i.e., include all of) of the network constraints in $\Gamma_N \cap \mathcal{L}(\Asf_i)$, the subvectors $[\hbf^i,\rbf^i]$ of $[\hbf,\rbf]$  associated only with the variables in $\Asf_i$ (with $Y_{\pi(s)}$ being recognized as $Y_{s}$ for all $s\in\hat{S}$) are in $\Gamma_{N_i}\cap\mathcal{L}(\Asf_i)$ and obey $\mathcal{L}_0$, implying $\Rbf \in{\rm Proj}( (\Rmc_o(\Asf_1)\times \Rmc_o(\Asf_2)) \cap \Lmc_0)$, and hence 
%\begin{equation}\label{eq:outerF}
$\Rmc_{o}(\Asf)\subseteq {\rm Proj}( (\Rmc_o(\Asf_1)\times \Rmc_o(\Asf_2)) \cap \Lmc_0)$.
%\end{equation} 

Next, if we select two points $\Rbf_1\in \Rmc_o(\Asf_1)$, $\Rbf_2\in \Rmc_o(\Asf_2)$ such that $H(Y_{s})=H(Y_{\pi(s)}),\forall s\in\hat{\Smc}$, then there exists $[\hbf^{i},\rbf^{i}] \in\Gamma_{N_i} \cap \mathcal{L}(\Asf_i)$, where $\hbf^i\in\Gamma_{N_i}$ and $\rbf^i=[R_e | e\in\Emc_i]$, such that $\Rbf_i = {\rm Proj}_{\boldsymbol{r}_i,\boldsymbol{\omega}_i} [\hbf^i,\rbf^i]$, $i\in\{1,2\}$ with $h^1_{X_s} = h^2_{X_{\pi(s)}}$ for all $s\in\hat{S}$.  Define $\hbf$ whose element associated with the subset $\mathcal{A}$ of $\mathcal{N} = \Smc\cup \Emc $ is
%\begin{equation}\label{eq:sumr}
$h_{\mathcal{A}} = h^1_{\mathcal{A} \cap \mathcal{N}_1} + h^2_{\mathcal{A}\cap \mathcal{N}_2} - h^2_{\mathcal{A}\cap \pi(\hat{\Smc})}$
%\end{equation}
where $\mathcal{N}_i = \Smc_i \cup \Emc_i$, $i\in\{1,2\}$.  By virtue of its creation this way, this function is submodular and $\hbf \in\Gamma_N$.
% more detail about this would be better -- do it in journal version
Since the two networks are disjoint, the list of equalities in $\mathcal{L}_3(\Asf)$ is simply the concatenation of the lists in $\mathcal{L}_3(\Asf_1)$ and $\mathcal{L}_3(\Asf_2)$, each of which involved inequalities in disjoint variables $\Nmc_1$ and $\Nmc_2$, and the same thing holds for $\mathcal{L}_{4'}$ with consideration of $\rbf^i$.  Furthermore, since $\hbf^i \in \Lmc_2(\Asf_i)$ and $h^1_{Y_s}=h^2_{Y_{\pi(s)}}, s\in\hat{\Smc}$, $\hbf$ obeys $\mathcal{L}_2(\Asf)$.
%again, would be better to express this out further -- should be done in journal version
The definition of $\hbf$, together with $\hbf^i\in\mathcal{L}_1(\Asf_i), \ i\in\{1,2\}$ and $h^1_{Y_s} = h^2_{Y_{\pi(s)}}, \ s\in\hat{\Smc}$, implies that $\hbf\in\mathcal{L}_1(\Asf)$.  Finally $\hbf^1\in\mathcal{L}(\Asf_1)$ and $\hbf^2 \in\Lmc(\Asf_2)$ imply $\hbf\in\Lmc_5(\Asf)$.  Putting these facts together we observe that 
$[\hbf,\rbf]\in\Gamma_N \cap \mathcal{L}_{\Asf}$, so $\Rbf \in \Rmc_o(\Asf)$, implying $\Rmc_{o}(\Asf)\supseteq {\rm Proj}( (\Rmc_o(\Asf_1)\times \Rmc_o(\Asf_2)) \cap \Lmc_0)$.
%again, more detail would be better
\hfill $\blacksquare$
%\end{IEEEproof}

\begin{theorem} \label{thm:SnkCom}
Suppose a network $\Asf$ is obtained by merging sink nodes $\hat{\Tmc}$ with $\pi(\hat{\Tmc})$, i.e., $(\Asf_1.\hat{\Tmc}+\Asf_2.\pi(\hat{\Tmc}))$.  Then
\begin{equation}
\Rmc_l(\Asf)= \Rmc_l(\Asf_1)\times \Rmc_l(\Asf_2) , \ l \in\{\ast,q,(s,q),o\} \label{eq:SnkComeq1}
\end{equation}
with the index on the dimensions mapping from $\{e\in\Emc_2| \text{Hd}(e)\in\pi(\hat{\Tmc})\}$ to $\{e\in\Emc| \text{Hd}(e)\in\hat{\Tmc},\text{Tl}(e)\in \Gmc_2\}$.
\end{theorem}
%\begin{IEEEproof}
\noindent \emph{Proof:}  Consider a point $\Rbf\in\Rmc_l(\Asf)$ with conic combination of $\Rbf=\sum\limits_{\rbf_j\in \Rmc_l(\Asf)}\alpha_{l,j} \rbf_{l,j}$, where $\alpha_{l,j}\geq 0$ for any $j$ and $l\in\{\ast,q,(s,q),o\}$.  Each $\rbf_{l,j}$ has associated  random variables or the associated codes.  Due to the independence of sources in networks $\Asf_1,\Asf_2$, and the fact that their sources and intermediate nodes are disjoint, the variables arriving at a merged sink node from $\Asf_1$ will be independent of the sources in $\Asf_2$ and the variables arriving at a merged sink node from $\Asf_2$ will be independent of the sources in $\Asf_1$.  In particular, Shannon type inequalities imply the Markov chains
$H(\Ybf_{\Smc_1}|\Ubf_{\text{In}(t)\cap\Emc_1},\Ubf_{\text{In}(t)\cap\Emc_2}) =H(\Ybf_{\Smc_1}|\Ubf_{\text{In}(t)\cap\Emc_1})$
and $H(\Ybf_{\Smc_2}|\Ubf_{\text{In}(t)\cap\Emc_1},\Ubf_{\text{In}(t)\cap\Emc_2}) =H(\Ybf_{\Smc_2}|\Ubf_{\text{In}(t)\cap\Emc_2})$ for all $t\in\Tmc$ (even if the associated ``entropies'' are only in $\Gamma_N$ and not necessarily $\bar{\Gamma}^*_N$).  This then implies, together with the independence of the sources, that $H(\Ybf_{\beta(t)}|\Ubf_{\text{In}(t)}) = H(\Ybf_{\beta(t)\cap\Smc_1}|\Ubf_{\text{In}(t)\cap\Emc_1})+H(\Ybf_{\beta(t)\cap\Smc_2}|\Ubf_{\text{In}(t)\cap\Emc_2})$, showing that the constraints in $\mathcal{L}_5(\Asf)$ imply the constraints in $\mathcal{L}_5(\Asf_1)$ and $\mathcal{L}_5(\Asf_2)$. Furthermore, given the disjoint nature of $\Asf_1$ and $\Asf_2$, the constraints in $\mathcal{L}_i(\Asf)$, are simply the concatenation of the constraints in $\mathcal{L}_i(\Asf_1)$ and $\mathcal{L}_i(\Asf_2)$, for $i\in\{2,3,4'\}$.  Furthermore, the joint independence of all of $\Ybf_{\Smc_1},\Ybf_{\Smc_2}$ imply the marginal independence of the collections of variables $\Ybf_{\Smc_1}$ and $\Ybf_{\Smc_2}$, so that $\mathcal{L}_1(\Asf)$ implies $\mathcal{L}_1(\Asf_i),i\in\{1,2\}$.  This shows that $\rbf_{l,j}\in \rbf_{l,j}^{1}\times \rbf_{l,j}^{2}$ and further $\Rbf \in \Rmc_l(\Asf_1)\times\Rmc_l(\Asf_2)$, and hence $\Rmc_l(\Asf) \subseteq \Rmc_l(\Asf_1)\times\Rmc_l(\Asf_2), l \in\{\ast,q,(s,q),o\}$.

Next, consider two points $\Rbf_i \in\Rmc_l(\Asf_i),\ i\in\{1,2\}$ for any $l\in\{q,(s,q),o\}$.  By definition these are projections of $[\hbf^i,\rbf^i] \in \Gamma^q_{N_i,\infty}\cap \Lmc(\Asf_i)$, $[\hbf^i,\rbf^i]  \in \Gamma^q_{N_i}\cap \Lmc(\Asf_i)$, $[\hbf^i,\rbf^i] \in \Gamma_{N_i}\cap \Lmc(\Asf_i)$, respectively, for $i\in\{1,2\}$, where $\hbf^i\in\Gamma_{N_i}$ and $\rbf^i=[R_e | e\in\Emc_i]$.  Define $\hbf$ with value associated with subset $\Amc\subseteq \Nmc$ of $
h_{\Amc} = h^1_{\Amc\cap \Nmc_1} + h^2_{\Amc\cap \Nmc_2}$, then it is easily verified that the resulting $[\hbf, \rbf^1,\rbf^2]\in \Gamma^q_{N,\infty}\cap \Lmc_{\Asf}$, $[\hbf, \rbf^1,\rbf^2] \in \Gamma^q_{N}\cap \Lmc_{\Asf}$, $[\hbf, \rbf^1,\rbf^2] \in \Gamma_{N}\cap \Lmc_{\Asf}$, respectively, (simply use the same codes from $\Asf_1$ and $\Asf_2$ on the corresponding parts of $\Asf$).  Since $\Rbf = {\rm Proj}_{\boldsymbol{\omega},\boldsymbol{r}} [\hbf,\rbf^1,\rbf^2]$, we have proven $\Rbf\in\Rmc_l(\Asf)$, and hence that $\Rmc_l(\Asf) \supseteq \Rmc_l(\Asf_1)\times \Rmc_l(\Asf_2)$.  Further, for two points $\Rbf_i \in\Rmc_*(\Asf_i),\ i\in\{1,2\}$, there exist a conic combination of $\rbf^i_j$, $\Rbf_i=\sum\limits_{\rbf^i_j\in\Rmc_*(\Asf_i)}\alpha_j^i \rbf^i_j$, with associated random variables $\left\{Y_s^{(j)},U_i^{(j)} | s\in\Smc_1,i\in\Emc_1\right\},\left\{Y_s^{(j)},U_i^{(j)} | s\in\Smc_2,i\in\Emc_2\right\}$ satisfying the network constraints determined by $\Asf_1,\Asf_2$.  Due to the independence of sources and disjoint of edge variables, the union of variables in $\Asf_1,\Asf_2$ will satisfy the network constraints in the merged $\Asf$.  With the same conic combinations, we have $\Rbf=\sum\limits_{\rbf^i_j\in\Rmc_*(\Asf_i)} [\alpha_j^1 \rbf^1_j,\alpha_j^2 \rbf^2_j]\in\Rmc_*(\Asf)$.  Thus, $\Rmc_*(\Asf) \supseteq \Rmc_*(\Asf_1)\times \Rmc_*(\Asf_2)$. \hfill $\blacksquare$
%Putting these facts together with the previously proven containment completes the proof.

\begin{theorem} \label{thm:NodCom}
Suppose a network $\Asf$ is obtained by merging $g$ and $\pi(g)$, i.e., $\Asf_1.g+\Asf_2.\pi(g)$.  Then
\begin{equation}
\Rmc_l(\Asf)= \Rmc_l(\Asf_1)\times\Rmc_l(\Asf_2),\ l\in\{\ast,q,(s,q),o\} \label{eq:NodComeq1}
\end{equation}
with dimensions/ indices mapping from $\{e\in\Emc_2| \text{Hd}(e)=\pi(g)\}$ to $\{e\in\Emc| \text{Hd}(e)=g,\text{Tl}(e)\in \Gmc_2\}$ and from 
$\{e\in\Emc_2|\text{Tl}(e)=\pi(g)\}$ to $\{e\in\Emc| \text{Tl}(e)=g,\text{Hd}(e)\in \Gmc_2\}$.
\end{theorem}
%\begin{IEEEproof}
\vspace{-2mm}
\noindent \emph{Proof:} 
Consider a point $\Rbf\in\Rmc_l(\Asf)$ for any $l\in\{*,q,(s,q)\}$ and all random variables associated with each component $\rbf_{l,j}$ in the conic combinations $\Rbf=\sum\limits_{\rbf_j\in \Rmc_l(\Asf)}\alpha_{l,j} \rbf_{l,j}$, where $\alpha_{l,j}\geq 0$ for any $j$ and $l\in\{\ast,q,(s,q),o\}$.   The associated variables satisfy $\Lmc_i(\Asf),i=1,3,4',5$.   Partition the incoming edges of the merged node $g$ in $\Asf$, $\textrm{In}(g)$, up into $\textrm{In}_1(g) = \textrm{In}(g)\cap \Emc_1$ the edges from $\Asf_1$, and $\textrm{In}_2(g)=\textrm{In}(g) \setminus \text{In}_1(g)$, the new incoming edges resulting from the merge.  Similarly, partition the outgoing edges $\textrm{Out}(g)$ up into $\textrm{Out}_1(g) = \textrm{Out}(g)\cap \Emc_1$ and $\textrm{Out}_2(g) = \textrm{Out}(g) \setminus \textrm{Out}_1(g)$.  The $\Lmc_3$ constraints dictate that there exist functions $f_e$ such that for each $e\in\textrm{Out}(g)$, $U_e=f_e(U_{\textrm{In}_1(g)},U_{\textrm{In}_2(g)}) $.  Define the new functions $f_e'$ via
\begin{equation}\label{eq:code}
f_e'(U_{\textrm{In}_1(g)},U_{\textrm{In}_2(g)}) = \left\{ \begin{array}{cc} 
f_e(U_{\textrm{In}_1(g)},\boldsymbol{0}) & e \in \textrm{Out}_1(g) \\
f_e(\boldsymbol{0},U_{\textrm{In}_2(g)}) & e \in \textrm{Out}_2(g) 
\end{array} \right.
\end{equation}
i.e., set the possible value for the incoming edges from the other part of the network (possibly erroneously) to a particular constant value among their possible values -- let's label it $\boldsymbol{0}$.  The network code using these new functions $f_e'$ will utilize the same rates as before.  The constraints and the topology of the merged network further dictated that $U_{\textrm{In}_i(g)}$ were expressible as a function of $\Smc_i$, $i\in\{1,2\}$.   In the remainder of the network (moving toward the sink nodes) after the merged nodes, at no other point is any information from the sources in the other part of the network encountered, and the decoders at the sink nodes in $\Tmc_2$ need to work equally well decoding subsets of $\Smc_2$, regardless of the value of $\Smc_1$.  Since the erroneous value for the $U_{\textrm{In}_1(g)}$ used for $f_e', e \in \textrm{Out}_2(g)$ was still a valid possibility for some (possible other) value(s) of the sources in $\Smc_1$, the sinks must still produce the correct values for their subsets of $\Smc_2$.  A parallel argument for $\Tmc_1$ shows that they still correctly decode their sources, which were subsets of $\Smc_1$, even though the $f_e$s were changed to $f_{e}'$s.  Note further that (\ref{eq:code}) will still be scalar/vector linear if the original $f_e$s were as well.

However, since the $f_{e'}$s no longer depend on the other half of the network, the resulting code can be used as separate codes for $\Asf_1$ and $\Asf_2$, given the associated rate points $\Rbf_i$ by keeping the elements in $\Rbf$ associated with $\Asf_i$, $i\in\{1,2\}$ (or the associated rate points $\rbf_{l,j}^i$ by keeping elements in $\rbf_{l,j}$) in the natural way, implying that $\Rbf \in \Rmc_l(\Asf_1)\times\Rmc_l(\Asf_2)$.  This then implies that $\Rmc_l(\Asf) \subseteq \Rmc_l(\Asf_1)\times\Rmc_l(\Asf_2)$ for all $l\in\{\ast,q,(s,q)\}$. The opposite containment is obvious, since any rate points or codes for the two networks can be utilized in the trivial manner for the merged network.  This proves (\ref{eq:NodComeq1}) for $l\in\{\ast,q,(s,q)\}$. 

Next, consider any pair $\Rbf_i \in \Rmc_o(\Asf_i)$ $i\in\{1,2\}$, which are, by definition, projections of some $[\hbf^i,\rbf^i] \in \Gamma_{N_i} \cap \Lmc(\Asf_i)$, where $\hbf^i\in\Gamma_{N_i}$ and $\rbf^i=[R_e | e\in\Emc_i]$, $\ i\in\{1,2\}$.  Defining $\hbf$ whose element associated with the subset $\Amc \subset \Nmc$ is $h_{\Amc} = h^1_{\Amc \cap \Nmc_1} + h^2_{\Amc \cap \Nmc_2}$, where the intersection respects the remapping of edges under the intermediate node merge, we observe that $[\hbf,\rbf^1,\rbf^2]\in\Gamma_N \cap \Lmc_{\Asf}$, and hence its projection $\Rbf \in \Rmc_o(\Asf)$, proving $\Rmc_o(\Asf) \supseteq \Rmc_o(\Asf_1)\times \Rmc_o(\Asf_2)$.  

Finally, consider a point $\Rbf \in \Rmc_{o}(\Asf)$, which is a projection of some $[\hbf,\rbf] \in \Gamma_N \cap \Lmc_{\Asf}$, where $\hbf\in\Gamma_{N}$ and $\rbf=[R_e | e\in\Emc]$.  For every $\Amc \subseteq \Nmc_i$, define $h^i_{\Amc} = h_{\Amc \cup \Smc_{3-i}} - h_{\Smc_{3-i}}$, and define $\hbf'$ with $h'_{\Amc} = h^1_{\Amc\cap\Nmc_1} + h^2_{\Amc \cap \Nmc_2}$ and $\Rbf' = \textrm{proj}_{\boldsymbol{\omega},\boldsymbol{r}}\hbf'$.  We see that $\hbf^i \in \Lmc(\Asf_i), i \in\{1,2\}$, because conditioning reduces entropy and entropy is non-negative, but all of the conditional entropies at nodes other than $g$ were already zero, while at $g$, the conditioning on the sources from the other network will enable the same conditional entropy of zero since the incoming edges from the other network were functions of them.  This shows that $\Rbf' \in \Rmc_o(\Asf_1)\times \Rmc_o(\Asf_2)$.  Owing to the independence of the sources $\textrm{proj}_{\boldsymbol{\omega}} \hbf=\textrm{proj}_{\boldsymbol{\omega}} \hbf'$, while $\textrm{proj}_{\boldsymbol{r}} \hbf \geq \textrm{proj}_{\boldsymbol{r}} \hbf'$ due to the fact that conditioning reduces entropy.  The coordinate convex nature then implies that $\Rbf \in \Rmc_o(\Asf_1)\times \Rmc_o(\Asf_2)$ showing that $\Rmc_o(\Asf) \subseteq \Rmc_o(\Asf_1) \times \Rmc_o(\Asf_2)$ and completing the proof.
\hfill $\blacksquare$

\begin{theorem} \label{thm:EdgCom}
Suppose a network $\Asf$ is obtained by merging $e$ and $\pi(e)$, i.e., $\Asf_1.e+\Asf_2.\pi(e)$.  Then
\begin{equation}
\Rmc_l(\Asf)={\rm Proj}_{\setminus\{e,\pi(e)\}}((\Rmc_l(\Asf_1)\times\Rmc_l(\Asf_2))\cap\Lmc'_0), \ l\in\{*,q,(s,q),o\} \label{eq:EdgComeq1}
\end{equation}
where $\Lmc'_0=\left\{R_{({\rm Tl}(j),g_0)}\geq R_j,R_{(g_0,{\rm Hd}(j))}\geq R_j,j\in\{e,\pi(e)\}\right\}$, and projection dimension $\setminus\{e,\pi(e)\}$ means projecting out dimensions associated with $e,\pi(e)$.  Furthermore,  $\Rmc(\Asf_q)\times \Rmc(\Asf_q)$ and $\Lmc'_{0}$ are viewed in the dimension of $|\Nmc_1|+|\Nmc_2|+4$ with assumption that all dimensions not shown are unconstrained.
\end{theorem}
%\begin{IEEEproof}
\noindent \emph{Proof:} As observed after the definition of edge merge, one can think of edge merge as the concatenation of two operations: $i)$ split $e$ in $\Asf_1$ and $\pi(e)$ in $\Asf_2$ each up into two edges with a new intermediate node ($g$ and $\pi(g)$, respectively) in between them, forming $\Asf_1'$ and $\Asf_2'$, respectively, followed by $ii)$ intermediate node merge of $\Asf_1'.g+\Asf_2'.\pi(g) $.  It is clear that $\Lmc'_0$ describes the operation that must happen to the rate region of $\Asf_i$, $i\in\{1,2\}$ to get the rate region of $\Asf_i'$, because the contents of the old edge $e$ or $\pi(e)$ must now be carried by both new edges after the introduction of the new intermediate node.  Applying Thm. \ref{thm:NodCom} to $\Asf'_1$ and $\Asf_2'$ yields (\ref{eq:EdgComeq1}).
\hfill $\blacksquare$
%\end{IEEEproof}

With Theorems \ref{thm:SrcIdn} -- \ref{thm:EdgCom}, one can easily derive the following corollary regarding the preservation of sufficiency of linear network codes and tightness of Shannon outer bound.
\begin{corollary} \label{cor:General}
Let network $\Asf$ be a combination of networks $\Asf_1,\Asf_2$ via one of the operations defined in \S\ref{sec:embedding}.  If $\Fbb_q$ vector (scalar) linear codes suffice or the Shannon outer bound is tight for both $\Asf_1,\Asf_2$, then the same will be true for $\Asf$. Equivalently, if $\Rmc_l(\Asf_i)=\Rmc_{\ast}(\Amc_i), i\in\{1,2\}$ for some $l\in\{o,q,(s,q)\}$ then also $\Rmc_l(\Asf) = \Rmc_{\ast}(\Asf)$.
\end{corollary}

Now we have defined both embedding and combination operators, and have demonstrated methods to obtain the rate regions of networks after applying the operators.  Next, we would like to discuss how to use them to do network analysis and solve large networks.

\section{Results with Operators}\label{sec:resultsoperators}
In this section, we demonstrate the use of the operators defined in \S\ref{sec:embedding} and \S\ref{sec:combination}.  We will first discuss the use of network embedding operations (\S\ref{sec:embedding}) to obtain the forbidden network minors and to predict code sufficiency for a larger network given the sufficiency is known for the smaller embedded network.  Then, we discuss the use of combination operations (\S\ref{sec:combination}) to solve large networks.  Finally, we discuss how to use both combination and embedding operations together to obtain even more  solvable networks.

\subsection{Use Embedding Operations to Obtain Network Forbidden Minors}

\begin{figure}
\centering
\captionsetup{justification=centering}
\includegraphics[scale=0.5]{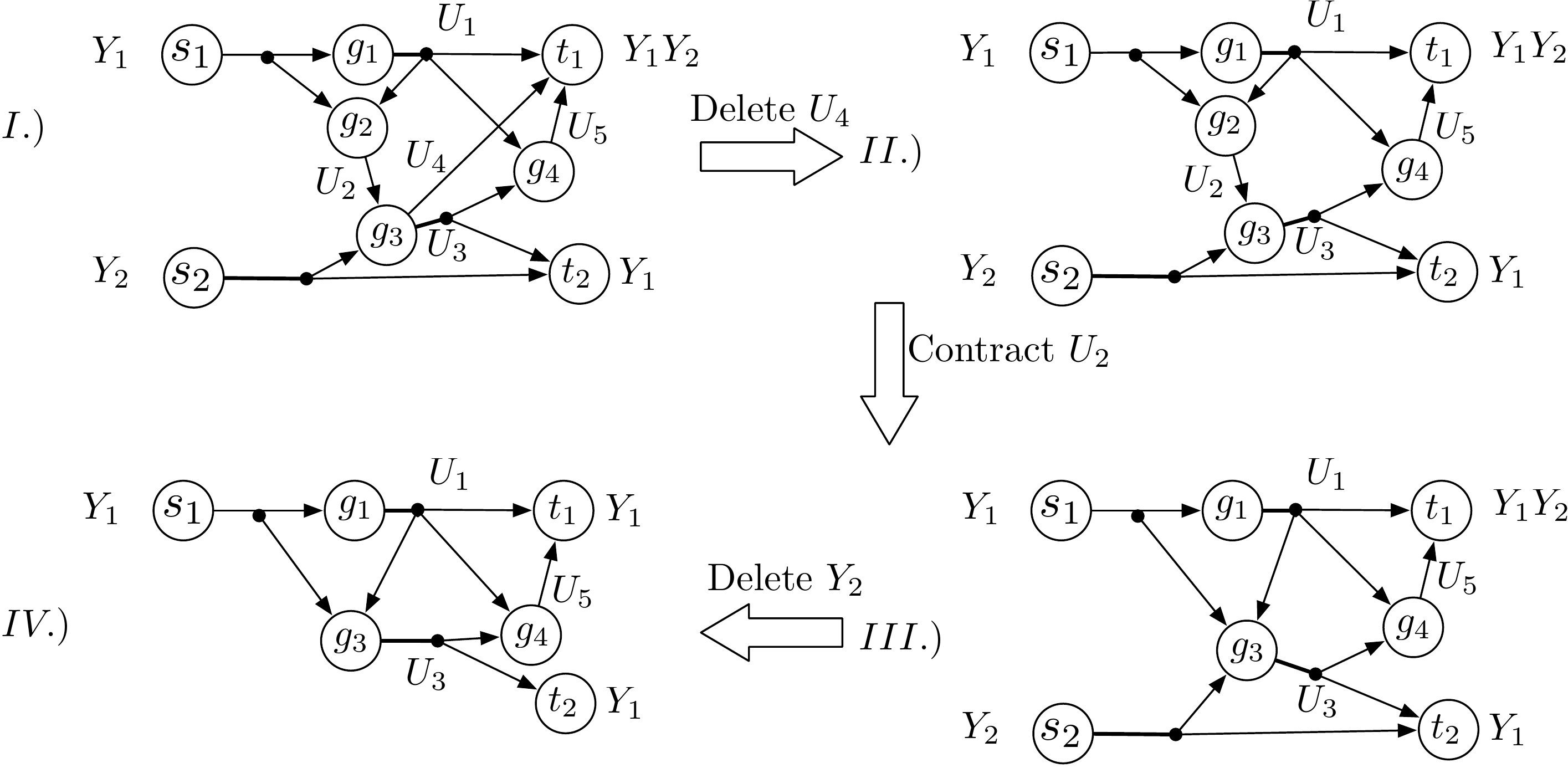}
\caption{\label{fig:embeddingexample} Using embedding operations to predict the insufficiency of scalar binary codes for a large network: since the large network $I$ and intermediate stage networks $II,III$ contain the small network $IV$ as a minor, the insufficiency of scalar binary codes for network $IV$, which is easier than network $I$ to see, predicts the same property for networks $III,II$ and $I$.}
\end{figure}

One natural use of embedding operations is to obtain rate regions for smaller embedded networks given the rate region for a larger network, as shown in \S\ref{sec:embedding}, since the rate regions of embedded networks are projections of the rate region of the larger network with some constraints. 

We can use the embedding operators in a reverse manner.  From Corollary \ref{cor:embedded}, we observe that if a class of linear codes suffice for a larger network, then it will suffice for the smaller networks embedded in it as well.  Equivalently, the insufficiency of a class of codes for a network is inherited by   the larger networks that have this network embedded inside.  This is similar to the forbidden minor property in matroid theory, which states that if a matroid is not $\Fbb_q$ representable, then neither are its extensions.  Therefore, if we know a small network has the property that a class of linear codes does not suffice, then we can predict that all networks containing it as a {\em network minor} through some embedding operations, will have the same property, without any calculations.  For instance, in Fig.\,\ref{fig:embeddingexample}, the small network $IV$ is a $(1,3)$ network for which scalar binary codes do not suffice, because when $H(Y_1)=2,H(U_1)=H(U_3)=H(U_5)=1$, there is no scalar solution (but there is a vector solution).  Since networks $I,II,III$ contain it as a minor (through the operations in the figure), we  know that scalar binary codes will not suffice for them, either.  This is verified by our computations in \S\ref{sec:resultssmall}.  

\begin{figure}
\centering
\includegraphics[scale=0.6]{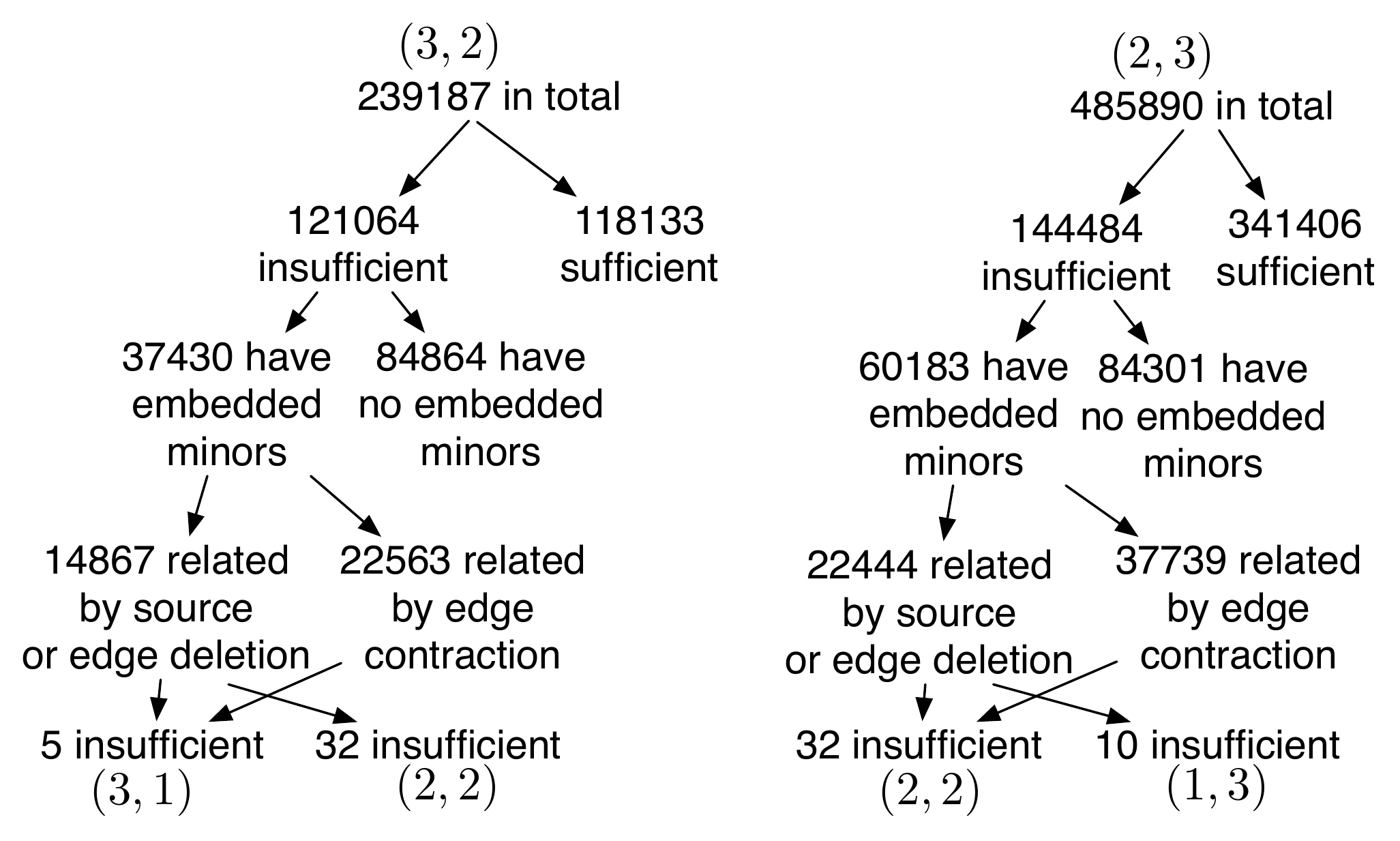}
\caption{Relations between networks of different sizes that scalar binary codes do not suffice.  The deletion operation considers both source and edge deletion, while the contraction operation only considers edge contraction.}\label{fig:forbiddenminor}
\end{figure}

As in matroid theory, we may have a collection of small networks that we know should be "forbidden" as a minor in larger networks, in the sense of ensuring the sufficiency of a class of linear codes.  For instance, in \cite{CongduanTranIT2014}, it was shown that for the thousands of MDCS networks for which scalar binary codes do not suffice, there are actually only 12 forbidden network minors.  For general hyperedge MSNC problems, we also built similar relationships as shown in Fig.\,\ref{fig:forbiddenminor}.  As Table \ref{tab:resultsgeneral} shows, the numbers of instances that scalar binary codes do not suffice for $(1,3), (3,1), (2,2), (3,2), (2,3)$ networks are $5,10,32,121064,144484$, respectively.  Those insufficient instances should not be a minor of any larger network that scalar binary codes suffice, and hence are forbidden minors.  To obtain a minimal list of network forbidden minors, we build a hierarchy among them.  As Fig.\,\ref{fig:forbiddenminor} shows, there are $37430$ out of $121064$ insufficient $(3,2)$ networks actually containing smaller forbidden network minors, of which $14867$ can be related by source or edge deletion and $22563$ can be related by edge contraction.  Similarly,  there are $37739$ out of $144484$ insufficient $(2,3)$ networks actually containing smaller forbidden network minors, of which $22444$ can be related by source or edge deletion and $37739$ can be related by edge contraction.  Those networks that do not have smaller forbidden network minors are new ones.

%Next, we will discuss the use of combination operations.
\vspace{-2mm}
\subsection{Use Combination Operations to Solve Large Networks}
As shown in \S\ref{sec:combination}, the rate region of a combined network can be directly obtained from the rate regions of the networks involved in the combination.  We show that the sufficiency of a class of linear codes are preserved after combination.  Actually, if we know one of the networks involved in the combination has the property that a class of linear codes does not suffice, the combined network will have the same property, since that small network is embedded in the combined network and it can be obtained by deleting the other networks.

%We could also see how many networks for a given size can be reached by combination of smaller networks.  For example, the $(3,3)$ networks can be reached by merging sources of one $(2,1)$ network and one $(3,2)$ network, or by merging intermediate nodes of one $(1,1)$ network and one $(2,3)$ network.  After considering all combinations, we see that the numbers of size $(2,3),(3,3)$ networks reachable from combinations of smaller networks are 172 and 10658, about $2.1\%$ and $13\%$ of the total number of instances, respectively.   If we allow the embedding operations  \cite{CongduanAllerton2014} to follow the combination operations, there will be more reachable networks.

Fig.\,\ref{fig:big} demonstrates the idea of solving large networks obtained by combination operations.  It shows that the rate region of the large combined network can be obtained directly from the rate regions of smaller network.  In addition,  sufficiency of linear codes is preserved.

\begin{figure}
\centering
\includegraphics[width=.95\textwidth]{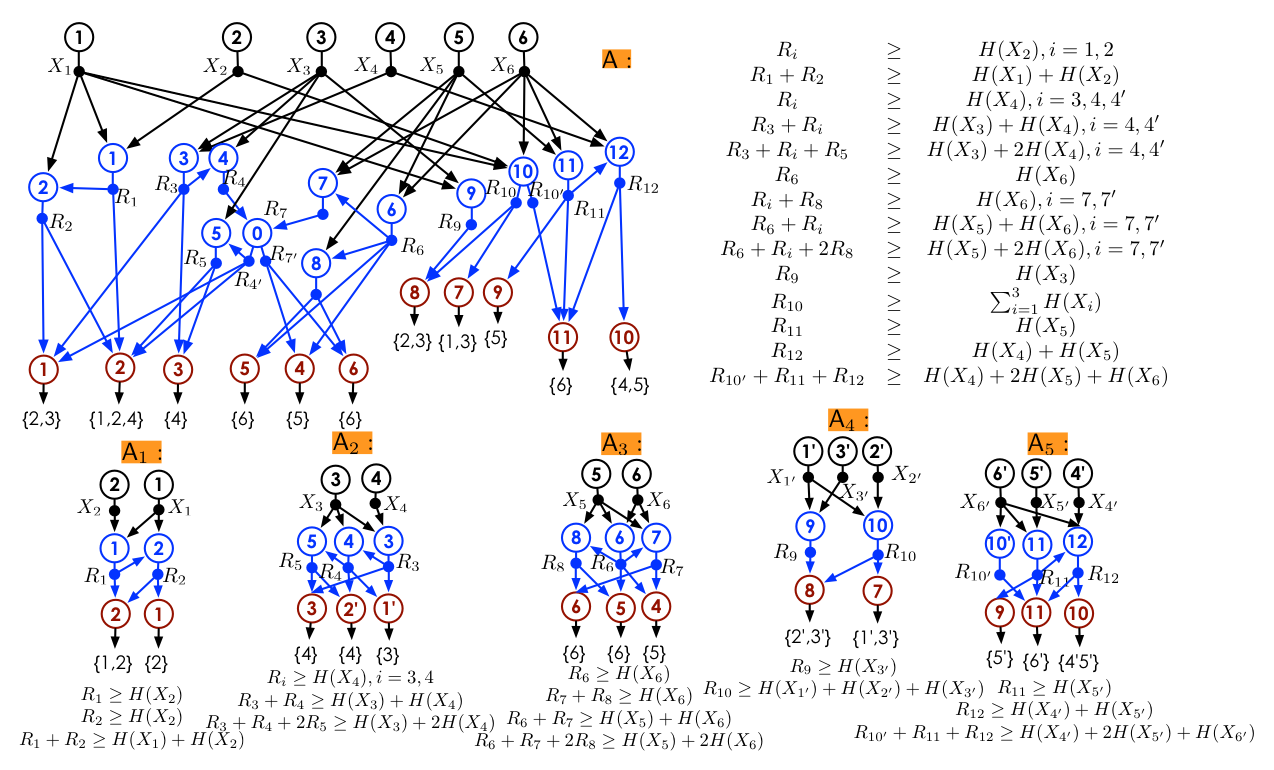}
\caption{A large network and its rate region created with the operations in this paper from the 5 networks below it.}\label{fig:big}
\end{figure}

\begin{example}\label{example1}
A $(6,15)$ network instance $\Asf$ can be obtained by combining five smaller networks $\Asf_1,\ldots,\Asf_5$, of which the representations are shown in Fig. \ref{fig:big}.  The combination process is I) $\Asf_{12}=\Asf_1.\{t_1,t_2\}+\Asf_2.\{t_{1'},t_{2'}\}$; II) $\Asf_{123}=\Asf_{12}.e_4+\Asf_3.e_7$ with extra node $g_0$ and edges $e_{4'},e_{7'}$; III) $\Asf_{45}=\Asf_4.g_{10}+\Asf_5.g_{10'}$; IV) $\Asf=\Asf_{123}.\{X_1,\ldots,X_6\}=\Asf_{45}.\{X_{1'},\ldots,X_{6'}\}$.  From the software calculations and analysis \cite{EntVecSoft, CongduanNetworkEnumerationfile}, 
one obtains the rate regions below the 5 small networks.
According to the theorems in \S\ref{sec:combination}, the rate region $\Rmc_*(\Asf)$ for $\Asf$ obtained from $\Rmc_*(\Asf_1),\ldots,\Rmc_*(\Asf_5)$, is depicted next to it.
Additionally, since calculations showed binary codes and the Shannon outer bound suffice for $\Asf_i,\ i\in\{1,\ldots,5\}$, Corollary \ref{cor:General} dictates the same for network $\Asf$.
\end{example}

\subsection{Use combination and embedding operations to generate solvable networks}

\begin{figure}
\centering
\includegraphics[width=.89\textwidth]{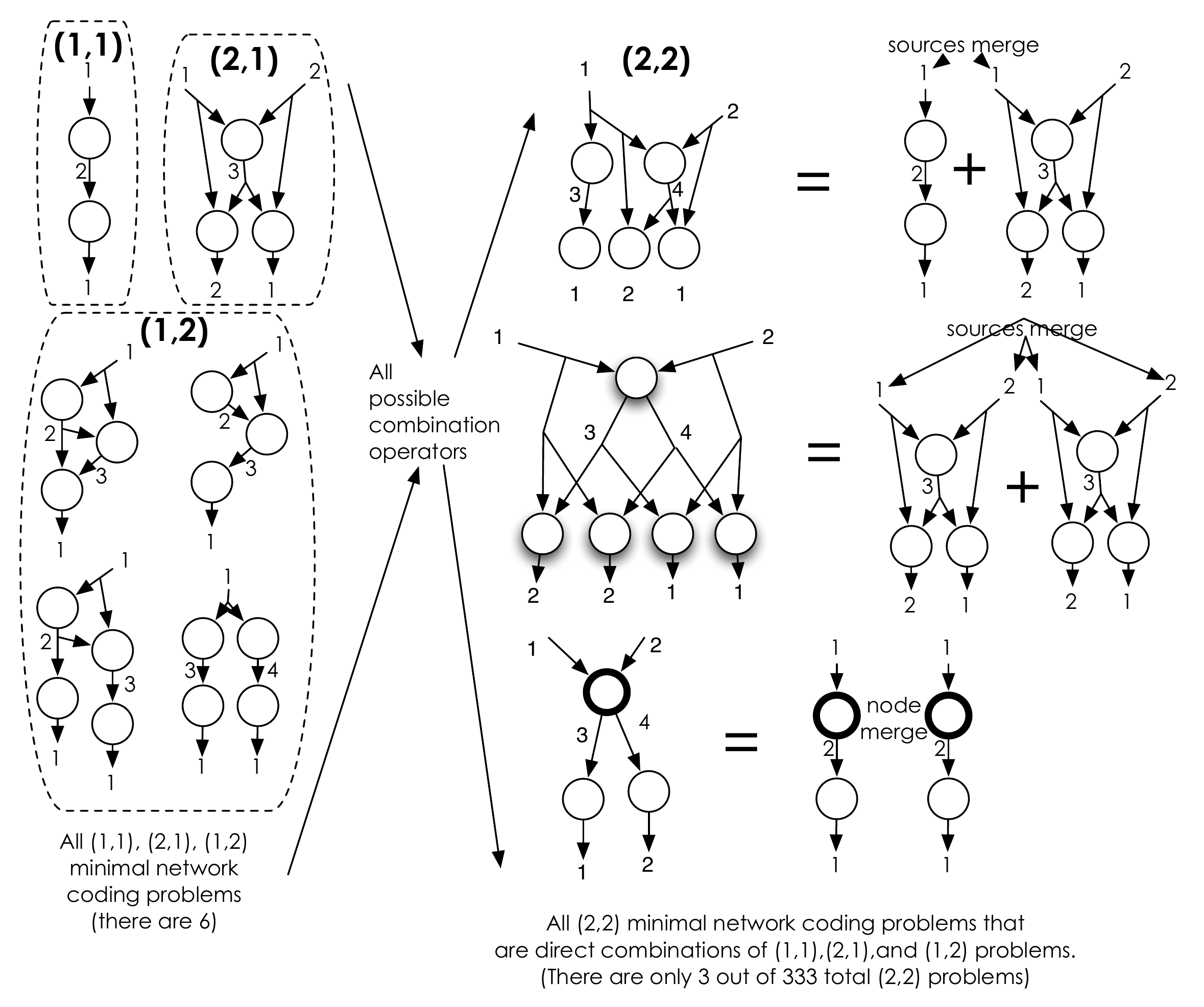}
\caption{There are a total of $3$ minimal $(2,2)$ network coding problems directly resulting from combinations of the $6$ small network coding problems with sizes $(1,1)$, $(1,2)$, and $(2,1)$.  However, as shown in Fig. \ref{fig:path}, by utilizing \emph{both} combinations \emph{and} embeddings operators, far more $(2,2)$ cases can be reached by iteratively combining and embedding the pool of networks starting from these 6 $(1,1)$, $(1,2)$, and $(2,1)$ networks via Algorithm \ref{alg:closure}.}\label{fig:all22combs}
\end{figure}

\begin{algorithm}
\SetAlgoLined
 \KwIn{Seed list of networks $seedList$, size limits on number of sources and edges} 
 \KwOut{All network instances generated by combination and/or embedding operations on the seed list}
 \BlankLine
 \textbf{Initialization:} network list for previous round $prevList=\emptyset$, new networks from previous round $prevAdd=seedList$, current list of networks $curList=\emptyset$, new networks generated in current round $curAdd=\emptyset$\;
\While{$size(prevAdd)>0$}{
 \For{every pair $\Imc\times\Jmc\in prevAdd\times prevAdd \cup prevAdd\times curList$}{
\If{prediction of network size after merge does not exceed size limits}{
 consider source, sink, node, edge merge on $\Imc,\Jmc$\;
 convert the new network to its canonical form $newNet$ \;
 \If{$newNet\notin curList$}{
 $curAdd=curAdd\cup newNet$\;
 } 
}
}
%\For{every pair $\Imc\times\Jmc\in prevAdd\times prevAdd$}{
%\If{prediction of network size after merge does not exceed size limits}{
%consider source, sink, node, edge merge on $\Imc,\Jmc$\;
% convert the new network to its canonical form $newNet$ \;
% \If{$newNet\notin curList$}{
% $curAdd=curAdd\cup newNet$\;
%}
%}
%}
\For{every  $\Imc\in prevAdd$}{
consider source deletion, edge deletion and edge contraction on $\Imc$\;
 convert the new network to its canonical form $newNet$ \;
 \If{$newNet\notin curList$}{
 $curAdd=curAdd\cup newNet$\;
}
}
$prevAdd=curAdd$\;
$prevList=curList$\;
$curList=curList\cup curAdd$\;
}
\caption{Generate all networks from a seed list of small networks using combination and embedding operations.}
\label{alg:closure}
\end{algorithm}

\begin{figure}
\centering
\includegraphics[width=.95\textwidth]{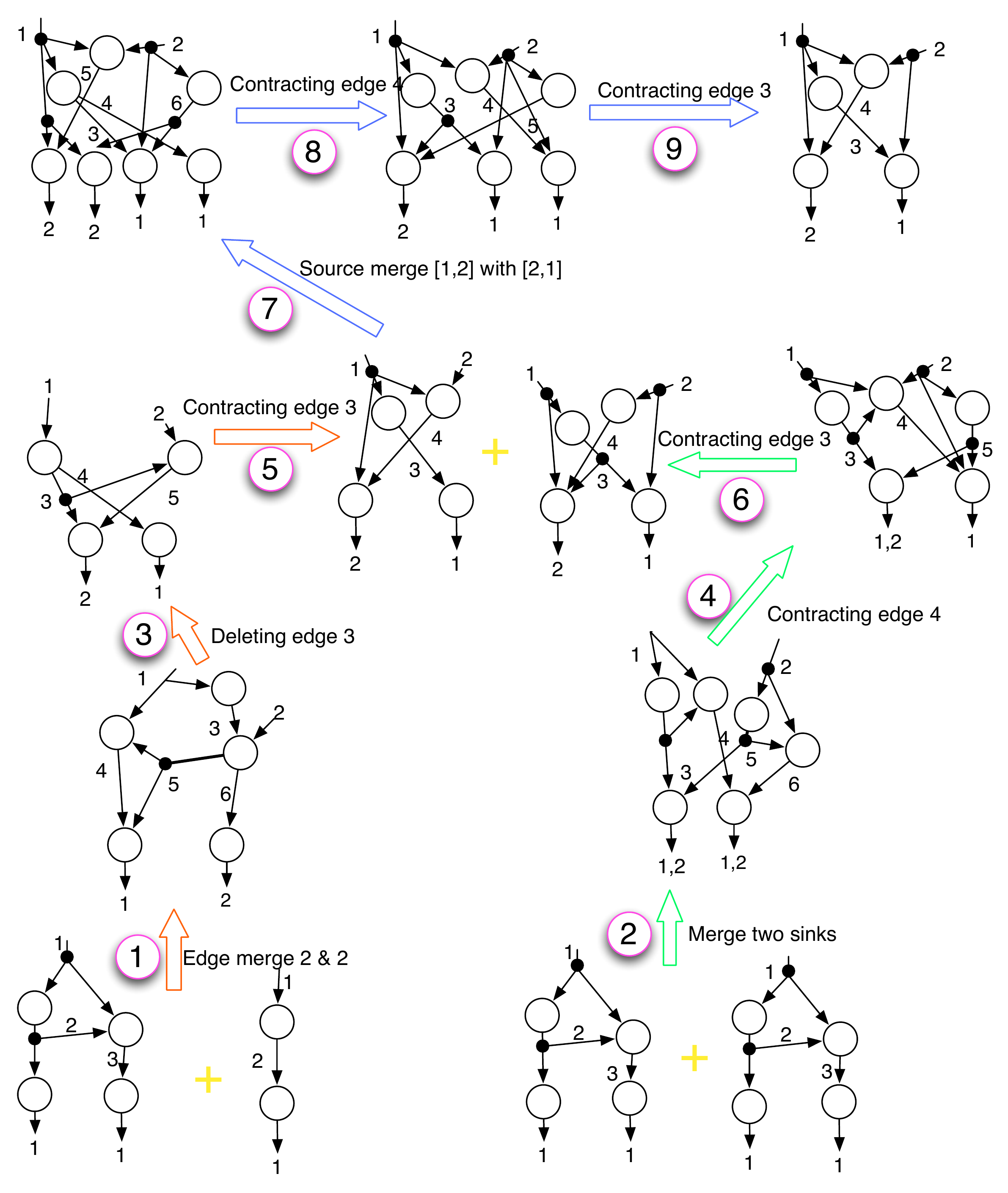}
\caption{The path of operations on a seed list of small networks to get a $(2,2)$ network that cannot be directly obtained by simple combination.  The size limits on networks involved in the operation process is $K\leq 3,L\leq 4$.}\label{fig:path}
\end{figure}

The combination operators provide a method for building large networks directly from smaller networks in such a way that the rate region of the large network can be directly obtained from those of the small networks.  The embedding operators provide a method for obtaining a small network from a large network in such a way that the rate region of the small network can be expressed in terms of the rate region of the large network.  These facts indicate that we can obtain networks by integrating combination and embedding operations.  If we start with a list of solved networks, all networks obtained in the following combination and/or embedding process will be solvable.  Note that one can apply these operations, especially the combination operations, an infinite number of times to obtain an infinite number of networks.  For demonstration purposes, we would like to limit the size of networks involved in the process.  As \S\ref{sec:combination} shows, it is not difficult to predict the worst case network size after combination.   Since the combined network may have redundancies, the worst case here means there is no redundancy after combination so that the network size is easy to predict.  For instance, after merging $k$ sources of a $(K,L)$ with $k$ sources of a $K',L'$ network, the network size after merging will be $K+K'-k+L+L'$ in the worst case.  We define the \emph{worst case partial closure of networks} as the networks obtained in the combination and embedding process such that no network involved exceeds the size limit.  An algorithm to generate the worst case partial closure of networks is shown in Algorithm \ref{alg:closure}.  As it shows, as long as the predicted network does not exceed the preset size limit, it will be counted as a new network and will be used as a seed, as long as it is not isomorphic to the existing ones.  We start with a seed list of networks ($seedList=prevAdd$), then through combinations and embeddings of these networks, a list ($curAdd$) of new networks ($newNet$) will be generated, which will, in turn, be used as seeds again.  The list of networks ($curList$) will be updated after each iteration.  The process stops when there is no new network that would not exceed the size limitations (size cap) found, in the sense of closure under these operations.  This tool is able to generate a large number of network from even small seed lists.  For instance, if we use as a seed list the single $(1,1)$ and single $(2,1)$, together with the four $(1,2)$ networks, and set the size limit for intermediate networks to $K\leq 4, L\leq 4$, there will be $11635$ new networks generated.  Even when embedding operations are not allowed, there will be $568$ new networks.  The details on the number of networks generated for different size limits is shown in Table \ref{tab:closuredata}.

\begin{table}
\caption{\label{tab:closuredata}The number of new canonical minimal network coding problems that can be generated from the 6 smallest canonical minimal network coding problems (the single $(1,1)$ network, the single $(2,1)$ network, and the four $(1,2)$ networks), by using combination operators (left), and both combination and embedding operators (right), in a partial closure operation where the largest network involved in a chain of operations never exceeds the ``cap'' (different columns).}
\centering
\includegraphics[width=0.7\textwidth]{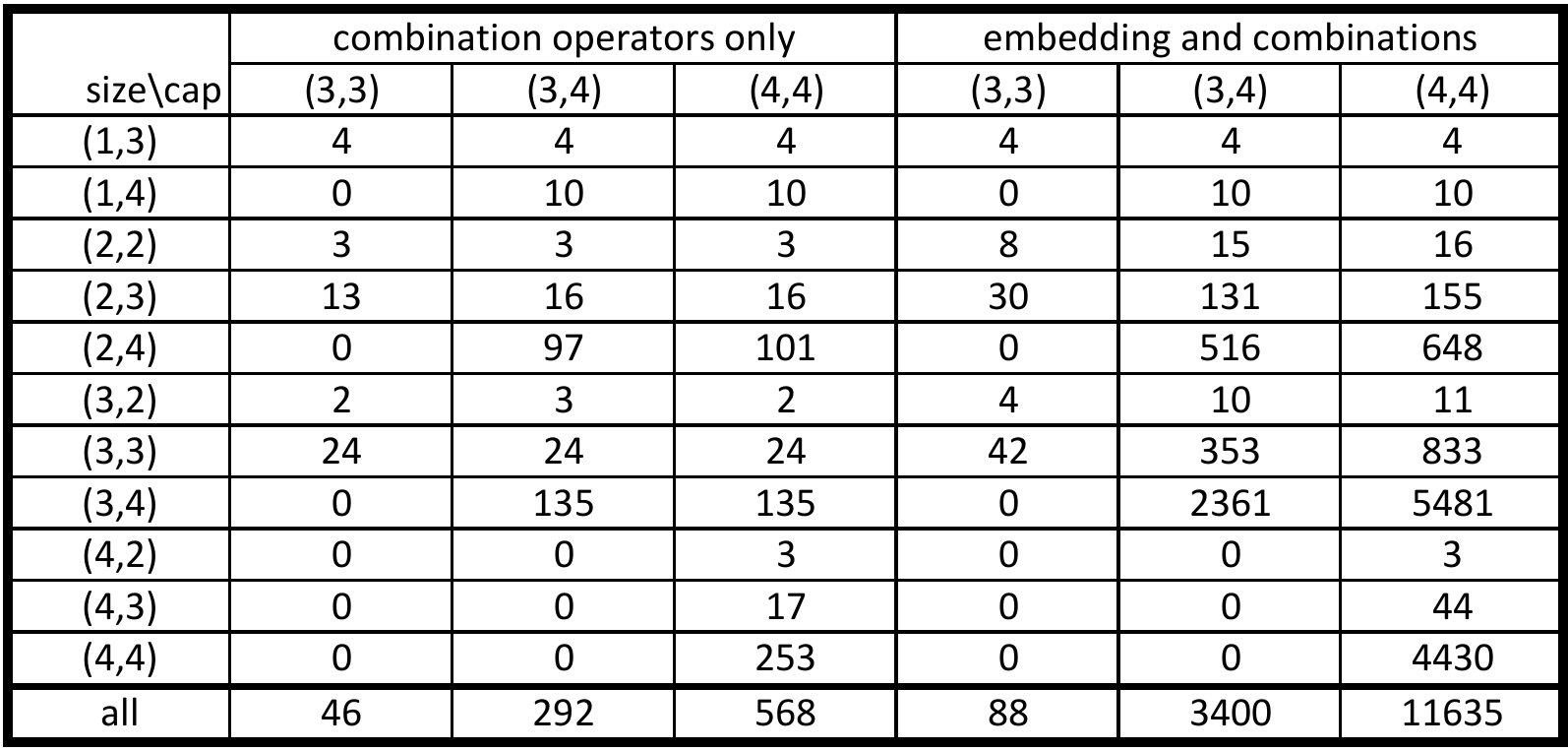}
\end{table}

From Table \ref{tab:closuredata} we first see that these six networks can generate a very large number of larger networks (see the bottom row of the table).   We also see that the number of networks generated grows, even for small target network sizes, rapidly as the cap on the largest network size is increased. 
Furthermore, it is important to note that, when trying to calculate the rate regions of larger networks from a list of rate regions for smaller networks, \emph{both combinations \underline{and} embeddings} are useful.  As Table \ref{tab:closuredata} shows, when no embedding operations are allowed in the generation process, the number of reachable networks are much less than those with embedding operations.   To further see a demonstration of this fact, consider all 6 canonical minimal network coding problems of dimensions $(1,1)$, $(1,2)$, $(2,1)$ as depicted in Fig. \ref{fig:all22combs}.  There are only 3 out of 333 networks with size $K=2,L=2$ can be reached by combination of smaller networks using only combination operations.  The three networks and their smaller components are shown in Fig. \ref{fig:all22combs}.  However, if we are allowed to use both combination and embedding operation together, we will reach more networks that are not reachable by merely combinations.  As Fig. \ref{fig:path} shows, though we still use the same network pool as in Fig. \ref{fig:all22combs}, another $(2,2)$ network is obtained by first combining smaller networks to a larger size and then using embedding operations to decrease the network size.  Several steps of combination and embedding operations are necessary to get this network, and Fig. \ref{fig:all22combs} shows the path through the operations from the initial seed list and the intermediately obtained network to reach it.  As Table \ref{tab:closuredata} shows, there are at least 12 $(2,2)$ networks which can be obtained in this manner.

%\vspace{-3mm}
\section{Conclusions and Future Work}\label{sec:conclusion}
%\vspace{-2mm}

This paper investigated the enumeration, rate region computation, and hierarchy of general multi-source multi-sink hyperedge networks.  The network model includes several special ones such as independent distributed storage systems and index coding problems.  
%A notion of equivalence of network coding problems under isomorphism was defined, and an algorithm was developed to directly list canonical representatives of the equivalence classes of networks under this isomorphism.  
This definition is further refined to a notion of minimal networks, containing no redundant sources, edges, or nodes, whose presence directly determines the rate regions from minimal networks.  Furthermore, since  networks related to one another through a permutation of the edge labels are derivable from one another, a notion of network equivalence or isomorphism under group action is defined.  By harnessing the Leiterspiel algorithm, which calculates the orbits of subsets incrementally in subset size, an efficient enumeration algorithm is presented for enumerating non-isomorphic networks, i.e., canonical representatives of the equivalence classes under this isomorphism, directly.  Using this algorithm, millions of non-isomorphic networks are obtained that represent trillions of network coding problems.  Then by applying computation tools, exact rate regions of most of them are obtained, leaving the rest with outer bound and simple code achieving inner bounds.  Only binary codes are considered here, and binary codes are shown to suffice for most of the networks under consideration.  
%If ternary codes, or codes in even larger field size, are considered, more and more exact rate regions will be obtained.  
In order to better understand and analyze the huge repository of rate regions, a notion of network hierarchy through embedding and combination operators is created.  These operations are defined in a manner such that the rate region of the network after each operation can be derived from the rate region of each of the networks involved in the operation.  
%The embedding operations include deletion of a source or edge, contraction of an edge.  The combination operations include merge of pairs of sources, sinks, edges, and intermediate nodes in two networks.  
The embedding operations enable us to obtain a list of forbidden network minors for the sufficiency of a class of linear codes.  It is shown that for many networks that scalar binary codes do not suffice, they contain a smaller network, for which scalar binary codes also do not suffice, as a minor under the embedding operations.  The combination operations enable us to solve large networks that can be obtained by combining some solvable networks.  The integration of both embedding and combination operations is able to generate rate regions for even more networks than can be solved directly with combination alone.  These operations open a door to many new avenues of network coding research.  Some of the pressing future problems for investigation include: I) assessing the coverage of the operators in the space of all problems;  II) if necessary, the creation of more powerful combination operations, such as node and edge merge, source and sink merge, etc;  III) a notion of forbidden minors which can harness both combination and embedding operators.  

\section*{Acknowledgment}
\noindent Support from NSF under CCF 1016588 \& 1421828 is gratefully acknowledged.
\bibliographystyle{IEEEtran}
\bibliography{CLbib}

\newpage
\begin{appendices}
\section{Complete proof of Theorem \ref{thm:rateregion}}\label{app:rateregionproof}
\label{app:thm1proof}
\subsection{Converse}

We need to prove that for any achievable rate tuple $\Rbf\in\Rmc_c(\Asf)$, we have $\Rbf\in \mathrm{Proj}_{\rbf,\boldsymbol{\omega}}(\overline{\text{con}(\Gamma_{N}^{*}\cap\mathcal{L}_{13})}\cap\mathcal{L}'_{4}\cap\mathcal{L}_{5})$.

Pick a point $\Rbf=[H(Y_1),\ldots,H(Y_K),R_1,\ldots,R_{L}]$ that is in $\Rmc_c(\Asf)$.  For convenience in comparing with the notations in \cite{YanYeungTranIT2012}, we let $\omega_s=H(Y_s)$ be the source rate to achieve.   Let an arbitrarily small $\epsilon>0$ be given. Since $\mathbf{R}$ is achievable, for all sufficiently
large $n$, there exists a block $n$ code such that
\begin{eqnarray}
\frac{\log\eta_{e}}{n}\leq R_{e}+\epsilon,\ e\in \Emc_U\label{eq:d1}\\
\omega_{s}\geq\tau_{s}\geq \omega_{s}-\epsilon,\ s\in \Smc\label{eq:d2}\\
p^{(n),err}\leq\epsilon,\label{eq:d3}
\end{eqnarray}
where $\eta_e$ is the index set of messages sent on edge $e$ and $\tau_s$ is the transmitted source rate at source $s$ under this block code.  For all $s\in\Smc$, we let $Y_s^{(n)}$ be the block source variable which takes value from the set $\{1,\ldots,\lceil \tau_s\rceil\}$.  For all $e\in\Emc_U$, we let $U_e^{(n)}$ be the codeword sent on edge $e$ and denote the alphabet of $U_e^{(n)}$ as $\Umc_e=\{0,1,\ldots,\eta_e-1\}$.

From (\ref{eq:d1}), we know for all $e\in \Emc_U$,
\begin{IEEEeqnarray}{rCl}
H(U_{e}^{(n)}) & \leq & \log |\mathcal{U} _{e}| \nonumber \\
& = & \log (\eta _{e} ) \nonumber \\
& \leq & n(R_{e} + \epsilon). \label{eq:convrate}
\end{IEEEeqnarray}

For all source $s\in\Smc$, from (\ref{eq:d2}) we have
\begin{equation}
n\omega_{s}\geq H(Y_s^{(n)})=\log\left\lceil 2^{n\tau_{s}}\right\rceil \geq n(\omega_{s}-\epsilon).
\end{equation}

Since it is assumed that $\Rbf$ is achievable, there exist random variables $Y_s^{(n)},U_e^{(n)}$ such that
\begin{eqnarray}
H(\Ybf_{\Smc}^{(n)}) & = & \sum_{s\in\Smc}H(Y_{s}^{(n)})\label{eq:indepSources1}\\
H(U_{e}^{(n)}|\Ubf_{{\rm In}({\rm Tl}(e))}^{(n)}) & = & 0,e\in \Emc_U\label{eq:EncFunc1}
%h_{Y_{k}} & \geq & H(X_{k}),k\in \{1,\ldots,K\}\label{eq:sourceEnt}\\
%h_{\Xbf_{1:k}|\mathcal{U}_{d}} & = & 0,\ D_{d}\in{\rm Fan}(D_{d})\label{eq:DecFunc1}
%R_{e} & \geq & h_{U_{e}},e\in\Emc\label{eq:RateCons}
\end{eqnarray}

Following Lemma 1 in \cite{YanYeungTranIT2012} (note that the input of a sink $t$ may include some other sources that not in $\beta(t)$), which applies the Fano's inequality, we have
\begin{equation}
H(\Ybf_{\beta(t)}^{(n)}|\Ubf_{{\rm In}(t)}^{(n)})\leq n\phi_t (n,\epsilon), \label{eq:convdec}
\end{equation}
where $\phi_{t}(n,\epsilon)$ has the following properties:
\begin{enumerate}
\item $\phi_{t}(n,\epsilon)$ is bounded;
\item $\phi_{t}(n,\epsilon)\rightarrow0$ as $n\rightarrow \infty$ and
$\epsilon\rightarrow0$;
\item $\phi_{t}(n,\epsilon)$ is monotonically decreasing with increase of $n$ and
decrease of $\epsilon$.
\end{enumerate}

From \eqref{eq:convrate} -- \eqref{eq:convdec} we get the existence of entropic vector such that
\begin{eqnarray}
h_{\Ybf_{\Smc}} & = & \sum_{s\in\Smc}h_{Y_{s}}\label{eq:indepSources2}\\
h_{U_{e}|\Ubf_{{\rm In}({\rm Tl}(e))}} & = & 0,e\in \Emc_U\label{eq:EncFunc2}\\
h_{Y_{s}} & \geq & n(\omega_s-\epsilon),s\in \Smc\label{eq:sourceEnt2}\\
h_{\Ybf_{\beta(t)}|\Ubf_{{\rm In}(t)}} & \leq& n\phi_t(n,\epsilon),\, \forall t\in \Tmc\label{eq:DecFunc2}\\
h_{U_e} & \leq & n(R_e+\epsilon),e\in\Emc\label{eq:RateCons2}
\end{eqnarray}

Now define the following two regions in $\Rbb^{2^N-1}$
that depend on $n$, and view the $[R_e | e\in\Emc_U]$ as fixed vector:
\begin{equation}
\mathcal{L}_{4,\epsilon}^{n}=\{\mathbf{h}\in\Rbb^{2^N-1}:h_{U_e}\leq n(R_e+\epsilon),e\in\Emc_U\},
\end{equation}
\begin{equation}
\mathcal{L}_{5,\epsilon}^{n}=\{\mathbf{h}\in\Rbb^{2^N-1}:h_{\Ybf_{\beta(t)}|\Ubf_{{\rm In}(t)}}\leq n\phi_{t}(n,\epsilon),\forall t\in \Tmc\}.
\end{equation}

Then from \eqref{eq:indepSources2}-- \eqref{eq:RateCons2} we see that there exists
\begin{equation}
\mathbf{h}\in\Gamma_N^*\label{eq:h_tau}
\end{equation}
such that
\begin{equation}
\mathbf{h}\in\mathcal{L}_{13}\cap\mathcal{L}_{4,\epsilon}^{n}\cap\mathcal{L}_{5,\epsilon}^{n}\label{eq:h_cons}
\end{equation}
and $\forall s\in \Smc$
\begin{equation}
h_{Y_{s}}\geq n(\omega_{s}-\epsilon),\label{eq:hxe}
\end{equation}
i.e.
\begin{equation}
\frac{h_{Y_{s}}}{n}\geq \omega_{s}-\epsilon.
\end{equation}
From (\ref{eq:h_tau}) and (\ref{eq:h_cons}) we obtain
\begin{equation}
\mathbf{h}\in\Gamma_N^*\cap\mathcal{L}_{13}\cap\mathcal{L}_{4,\epsilon}^{n}\cap\mathcal{L}_{5,\epsilon}^{n}.
\end{equation}
Since $\Gamma_N^*\cap\mathcal{L}_{13}$ contains
the origin, we get that
\begin{equation}
n^{-1}\mathbf{h}\in\overline{\text{con}(\Gamma_N^*\cap\mathcal{L}_{13})}\cap\mathcal{L}_{4,\epsilon}\cap\mathcal{L}_{5,\epsilon},\label{eq:h/n}
\end{equation}
where
\begin{equation}
\mathcal{L}_{4,\epsilon}=\{\mathbf{h}\in\Rbb^{2^N-1}:h_{U_e}\leq R_e+\epsilon\}
\end{equation}
and
\begin{equation}
\mathcal{L}_{5,\epsilon}=\{\mathbf{h}\in\Rbb^{2^N-1}:h_{\beta(t)|\Ubf_{{\rm In}(t)}}\leq \phi_{t}(n,\epsilon),t\in \Dmc\}.
\end{equation}

For all $n$ and $\epsilon$, define the set
\begin{equation}
\Bmc^{(n,\epsilon)}= 
\{\mathbf{h}\in\overline{\text{con}(\Gamma_N^*\cap\mathcal{L}_{13})}\cap\mathcal{L}_{4,\epsilon}\cap\mathcal{L}_{5,\epsilon}:
\omega_s\geq h_{Y_{s}}\geq \omega_s-\epsilon,\forall s\in \Smc\}.
\end{equation}

The fact that $\Bmc^{(n,\epsilon)}$ is closed and bounded follows immediately
from a similar proof in \cite{YanYeungTranIT2012}.  Then, we conclude that $\Bmc^{(n,\epsilon)}$ is compact.

Now from the fact that $\phi_{t}(n,\epsilon)$ is monotonically
decreasing in both $n$ and $\epsilon$, so for all $\epsilon'<\epsilon$
and $n$,
\begin{equation}
\Bmc^{(n+1,\epsilon)}\subset\Bmc^{(n,\epsilon)}
\end{equation}
and
\begin{equation}
\Bmc^{(n,\epsilon')}\subset\Bmc^{(n,\epsilon)}
\end{equation}

Note that from (\ref{eq:hxe}) and (\ref{eq:h/n}) we see that for
any $\epsilon>0$, $\Bmc^{(n,\epsilon)}$ is nonempty. Since
$\Bmc^{(n,\epsilon)}$ is compact and nonempty,
\begin{equation}
\lim_{\epsilon\rightarrow0}\lim_{n\rightarrow\infty}\Bmc^{(n,\epsilon)}=\bigcap_{\epsilon}^{0}\bigcap_{n=1}^{\infty}\Bmc^{(n,\epsilon)}
\end{equation}
is also nonempty and compact, which equals to
\begin{equation}
\{\mathbf{h}\in\overline{\text{con}(\Gamma_N^*\cap\mathcal{L}_{13})}\cap\mathcal{L}_{4}\cap\mathcal{L}_{5}:h_{Y_{s}}=\omega_s,\forall s\in \Smc\},
\end{equation}
where $\Lmc_4=\{\mathbf{h}\in\Rbb^{2^N-1}:h_{U_e}\leq R_e\}$ as defined in \cite{YanYeungTranIT2012} with $[R_e | e\in\Emc_U]$ as constants.

Hence, if we let the $R_e,e\in\Emc_U$ be unconstrained variables, we conclude that
\begin{equation}
\mathbf{R}\in\mathrm{Proj}_{\boldsymbol{r},\boldsymbol{\omega}}(\overline{\text{con}(\Gamma_N^*\cap\mathcal{L}_{13})}\cap\mathcal{L}_{4'}\cap\mathcal{L}_{5}),
\end{equation}
where $\boldsymbol{r} = \left( R_e, e\in\Emc_U\right)$ and $\boldsymbol{\omega} = \left(H(Y_s),s\in\Smc\right)$.
\subsection{Achievability}

We need to prove that for any point $\mathbf{R}\in\mathrm{Proj}_{\boldsymbol{r},\boldsymbol{\omega}}(\overline{\text{con}(\Gamma_N^*\cap\mathcal{L}_{13})}\cap\mathcal{L}_{4'}\cap\mathcal{L}_{5})$, there exists a code such that this
rate is achievable. 

Similar as Lemma 3 in \cite{YanYeungTranIT2012}, if we define
\begin{equation}
\Amc_1=\overline{{\rm con}(\Gamma_N^*\cap \Lmc_{13})}
\end{equation}
and
\begin{equation}
\Amc_2=\overline{D(\Gamma_N^*\cap \Lmc_{13})},
\end{equation}
where $D(\Amc')=\{\alpha \hbf: \hbf\in \Amc' \ \&\  0\leq \alpha\leq 1\}$, we have 
\begin{equation}
\Amc_1=\Amc_2. \label{eq:lemma3}
\end{equation}

The proof of \eqref{eq:lemma3} is identical to the proof of Lemma 3 in \cite{YanYeungTranIT2012} except that we only have $\Lmc_{13}$ but \cite{YanYeungTranIT2012} has $\Lmc_{123}$.

Let $\mathbf{R}$ be the point picked, we also have $\mathbf{R}\in\mathrm{Proj}_{\boldsymbol{r},\boldsymbol{\omega}}(\overline{D(\Gamma_N^*\cap\mathcal{L}_{13})}\cap\mathcal{L}_{4'}\cap\mathcal{L}_{5})$, 
then there exists an $\mathbf{h}\in\overline{D(\Gamma_N^*\cap\mathcal{L}_{13})}\cap\mathcal{L}_{4'}\cap\mathcal{L}_{5}$
such that $\mathbf{R}=\mathrm{Proj}_{\boldsymbol{r},\boldsymbol{\omega}}(\mathbf{h})$. Furthermore, there exists an entropic vector $\hat{\hbf}\in \Gamma_N^*\cap \Lmc_{13}$ and an $\alpha$ such that $\hbf=\alpha \hat{\hbf}$. 

Since $\hat{\hbf}\in \Gamma_N^*\cap \Lmc_{13}$, there exists a collection of random variables
$\mathcal{N}=\left\{Y_{s},U_{e}| s\in \Smc,e\in \Emc_U\right\}$ such that
\begin{IEEEeqnarray}{rcl}
\alpha \hat{\hbf}_{Y_s} &=& \omega_s,s\in \Smc\\
\hat{\hbf}_{\Ybf_{\Smc}}&=&\sum_{s\in\Smc}\hat{\hbf}_{Y_s}\\
\hat{\hbf}_{U_e|\Ubf_{{\rm In}({\rm Tl}(e))}}&=&0,e\in \Emc
\end{IEEEeqnarray}
where $\omega_s=H(Y_s)$ is the source rate to achieve at source $s$.

Furthermore, since $\hbf\in \Lmc_4'\cap\Lmc_5$ and we only need to show $\hbf$ is asymptotically achievable, we have
\begin{eqnarray}
\alpha \hat{\hbf}_{U_e}\leq R_e+\mu,e\in \Emc\\
\alpha \hat{\hbf}_{\Ybf_{\beta(t)}|\Ubf_{{\rm In}(t)}}\leq \gamma,t\in \Tmc
\end{eqnarray}
where $\mu$ and $\gamma$ are arbitrarily small positive numbers.

Our next step is to show that the rate vector $\Rbf'=[\omega_s,R_e+\mu| s\in\Smc, e\in \Emc]$ is achievable. Then as $\gamma, \mu \rightarrow 0$, we know that $\Rbf$ is achievable. 

Use the similar code construction and performance analysis as in the proof of achievability in \cite{YanYeungTranIT2012}, we can show that $[\omega_s|s\in \Smc]$ is achievable if we set $[R_e+\mu| e\in\Emc]$ as the capacities on edges.  Encoding functions at sources are simply identity functions.  Equivalently, the similar construction shows that $\Rbf'$ is achievable. Therefore, $\Rbf$ is achievable as $\gamma, \mu\rightarrow 0$.

\end{appendices}

\end{document}